\documentclass[a4paper,12pt]{book}
\usepackage{ucs}
\usepackage[utf8x]{inputenc} 

\usepackage{multicol} 
\usepackage{mathrsfs} 
\usepackage{amsmath, amsthm, amssymb,bm}

\RequirePackage[T1]{fontenc}
\usepackage{amsfonts}

\usepackage[
  breaklinks,
  naturalnames,
  ]{hyperref}

\usepackage{amsrefs}

\makeatletter
\renewcommand\PrintNames@a[4]{%
    \PrintSeries{\name}
        {#1}
        {}{ y \set@othername}
        {,}{ \set@othername}
        {,}{ y \set@othername}
        {#2}{#4}{#3}%
}

\BibSpec{article}{%
    +{}  {\PrintAuthors}                {author}
    +{,} { \textit}                     {title}
    +{.} { }                            {part}
    +{:} { \textit}                     {subtitle}
    +{,} { \PrintContributions}         {contribution}
    +{.} { \PrintPartials}              {partial}
    +{,} { }                            {journal}
    +{}  { \textbf}                     {volume}
    +{}  { \PrintDatePV}                {date}
    +{,} { \issuetext}                  {number}
    +{,} { \eprintpages}                {pages}
    +{,} { }                            {status}
    +{,} { \PrintDOI}                   {doi}
    +{,} { available at \eprint}        {eprint}
    +{}  { \parenthesize}               {language}
    +{}  { \PrintTranslation}           {translation}
    +{;} { \PrintReprint}               {reprint}
    +{.} { }                            {note}
    +{.} {}                             {transition}
    +{}  {\SentenceSpace \PrintReviews} {review}
}
\makeatother

\usepackage[all]{xy}

\usepackage{braket} 
\usepackage{enumitem}

\newtheorem{thm}{Theorem}[section]
\newtheorem*{thm*}{Theorem}
\newtheorem{coro}[thm]{Corollary}

\newtheorem{lemma}[thm]{Lemma}
\newtheorem{prop}[thm]{Proposition}
\theoremstyle{plain}

\newtheorem{definicion}[thm]{Definition}
\theoremstyle{definition}

\newtheorem{ej}[thm]{Example}
\newtheorem{rem}[thm]{Remark}
\theoremstyle{remark}

\usepackage[pdftex]{graphicx}
\usepackage{capt-of}

\usepackage[usenames,dvipsnames,svgnames,table]{xcolor}

\usepackage{a4,color,palatino,fancyhdr}

\usepackage{longtable}
\usepackage[acronym,xindy,toc]{glossaries}
\newacronym{qft}{QFT}{quantum field theory}
\newacronym{pqft}{pQFT}{perturbative quantum field theory}
\newacronym{tvs}{TVS}{topological vector space}
\newacronym{lcs}{LCS}{locally convex space}
\newglossaryentry{angelsperarea}{
  name = $a$ ,
  description = The number of angels per unit area,
}
\newglossaryentry{infdiff}{name=space of infinitely differentiable functions,symbol={\[\mathscr{E}\left( \Omega\right)  \]},
description=blah}
\makeglossaries

\usepackage{hyperref}
\hypersetup{
    colorlinks=true,
    linkcolor=blue,
    filecolor=magenta,      
    urlcolor=blue,
    citecolor=violet,
}
\usepackage{makeidx}
\makeindex

\begin{document}

\thispagestyle{empty}

\begin {center}

\medskip
\textbf{UNIVERSIDAD DE BUENOS AIRES}

\smallskip

\textbf{Facultad de Ciencias Exactas y Naturales}

\smallskip

\textbf{Departamento de Matem\'atica}

\vspace{3.5cm}

\textbf{\large Tesis de Licenciatura}

\vspace{1.5cm}

\textbf{\large On the renormalization group of quantum field theory \\
\vspace{0.3cm}
(after R. Bocherds)}

\vspace{1.5cm}

\textbf{Gabriel Alexis Palau}

\end {center}

\vspace{1.5cm}

\noindent \textbf{Director: Estanislao Benito Herscovich Ramoneda} \

\vspace{3cm}


\chapter*{Agradecimientos}
 Me gustaría empezar esta tesis agradeciendo a todas las personas que de una forma u otra me incentivaron, ayudaron y/o soportaron durante el estudio de esta carrera. Empezando por mi padre Hector Palau, quien en mi infancia siempre que podía me daba algún acertijo sobre la naturaleza, por ejemplo contándome que es posible que las estrellas que uno ve en la noche ya no existan, o enseñándome a calcular la distancia a la que una tormenta se encuentra a partir del sonido de sus truenos. 
También quiero agradecer a mi madre, Isabel do Eyo quien desde siempre me apoyó e incentivó en este camino. A mis hermanas Cintia y Flavia, por aguantar tanta desatención en ciertos períodos de la carrera, y a mi novia Elim Garcia por lo mismo sobre todo en la etapa final de la misma y por tanta ayuda y apoyo.
 Agradezco también a mi abuelo Eduardo por sus charlas y a mi abuela Emilia por el apoyo.


A mis compañeros y amigos por los ratos compartidos Melanie Bondorevsky, Virginia Gali y los "eje eques" y todos aquellos con quienes compartí momentos gratos.    

Académicamente le agradezco infinitamente a mi director de tesis Estanislao Herscovich quien estuvo tan presente durante la tesis y me ayudó tanto que nadie adivinaría que nos separaba un océano, por la enorme paciencia y la gran enseñanza que me dejó. También por brindarme el placer de tener con él una que otra charla sobre física.

Agradezco también al jurado de esta tesis Gabriel Larotonda y Mauricio Leston por leer y juzgar este trabajo.

\chapter*{Introduction}
\par The aim of this thesis is to analyze part of the results presented by R.E. Borcherds in his article on renormalization and quantum field theories \cite{bor}. To facilitate this, the first chapter presents the basic facts on category theory, sheaf theory and the theory of distributions on a manifold. In the particular case of distributions on manifolds we differentiate between the H-distribution (introduced by Hörmander \cite{hor}) and the D-distributions (described in the book of Dieudonn{\'e} \cite{di}) and extend the product of distributions defined by Hörmander only for H-distibution, to an action from H-distributions over D-distributions. 
 Each one of these facts will be used along the thesis in order to exhibit more detailed descriptions and explanations for the results. 
\par In Chapter 2 the concepts of spacetime, fields, Lagrangians, propagators and Feynman measure are defined. The notion of Feynman measure is introduced with a more detailed explanation of the Gaussian condition and its relation with the cut propagator. 
\par In the last chapter we define the renormalization group and immediately characterize it using the results from Chapter 1. The algebraic structure of this group is described in detail. We present a descending filtration of the renormalization group by normal subgroups, compute the corresponding quotients and commutators, and show that any element of the renormalization group can be written as an infinite product of elements each one living in one of the aforementioned subgroups. Lastly we show that an action from the renormalization group over the set of Feynman measures is well defined and prove that this action is transitive over the Feynman measures associated with a given cut local propagator.


\tableofcontents
\pagebreak

\chapter{Preliminary results} 

\par Along this thesis we will need certain level of knowledge in category theory, sheaf theory and the theory of distributions on manifolds. Since we shall just use some specific result in each of these areas, we shall present them. The reader familiarized with these concepts can omit this chapter and eventually come back next to regard a specific result if needed.

\section{Results on category theory}

\subsection{Basic notions}

\par For the basic definitions on category theory, we refer the reader to \cite{Mariano}, Ch.9, or \cite{chris}, Ch.XI.1. 
Moreover, the definitions of monoidal (or tensor) category, monoidal (or tensor) functor and monoidal natural transformation can be found in \cite{chris}, Ch.XI.2. and XI.4.

\par We shall use the letters $\mathcal{C}$, $\mathcal{D}$,... to denote a monoidal category. 
If $\mathcal{C}$ is a monoidal category we will denote by  $\otimes_{\mathcal{C}}:\mathcal{C}\times\mathcal{C}\rightarrow\mathcal{C}$ the \emph{tensor product} bifunctor and by $\textbf{1}_{\mathcal{C}}$ the unit object.
The left and right unit morphisms will be denoted by $l_{M}:\textbf{1}_{\mathcal{C}}\otimes_{\mathcal{C}}M\to M$ and $r_{M}:M\otimes_{\mathcal{C}}\textbf{1}_{\mathcal{C}}\to M$ for any $M\in Obj(\mathcal{C})$; whereas the associativity morphism will be denoted by $a_{A,B,C}:(A\otimes B)\otimes C \to A\otimes (B\otimes C)$ for any $A,B$ and $C$ objects of $\mathcal{C}$.

\par  A strict monoidal category is a category where the associativity and unit morphism are the identities of the category, i.e. for all $X$, $Y$ and $Z$ objects of $\mathcal{C}$ we have $(X\otimes Y)\otimes Z=X\otimes(Y\otimes Z)$ and $\textbf{1}_{\mathcal{C}}\otimes X=X=X\otimes\textbf{1}_{\mathcal{C}}$.

\par We present a well-known result about monoidal categories, whose proof can be found in \cite{chris} Theorem XI.5.3  

\begin{thm} Any monoidal category is monoidal (or tensor) equivalent to a strict one.
\end{thm}

\par This theorem implies Mac Lane´s coherence theorem which states that in a monoidal category all diagrams built with the unit, associativity or identity morphisms commute. 
\par  Hence from now on we will suppose for simplicity in the exposition that all our monoidal categories are strict.

\begin{definicion}
 Given a monoidal category $\mathcal{C}$, a triplet $(A,m,u)$ with $A\in Obj(\mathcal{C})$, $m\in Hom_{\mathcal{C}}(A\otimes A,A)$ and $u\in Hom_{\mathcal{C}}(\textbf{1},A)$ is called an \emph{unitary algebra in $\mathcal{C}$} if we have:
\begin{eqnarray*}
& m\circ(m\otimes Id_{A})=m \circ(Id_{A}\otimes m),\\
& m\circ(u\otimes Id_{A})=Id_{A}, \\
& m\circ(Id_{A}\otimes u)=Id_{A}. 
\end{eqnarray*}
\end{definicion}

\begin{rem}
$\textbf{1}_{\mathcal{C}}$ with the product $l_{\textbf{1}_{\mathcal{C}}}=r_{\textbf{1}_{\mathcal{C}}}$ (see \cite{chris}, Lemma XI.2.3) and unit $Id_{\textbf{1}_{\mathcal{C}}}$ is a unitary algebra. 
\end{rem}
\par Notice that the natural domain of $m\otimes Id_{A}$ is $(A\otimes A)\otimes A$ and for 
$Id_{A}\otimes m$ it is $A\otimes (A\otimes A)$, so the only reason why the first equation of the definition makes sense is because in a strict monoidal category the two domains are equal. A similar observation applies to the second and third equations: if we were not in a strict monoidal category the first of the equations above should take the form $m\circ(m\otimes Id_{A})=m\circ(Id_{A}\otimes m)\circ a_{A,A,A}$, and the same applies to the rest. 

\par The morphism between two unitary algebras $(A,m_{A},u_{A})$ and $(B,m_{B},u_{B})$ in $\mathcal{C}$ is a morphism $f\in Hom_{\mathcal{C}}(A,B)$ such that
 \begin{align*}
 u_{B}\circ f=u_{A} & & \text{and} & & m_{B}\circ(f\otimes f)=f\otimes m_{A} .
 \end{align*}

\par The class of unitary algebras in $\mathcal{C}$ together with the previous morphisms form a category which we note by $Alg_{\mathcal{C}}$.  
Since we will be interested in unitary algebras, we will drop the adjective unitary and refer to them as algebras.

\begin{definicion}
Let $\mathcal{C}$ be a monoidal category and $(A,m,u)$ a unitary algebra in $\mathcal{C}$. A \emph{right module over $A$} is a pair $(M,\rho)$ such that $M\in Obj(\mathcal{C})$ and $\rho \in Hom_{\mathcal{C}}(M\otimes A,M) $ satisfying the following equations
\begin{eqnarray*}
& \rho \circ (Id_{M}\otimes m)=m\circ (\rho \otimes Id_{A}), \\
& \rho \circ (Id_{M}\otimes u)=Id_{M}.
\end{eqnarray*}
\end{definicion}   

\par One defines left modules over an algebra $A$ analogously.  A \emph{bimodule over $A$} is a triple $(M,\lambda,\rho)$ such that $(M,\lambda)$ is a left module and $(M,\rho)$ is a right module over $A$ and $\rho \circ (\lambda\otimes Id_{A})=\lambda\circ (Id_{A}\otimes \rho)$. Note that $(A,\mu,\mu)$ is a bimodule, called \emph{regular}. A morphism of bimodules is a morphism of left and right modules. They clearly form a category.

\begin{definicion}
 Let $\mathcal{C}$ be a monoidal category and $(A,m,u)$ a unitary algebra in $\mathcal{C}$. An \emph{ideal} I of $A$ is a subobject of $A$ in the category of bimodules over $(A,m,u)$.
\end{definicion} 

 \par Let $A$ be a commutative algebra and $_{A}Mod$
be the category of $A$-modules. It is a symmetric monoidal category for the canonical tensor product over $A$ and the usual braiding denote by $\tau$.  A \emph{good category} $\mathcal{C}$ will be a symmetric monoidal category together with a fully faithfull monoidal functor from $\mathcal{C}$ to the symmetric monoidal category $_{A}Mod$, for some $A$.

\par For the following definition we suppose our category satisfies Grothendieck's axiom $Ab3^{\ast}$ (see \cite{wei}, Appendix for the definition). 

\begin{definicion}
 Given an algebra $A$ in a monoidal category $\mathcal{C}$, and a subobject $X$ of $A$, the \emph{ideal generated} by $X$ is the intersection of the family of ideals $I$ of $A$ such that $X$ is a subobject of $I$ (see \cite{mitch} Ch.I, section 9). 
 
\end{definicion}

\begin{definicion}\label{coalgebra}
Given a monoidal category $\mathcal{C}$, a \emph{counitary coalgebra in $\mathcal{C}$} is a triplet $(C,\varepsilon_{C},\Delta_{C})$ where $C\in Obj(\mathcal{C})$,  $\varepsilon_{C} \in Hom(C,\textbf{1})$, $\Delta_{C}\in Hom(C, C\otimes C)$ such that

\begin{eqnarray*}
&(\Delta_{C}\otimes Id_{C})\circ\Delta_{C}=(Id_{C}\otimes\Delta_{C})\circ \Delta_{C},\\
& Id_{C}=(\varepsilon_{C}\otimes Id_{C})\circ\Delta_{C},\\
& Id_{C}=(Id_{C}\otimes \varepsilon_{C})\circ\Delta_{C}.
\end{eqnarray*}
\end{definicion}

\par The morphism $\varepsilon_{C}$ is called the \emph{counit} of $C$ and $\Delta_{C}$ the \emph{coproduct}.

\begin{rem}
The object $\textbf{1}_{\mathcal{C}}$ together with coproduct $l_{\textbf{1}_{\mathcal{C}}}^{-1}=r_{\textbf{1}_{\mathcal{C}}}^{-1}$ (see previous remark) and counit $Id_{\mathcal{C}}$ is a counitary coalgebra.
\end{rem}

\par Since we are mainly interested in counitary coalgebras, we will only refer to them as coalgebras.

\par Similarly to the case of algebras the class of coalgebras in $\mathcal{C}$ constitute a category if we take as morphism from  $(C,\varepsilon_{C},\Delta_{C})$ to $(D,\varepsilon_{D},\Delta_{D})$ the maps $g\in Hom_{\mathcal{C}}(C,D)$, satisfying
\begin{align*}
 f\circ \varepsilon_{C}=\varepsilon_{D} &  & \text{and} & &  (f\otimes f)\circ \Delta_{C}=\Delta_{D}\circ f.
\end{align*}

\par We denote this category by  $Coalg_{\mathcal{C}}$. Similarly one defines coalgebras without counit and their morphisms. We denote this category  $_{n}Coalg_{\mathcal{C}}$
\par We refer the reader to \cite{chris}, Ch XIII for the definition of braided category. For each pair of objects $V$ and $W$  of the category $\mathcal{C}$ denote by $c_{V,W}\in Hom(V\otimes W,W\otimes V)$ the \emph{braiding} isomorphism. A symmetric monoidal category is a braided monoidal category such that $c_{V,W}\circ c_{W,V}=Id_{W\otimes V}$ for all objects $V$ and $W$ in the category.


\begin{definicion}
 Given a symmetric monoidal category $\mathcal{C}$, a unitary algebra $(A,m,u)$ in $\mathcal{C}$ is called a \emph{commutative} if $m=m\circ c_{A,A}$. 
\end{definicion}

\par If $(A,m_{A},u_{A})$ and $(B,m_{B},u_{B})$ are two commutative unitary algebras, a morphism between them will be a map $f\in Hom_{Alg_{\mathcal{C}}}(A,B)$. 
The class of commutative unitary algebras together with these morphisms forms a category which we note by $Alg^{c}_{\mathcal{C}}$.

\begin{definicion} \label{cocomcoalgebra}Given a symmetric monoidal category $\mathcal{C}$, a cocommutative counitary coalgebra in $\mathcal{C}$ is a coalgebra $(C,\Delta_{C},\varepsilon_{C})$ such that $\Delta_{C}=c_{C,C}\circ \Delta_{C}$.
\end{definicion}

\par If $(C,\Delta_{C},u_{C})$ and $(D,\Delta_{D},\Delta_{D})$ are two cocommutative counitary coalgebras, a morphism between them will be a map $f\in Hom_{Coalg_{\mathcal{C}}}(C,D)$. 
The class of cocommutative counitary coalgebras together with these morphisms forms a category which we note by $Coalg^{c}_{\mathcal{C}}$.

\par The basic example of a braided (non-strict) monoidal category is the category of vector spaces over a field. Note that the category of algebras in a braided monoidal category is also braided, as recalled in \cite{su}, Ch.1.

\begin{definicion}
Given a monoidal category $\mathcal{C}$ and a coalgebra $(C,\Delta,\varepsilon)$ in  $\mathcal{C}$ a \emph{right comodule over $C$} is pair $(M,\delta)$ where  $M\in Obj(\mathcal{C})$ and $\delta\in Hom(M,M\otimes C)$ satisfy
\begin{eqnarray}\label{224}
& (Id_{M}\otimes\Delta)\circ \delta=(\delta\otimes Id_{C})\circ \delta,
& (Id_{M}\otimes \varepsilon)\circ \delta=Id_{M}.
\end{eqnarray}
\end{definicion} 

\par The morphism $\delta$ is called the \emph{coaction}. We denote by $Com_{C}$ the category whose objects are right comodules over $C$ and whose morphisms are given as follows. If $(M,\delta_{M})$ and $(N,\delta_{N})$ are two right comodules over $C$, a morphism $f$ from $M$ to $N$ is a map  $f\in Hom_{\mathcal{C}}(M,N)$ such that
 
\begin{equation}\label{1111}
\xymatrix{
M\ar[rr]^{f} \ar[d]_{\delta_{M}} & & N \ar[d]^{\delta_{N}}\\
M\otimes C \ar[rr]^{f\otimes Id_{C}}& & N\otimes C
}
\end{equation}
is commutative.
\par Since we shall work only with right comodules, we shall refer to them just as comodules.  
\begin{prop}\label{prop 1}
Let $\mathcal{C}$ be a monoidal category, $W\in Obj(\mathcal{C})$ and $(C,\Delta_{C},\varepsilon_{C})$ a coalgebra in $\mathcal{C}$. Then $W\otimes C$ has structure of right comodule over $C$, with coaction $\delta:=Id_{W}\otimes \Delta_{C}:W\otimes C \rightarrow (W\otimes C)\otimes C $.
\end{prop} 
\begin{proof}
To see that this is a right comodule over  $C$, one must check equation \eqref{224}.
 First we will show that the following diagram
\begin{equation*}
\xymatrix{
W\otimes C \ar[rr]^{\delta} \ar[d]_{\delta} &  & (W\otimes C)\otimes C \ar[d]^{Id_{W\otimes C}\otimes \Delta_{C}} \\ 
(W\otimes C)\otimes C \ar[rr]_{\delta\otimes Id_{C}} &  & (W\otimes C)\otimes C \otimes C  }
\end{equation*}
commutes.
\par Indeed, we have the following chain of identities
\begin{align*}
& (\delta\otimes Id_{C})\circ\delta=((Id_{W}\otimes \Delta_{C})\otimes Id_{C})\circ(Id_{W}\otimes \Delta_{C}) \hspace{2cm}\\
& =(Id_{W}\otimes (\Delta_{C}\otimes Id_{C}))\circ(Id_{W}\otimes \Delta_{C})=(Id_{W}\circ Id_{W})\otimes ((\Delta_{C}\otimes Id_{C})\circ \Delta_{C}),
\end{align*}
where we have used the strictness of the category in the second equality and the property  $(f\otimes g)\circ(f'\otimes g')=(f\circ f')\otimes (g\circ g')$. Moreover,
\begin{align*}
& (\delta\otimes Id_{C})\circ\delta=(Id_{W}\circ Id_{W})\otimes (( Id_{C}\otimes \Delta_{C})\circ \Delta_{C})\\
& =(Id_{W}\otimes(Id_{C}\otimes \Delta_{C}))\circ(Id_{W}\otimes \Delta_{C})
=((Id_{W}\otimes Id_{C})\otimes \Delta_{C})\circ (Id_{W}\otimes \Delta_{C})\\
& =(Id_{W\otimes C}\otimes \Delta_{C})\circ(Id_{W}\otimes \Delta_{C}) 
=(Id_{W\otimes C}\otimes \Delta_{C})\circ \delta,
\end{align*}
where in the last row we have used that $Id_{W\otimes C}=Id_{W}\otimes Id_{C}$.
    
\par Finally we must show the commutation of the following diagram
\begin{equation*}
\xymatrix{
W\otimes C \ar[rr]^{\delta} \ar[rrd]_{Id_{W\otimes C}} &  & (W\otimes C)\otimes C \ar[d]^{Id_{W\otimes C}\otimes \varepsilon_{C}} \\ 
 &  & (W\otimes C)\otimes \textbf{1}_{\mathcal{C}} }
\end{equation*}
Indeed, we have the following chain of identities
 \begin{align*}
 & (Id_{W\otimes C}\otimes \varepsilon_{C})\circ \delta=(Id_{W\otimes C}\otimes \varepsilon_{C})\circ (Id_{W}\otimes \Delta_{C})\\
 & =((Id_{W}\otimes Id_{C})\otimes \varepsilon_{C})\circ (Id_{W}\otimes \Delta_{C}) = (Id_{W}\otimes (Id_{C}\otimes \varepsilon_{C}))\circ (Id_{W}\otimes \Delta_{C})  \\
 & =(Id_{W}\circ Id_{W})\otimes ((Id_{C}\otimes \varepsilon_{C})\circ \Delta_{C})=Id_{W}\otimes Id_{C}=Id_{W\otimes C}, 
 \end{align*}
 where we have used the strictness of the category in the third equality.        
\end{proof}

  \begin{definicion}
  Given a (respectively symmetric) monoidal category $\mathcal{C}$, a \emph{ (respectively cocommutative) coaugmented coalgebra in $\mathcal{C}$} is a cocommutative coalgebra $(C,\Delta, \varepsilon)$ in $\mathcal{C}$ provided with a morphism of (respectively cocommutative) coalgebras  $\eta\in Hom(\textbf{1}_{\mathcal{C}},C)$ such that
  \begin{align}\label{conmutacion1}
    &    \xymatrix{
    \textbf{1}\ar[rd]_{Id_{\textbf{1}}}
     \ar[rr]^{\eta} & & C \ar[dl]^{\varepsilon_{C}} \\
      & \textbf{1} & \\
         }
  \end{align}
 \end{definicion}
   
   \par Given two cocommutative coaugmented coalgebras  $(C,\Delta_{C},\varepsilon_{C},\eta_{C})$ and \\ $(D,\Delta_{D},\varepsilon_{D},\eta_{D})$, a  morphism of  cocommutative coaugmented coalgebras from $C$ to $D$ is an element $f\in Hom_{Coalg_{\mathcal{C}}}(C,D)$ such that $f\circ \eta_{C}=\eta_{D}$.
   \par The class of all cocommutative coaugmented coalgebras in $\mathcal{C}$ with morphisms between them form a category for the usual composition and identity morphisms, that we denote by $_{c}Coalg_{\mathcal{C}}$. 
   
   
   \begin{ej}
  \par If $V$ is a $K$-vector space, let $TV=\oplus_{n \geqslant 0} V^{\otimes n}$ its tensor construction. We define a counit morphism as the map $\varepsilon:TV\to K$ given by the projection on the zero degree component and a coproduct as follows. If $v=v_{1}\otimes \cdots \otimes v_{n}$ is an element of $TV$ then 
  \begin{displaymath}
  \Delta(v)=\sum_{i=1}^{n-1}v_{1}\cdots v_{i}\otimes v_{i+1}\cdots v_{n} +1_{K}\otimes v +v \otimes 1_{K}.\end{displaymath} We also define $\Delta(1_{K})=1\otimes 1$. The coproduct is a linear extension of the previous definitions.
  The coaugmentation morphism is the map $\eta:K\to TV$ given by  $\eta(1_{K})=1\in T^{0}V$. It is easy to verify that $(TV,\Delta,\varepsilon,\eta)$ is a coaugmented coalgebra in the symmetric monoidal category of vector spaces over the field $K$. 
   \end{ej}

  \subsection{Relation between coaugmented and noncounitary coalgebras in an abelian category}\label{sec111}

   \par All along this section $\mathcal{C}$ will be an abelian symmetric monoidal category. We shall define two functors $F:{_{c}Coalg_{\mathcal{C}}}\to  {_{n}Coalg_{\mathcal{C}}}$ and $G:{_{n}Coalg_{\mathcal{C}}}\to  {_{c}Coalg_{\mathcal{C}}}$ which are quasi-inverses of each other.
   \par If $(C,\Delta_{C},\varepsilon_{C},\eta_{C})$ is a coaugmented coalgebra in $\mathcal{C}$, we will denote $(\overline{C},\pi_{C})=Coker(\eta_{C})$. The universal property of the cokernel gives us a map $\Delta_{\overline{C}}$ such that the following diagram 
   
   \begin{equation}\label{diag1}
      \xymatrix{
      & & & \overline{C}\otimes \overline{C} \\
      \textbf{1} \ar_{\eta_{C}}[r] &  C \ar_{\pi_{C}}[rr] \ar^{(\pi_{C}\otimes\pi_{C})\circ \Delta_{C}}[rru] & & \overline{C} \ar@{-->}_{\Delta_{\overline{C}}}[u] 
       }
     \end{equation} 
commutes. The morphism $\Delta_{\overline{C}}$ is the only one such that the diagram commutes, and it exists because $(\pi_{C} \otimes \pi_{C})\circ \Delta_{C}\circ \eta_{C}=0$. Indeed by the definition of cocommutative coaugmented coalgebra we have that $(\pi_{C} \otimes \pi_{C})\circ \Delta_{C}\circ \eta_{C}=(\pi_{C} \otimes \pi_{C})\circ (\eta_{C} \otimes \eta_{C}) \circ \iota=((\pi_{C}\circ \eta_{C})\otimes(\pi_{C}\circ \eta_{C}))\circ \iota $, which is zero by definition of the cokernel of  $\eta_{C}$. 
 \par If $(C,\Delta_{C},\varepsilon_{C},\eta_{C})\in obj({_{c}Coalg_{\mathcal{C}}})$ we define $F(C,\Delta_{C},\varepsilon_{C},\eta_{C})=(\overline{C},\Delta_{\overline{C}})$. If $f:(C,\Delta_{C},\varepsilon_{C},\eta_{C}) \rightarrow (D,\Delta_{D},\varepsilon_{D},\eta_{D})$ is a morphism in the category ${_{c}Coalg_{\mathcal{C}}}$, we see that there exists a unique map in the category $_{n}Coalg_{\mathcal{C}}$, which we call $\overline{f}$, satisfying that 
  
  
     \begin{equation}\label{diag1.5}
        \xymatrix{
        & & & \overline{D} \\
        \textbf{1} \ar_{\eta_{C}}[r] &  C \ar_{\pi_{C}}[rr] \ar^{\pi_{D}\circ f}[rru] & & \overline{C} \ar@{-->}_{\overline{f}}[u] 
         }
       \end{equation} 
  
  \par In order to guarantee  the existence of a mapping $\overline{f}$ in $\mathcal{C}$ making \eqref{diag1.5} commute, we must only see that $(\pi_{D}\circ f) \circ \eta_{C}=0$. By definition of $Hom_{{_{c}Coalg_{\mathcal{C}}}}(C,D)$ we have that $(\pi_{D}\circ f) \circ \eta_{C}=\pi_{D}\circ (f \circ \eta_{C})=\pi_{D}\circ \eta_{D}$ which is zero because $(\overline{D},\pi_{D})$ is the cokernel of $\eta_{D}$.  In order to see that $\overline{f}\in Hom_{_{n}Coalg_{\mathcal{C}}}(\overline{C},\overline{D})$, it is enough to prove the diagram 
  \begin{equation}
     \xymatrix{
      \overline{C} \ar[d]_{\Delta_{\overline{C}}}\ar[rr]^{\overline{f}}& &  \overline{D} \ar[d]^{\Delta_{\overline{D}}}  \\
       \overline{C}\otimes\overline{C} \ar[rr]_{\overline{f}\otimes \overline{f}}& & \overline{D}\otimes \overline{D} \\
       }
    \end{equation}
commutes.   
    
 \par  We will prove the equality  $\Delta_{\overline{D}}\circ \overline{f}=(\overline{f}\otimes \overline{f}) \circ \Delta_{\overline{C}}$ by using that $\pi_{C}:C\to \overline{C}$ is an epimorphism, because it is a cokernel. So it is enough proving $\pi_{\overline{D}}\circ \overline{f}\circ \pi_{C}=(\overline{f}\otimes \overline{f}) \circ \Delta_{\overline{C}}\circ \pi_{C}$. Indeed,
 \begin{eqnarray*}
 & \Delta_{\overline{D}}\circ \overline{f}\circ \pi_{C}=\Delta_{\overline{D}}\circ \pi_{D} \circ f=(\pi_{D}\otimes \pi_{D})\circ \Delta_{D}\circ f\\
 & =(\pi_{D}\otimes \pi_{D})\circ (f\otimes f)\circ \Delta_{C}=(\overline{f}\otimes \overline{f})\circ (\pi_{C}\otimes \pi_{C})\circ \Delta_{C}\\
 & =(\overline{f}\otimes \overline{f})\circ \Delta_{\overline{C}}\circ \pi_{C} 
  \end{eqnarray*}
 where we have used that $\overline{f}\circ \pi_{C}=\pi_{D} \circ f$ in the first equality and the fact that the morphisms of coalgebras satisfy $\Delta_{D}\circ f= (f\otimes f)\circ \Delta_{C}$ in the third one together with common properties of the morphisms of monoidal categories. 
  
  
  \par If $f \in Hom_{_{c}Coalg_{\mathcal{C}}}(C,D)$ we define $F(f)=\overline{f}\in Hom_{_{n}Coalg_{\mathcal{C}}}(\overline{C},\overline{D})$. To prove that this assignment is a functor we must see that
  \begin{itemize}
  \item $F(g\circ f)=F(g)\circ F(f)$,
  \item $F(Id_{C})=Id_{\overline{C}}$.
  \end{itemize} 
  \par In the following we suppose that $f\in  Hom_{_{c}Coalg_{\mathcal{C}}}(C,D)$ and $g\in Hom_{_{c}Coalg_{\mathcal{C}}}(D,E)$.
  \par We know that $F(g\circ f)$ is the only morphism such that $F(g\circ f)\circ \pi_{C}=\pi_{E}\circ(g\circ f)$. Since
  \begin{eqnarray*}
  F(g)\circ F(f)\circ\pi_{C}=F(g)\circ \overline{f}\circ \pi_{C} = \overline{g} \circ \pi_{D}\circ f=\pi_{E}\circ g\circ f,
  \end{eqnarray*} 
  then, by uniqueness we conclude that $F(g\circ f)=F(g)\circ F(f)$. Diagram \eqref{diag1.5} implies that $F(Id_{C})=Id_{\overline{C}}$. Then $F:{_{c}Coalg_{\mathcal{C}}}\to  {_{n}Coalg_{\mathcal{C}}}$ is a functor. Moreover; we will prove that it is a categorical equivalence, and in order to see that we will construct its quasi-inverse functor $G:{_{n}Coalg_{\mathcal{C}}}\to  {_{c}Coalg_{\mathcal{C}}}$. 
  

\par Given an object $(E,\Delta_{E})$ in the category ${_{n}Coalg_{\mathcal{C}}}$, we define its image under $G$ by $(E\oplus \textbf{1}_{\mathcal{C}},\Delta_{E\oplus \textbf{1}_{\mathcal{C}}},\pi_{2},i_{2})$, where $\pi_{2}:E\oplus \textbf{1}_{\mathcal{C}}\to \textbf{1}_{\mathcal{C}}$ is the canonical projection of the direct sum and $i_{2}:\textbf{1}_{\mathcal{C}}\to E\oplus \textbf{1}_{\mathcal{C}}$ is the canonical inclusion map of the direct sum. 
  The definition of $\Delta_{E\oplus \textbf{1}_{\mathcal{C}}}$ is a little more involved. We write $\textbf{1}$ instead of $\textbf{1}_{\mathcal{C}}$ to denote the unit object. We first define the morphisms
  \begin{align*}
 &  \tilde{\Delta}_{E}=(i_{1}\otimes i_{1})\circ \Delta_{E}+(i_{1}\otimes i_{2})\circ r_{E}^{-1}+(i_{2}\otimes i_{1})\circ l_{E}^{-1}, \\
 &  \Delta_{\textbf{1}}=(i_{2}\otimes i_{2})\circ \iota^{-1}
  \end{align*} 
   Hence, there is a unique morphism
   $$\Delta_{E\oplus \textbf{1}}\in  Hom_{\mathcal{C}}(E \oplus \textbf{1},(E\oplus \textbf{1})\otimes (E\oplus\textbf{1}))$$ 
   such that $\Delta_{E\oplus \textbf{1}}\otimes i_{1}=\tilde{\Delta}_{E}$ and $\Delta_{E\oplus \textbf{1}}\otimes i_{2}=\Delta_{\textbf{1}}$, i.e. making the following diagram
   
   \begin{equation}
   \xymatrix{
     & &  & \textbf{1} \ar^{i_{2}}[d] \ar_{\Delta_{\textbf{1}}}[llld] \\
     (E\oplus \textbf{1})\otimes (E\oplus \textbf{1})  & &  &  E\oplus 1 \ar@{-->}_{\hspace{1.2cm}\exists!\hspace{0.1cm} \Delta_{E\oplus \textbf{1}}}[lll] \\
    & & & E \ar^{\tilde{\Delta}_{E}}[lllu] \ar_{i_{1}}[u] \\
        }
  \end{equation}
commute.
  
  
  
  \par From now on, for simplicity we shall write $E_{+}=E\oplus \textbf{1}$. The next step is to define $G$ on morphisms. If $f$ is in  $Hom_{{_{n}Coalg_{\mathcal{C}}}}(C,D)$, then we define $G(f)=i_{1}^{D}\circ f\circ \pi_{1}^{C}+i_{2}^{D}\circ\pi_{2}^{C}$, where $i_{1}^{D}:D\to D_{+}$, and $\pi_{D_{+}}:D_{+}\to \overline{(D_{+})}$ is the cokernel of $i_{2}^{D}:\textbf{1}\to D_{+}$ (i.e. the coaugmentation of $D_{+}$). We will denote $G(f)$ by $f_{+}$.
  
  \begin{ej} If we are in the category of vector spaces over a field $K$, the previous constructions are of the form:
  \ 
     \begin{itemize}
  \item $G(E)=E\oplus \textbf{1}$ will actually be $E\oplus K\cdot 1_{E_{+}}$,
  \item The coaugmentation  $\eta(1_{K})=1_{E_{+}}$,
  \item $\Delta_{E\oplus K}(e)=\Delta_{E}(e)+1_{E_{+}}\otimes e+e\otimes 1_{E_{+}}$,
  \item If $f\in Hom_{_{n}Coalg_{\mathcal{C}}}(E,H)$, then $G(f)(e+\lambda \cdot 1_{E_{+}})=f(e)+\lambda \cdot 1_{H_{+}}$.
  \end{itemize}
  \end{ej}
  
\begin{thm}\label{teo1}Let $\mathcal{C}$ be a monoidal abelian category. Then, the functor $$F:{_{c}Coalg_{\mathcal{C}}}\to {_{n}Coalg_{\mathcal{C}}}$$ is a categorical equivalence whose quasi-inverse functor is $$G:{_{n}Coalg_{\mathcal{C}}}\to {_{c}Coalg_{\mathcal{C}}}.$$
\end{thm}
\begin{proof} If $X\in Obj(_{n}Coalg_{\mathcal{C}})$ we construct the natural map $N_{X}:X\to \overline{(X_{+})}$ given by the composition $N_{X}=\pi_{X_{+}}\circ i_{1}^{X}$. Given two objects $X, Y\in Obj(_{n}Coalg_{\mathcal{C}})$ and $f\in Hom_{_{n}Coalg_{\mathcal{C}}}(X,Y)$, we will show that
  

  \begin{equation}
   \xymatrix{
   X \ar[rr]_{f}   \ar[d]_{N_{X}} & & Y  \ar[d]^{N_{Y}}\\
    \overline{(X_{+})}\ar[rr]^{\overline{(f_{+})}}& &  \overline{(Y_{+})}}
  \end{equation}
is commutative. In order to do that, note that 
\begin{eqnarray*}
& f_{+}\circ i_{1}^{X}=(i_{1}^{Y}\circ f \circ \pi_{1}^{X} + i_{2}^{Y}\circ \pi_{2}^{X})\circ i_{1}^{X}\\
& =i_{1}^{Y}\circ f \circ \pi_{1}^{X}\circ i_{1}^{X} + i_{2}^{Y}\circ \pi_{2}^{X}\circ i_{1}^{X}\\
& = i_{1}^{Y}\circ f \circ Id_{X} + i_{2}^{Y}\circ 0 =i_{1}^{Y}\circ f,
\end{eqnarray*}    
  where we have used the bilinearity of the composition in an abelian category in the second equality and the facts that $\pi_{1}^{X}\circ i_{1}^{X}=Id_{X}$ and $\pi_{2}^{X}\circ i_{1}^{X}=0$ in the third one. Hence, 
    
  \begin{equation*}
   N_{Y}\circ f=\pi_{Y_{+}}\circ i_{1}^{Y}\circ f=\pi_{Y_{+}}\circ f_{+}\circ i_{1}^{X}=\overline{(f_{+})} \circ \pi_{X_{+}}\circ i_{1}^{X}= \overline{(f_{+})} \circ N_{X}
  \end{equation*}
 where in the third equality we use the commutation of the diagram \eqref{diag1.5} and in the last one the definition of the natural transformation. Then the functor $F\circ G$ is naturally isomorphic to $Id_{_{n}Coalg_{\mathcal{C}}}$. 
  \par Also it is possible to construct a natural map between $Id_{\mathcal{C}_{_{c}Coalg_{\mathcal{C}}}}$ and $G\circ F$. Suppose $(X,\Delta_{X},\varepsilon_{X},\eta_{X})\in Obj(_{c}Coalg_{\mathcal{C}})$, we define $L_{X}:X\to (\overline{X})_{+}$ by $L_{X}=i_{1}^{\overline{X}}\circ \pi_{X}+i_{2}^{X}\circ \varepsilon_{X}$, where $i_{1}^{\overline{X}}:\overline{X}\to (\overline{X})_{+}$ and $i_{2}^{X}:\textbf{1} \to (\overline{X})_{+}$ are the inclusion maps, $\pi_{X}:X\to \overline{X}$ is the cokernel of $\eta_{X}$ and $\pi_{1}^{\overline{X}}:(\overline{X})_{+}\to \overline{X}$ and $\pi_{1}^{\overline{X}}:(\overline{X})_{+}\to \textbf{1}$ are the projections.
  \par Note that
  \begin{eqnarray*}
  & (\overline{f})_{+}\circ L_{X}=(i_{1}^{\overline{Y}}\circ \overline{f}\circ \pi_{1}^{\overline{X}}+i_{2}^{Y}\circ \pi_{2}^{\overline{X}})\circ (i_{1}^{\overline{X}}\circ \pi_{X}+i_{2}^{X}\circ \varepsilon_{X})\\
  & =i_{1}^{\overline{Y}}\circ \overline{f}\circ \pi_{1}^{\overline{X}}\circ i_{1}^{\overline{X}}\circ \pi_{X}+
  i_{1}^{\overline{Y}}\circ \overline{f}\circ \pi_{1}^{\overline{X}}\circ i_{2}^{X}\circ \varepsilon_{X}\\
  & +i_{2}^{Y}\circ \pi_{2}^{\overline{X}}\circ i_{1}^{\overline{X}}\circ \pi_{X}  +i_{2}^{Y}\circ \pi_{2}^{\overline{X}}\circ i_{2}^{X}\circ \varepsilon_{X},
  \end{eqnarray*}
  where we have used the bilinearity of the composition, the definition of the functor $G$ and the natural map $N_{X}$ defined above. Using that $\pi_{1}^{\overline{X}}\circ i_{1}^{\overline{X}}=Id_{\overline{X}}$, $\pi_{1}^{\overline{X}}\circ i_{2}^{X}=0$, $\pi_{2}^{\overline{X}}\circ i_{1}^{\overline{X}}=0$ and $\pi_{2}^{\overline{X}}\circ i_{2}^{\overline{X}}=Id_{\textbf{1}}$, we get that
  
   \begin{equation*}
  (\overline{f})_{+}\circ L_{X}=i_{1}^{\overline{Y}}\circ \overline{f}\circ \pi_{X}+i_{2}^{Y}\circ \varepsilon_{X}. 
   \end{equation*}
     
 \par On the other hand, the computation of $N_{Y}\circ f$ gives us
 
 \begin{align*}
 & N_{Y}\circ f= (i_{1}^{\overline{Y}}\circ \pi_{Y}+ i_{2}^{Y}\circ \varepsilon_{Y})\circ f=i_{1}^{\overline{Y}}\circ \pi_{Y}\circ f+ i_{2}^{Y}\circ \varepsilon_{Y} \circ f\\
 & =i_{1}^{\overline{Y}}\circ \overline{f}\circ \pi_{X}  + i_{2}^{Y}\circ \varepsilon_{X},
 \end{align*}  
where we have used the 
definition of $\overline{f}$ and the fact that $f$ is a morphism of  (counitary) coalgebras in the last equality. Hence $N_{Y}\circ f=(\overline{f})_{+}\circ L_{X}$ which together with the previous result implies that $F$ and $G$ are quasi-inverse functors.   
  \end{proof}

  \ 
  
  
  

  
  
   \par Let us denote $\Delta^{(1)}=Id_{\overline{C}}$, $\Delta_{\overline{C}}^{(2)}=\Delta_{\overline{C}}$ and $\Delta_{\overline{C}}^{(n)}=(\Delta_{\overline{C}}\otimes Id^{\otimes (n-2)})\circ \Delta_{\overline{C}}^{(n-1)}$ if $n\geqslant 3$. 
    Notice that there is a canonical morphism $Ker\Delta_{\overline{C}}^{(n)} \hookrightarrow Ker\Delta_{\overline{C}}^{(n+1)}$ for all $n\geqslant 2$.
  
  \begin{definicion}\label{cocom}
  A (respectively cocommutative) coaugmented coalgebra $(C,\Delta,\varepsilon, \eta)$ is \emph{cocomplete (or conilpotent)} if $\displaystyle{\overline{C}=colim_{n\to \infty}Ker \Delta_{\overline{C}}^{(n)}}$.
  \end{definicion}

  \subsection{Symmetric coalgebras}
 
 \par Let us begin this section with an example.
  
  \begin{ej}\label{SV} Let $\mathcal{C}$ be a good symmetric monoidal category. Given $V\in Obj(\mathcal{C})$, let us consider the \emph{tensor algebra}
  \begin{equation}
  T(V)=\bigoplus_{n\in \mathbb{N}_{0}}V^{\otimes n},
  \end{equation}
  where $V^{\otimes 0}=A$ and the product is given by concatenation. The algebra $S(V)$ is defined by $S(V)=T(V)/I$, where $I$ is the ideal in $T(V)$ generated by the image of $Id_{V}^{\otimes 2}-\tau_{V,V}$ as a subobject of $V^{\otimes 2}$ (where $\tau$ is the braiding given by the usual twist). We write
  \begin{equation}
  S(V)=\bigoplus_{n\in \mathbb{N}_{0}}S^{n}V.
  \end{equation}
  We will define a structure of cocommutative coaugmented coalgebra on it, denoted by $S^{c}V$.
  The coproduct is defined as follows.  
  We set $\Delta(1_{S^{c}V})=1_{S^{c}V}\otimes 1_{S^{c}V}$, and for an element of the form $v_{1}\cdots v_{n}$ with $v_{i}\in V$ we define

  \begin{equation}\label{coproductoSV}
 \Delta(v_{1}\cdots v_{n})=\sum_{I,J/I\sqcup J=\{1\cdots n\}}v_{I}\otimes_{A} v_{J},
  \end{equation} 
  where $v_{I}=v_{i_{1}}\cdots v_{i_{k}}$  if $I=\{ i_{1}<\cdots <i_{k} \}\neq \emptyset$ and $v_{\emptyset}=1_{S^{c}V}$ 
  The counit $\varepsilon:S^{c}(V)\rightarrow A$ is given by the canonical projection onto the zeroeth component, and the coaugmentation $\eta:A\to S^{c}(V)$
is the canonical inclusion.
 \end{ej}

 \par From now on, we denote this coalgebra only by $SV$, unless we say the opposite. 
 Note that $\overline{SV}$ defined in 1.1.2. is just $\overline{SV}=\oplus_{n\in \mathbb{N}} S^{n}V$ and $\Delta_{\overline{SV}}$ is
  \begin{equation}\label{copSV}
   \Delta_{\overline{SV}}(v_{1}\cdots v_{n})=\sum_{I,J\neq \emptyset /I\sqcup J=\{1\cdots n\}}v_{I}\otimes v_{J}
  \end{equation} 
  for $n\geqslant 2$ and is zero over $V$.
 
  \begin{rem}
  Notice that $SV$ is cocomplete. To prove that, it is sufficient note that $\Delta_{\overline{SV}}^{(n)}$ vanishes on elements of $\oplus_{k=1}^{n}S^{k}V$.
    
    
  \end{rem}

   \begin{rem}\label{symme}
   Note that the construction of the coalgebra $SV$ can be done for any object $V$ in a symmetric monoidal category $\mathcal{C}$ with arbitrary direct sums such that $\otimes_{\mathcal{C}}$ commutes with $\oplus$. The Theorem 1 in Chapter XI.1 from \cite{mac} give us a canonical action from $\mathcal{S}_{n}$ to the tensor construction of degree $n$, $TV$.  This theorem allow us to construct the symmetric coalgebra as equivalence classes of this action.
   Let us denote by $\pi_{I}:S^{n}V\to S^{|I|}V$ the projection map in the symmetric product with component in $I\subset\{1,\cdots,n\}$. Then we can consider the maps $\Delta_{n}: S^{n}V\to SV \otimes SV$ given by 
   \begin{equation*}
   \Delta_{n}=\bigoplus_{I,J | I\sqcup J=\{ 1,\cdots,n\} } \pi_{I}\otimes \pi_{J},
   \end{equation*} 
   and define the coproduct  $\Delta:SV\to SV \otimes SV$  for $SV$ as $\Delta=\bigoplus_{n\in \mathbb{N}_{0}}\Delta_{n}$. 
   \par However in \cite{mac} there is another result in Chapter XI, establishing that there is a strong monoidal functor between the tensor algebra construction and a free algebra in the category algebras in a monoidal category. This result is equally valid for the case of coalgebras.
   \end{rem}
  
  \par If $C$ is a coaugmented coalgebra in $\mathcal{C}$, we shall denote by $1_{C}$ the image $\eta(1_{K})$. Note it is a distinguished \emph{group-like} element of $C$.
  
  \begin{prop}\label{prop3}Given a symmetric monoidal category $\mathcal{C}$ with arbitrary direct sums such that $\otimes_{\mathcal{C}}$ commutes with $\oplus$, for any  $(C,\Delta_{C},\varepsilon_{C},\eta_{C})\in {_{c}Coalg_{\mathcal{C}}^{c}}$, any $V\in Obj(\mathcal{C})$ and $p:C\rightarrow V$ morphism in $Hom_{\mathcal{C}}(C,V)$ such that $p\circ \eta_{C}=0$, there exists a unique morphism $P\in Hom_{_{c}Coalg_{\mathcal{C}}^{c}}(C,SV)$ such that the following diagram
  \begin{equation}
   \xymatrix{
    C \ar@{-->}[rd]_{P} \ar[rr]^{p}& &   V  \\
    & SV\ar[ru]_{\pi}  &  \\
    }
  \end{equation}
   commutes, where $\pi$ is the canonical projection to the elements of degree one. 
    \end{prop}
  \begin{proof}
  \par  The reader may suppose that $\mathcal{C}$ is a good category. If he or she is willing to, but the proof is exactly the same. 
  \par If $C$ is as in the statement, by Theorem \ref{teo1} we can write $C\simeq\overline{C} \oplus A$. So, it suffices to define $P$ by its restrictions to each direct summand. We first set $P$ such that  $P\circ \eta_{C}=\eta_{SV}$. In the sequel we shall denote  $p^{\odot \alpha}:\overline{C}^{\otimes \alpha}\to S^{\alpha}V$ to denote the symmetric product of the map $p:\overline{C}\to V$, i.e. $p^{\otimes l}$ composed with the quotient projection $T^{n}V\to S^{n}V$. 
   Moreover, we define $\overline{P}=P|_{\overline{C}}$ by
  \begin{equation}
  \overline{P}=\sum_{n\in \mathbb{N}}\frac{1}{n!}(p^{\odot n}\circ \Delta_{\overline{C}}^{(n)}),
  \end{equation}
   where we are committing the abuse of calling $p$ the morphism $\overline{C}\to V$ induced by $p:C\to V$. This induced map exists because of the universal property of the cokernel since $p\circ \eta_{C}=0$, as depicted below
   
   \begin{equation}
       \xymatrix{
          & & & & V  \\
          A \ar[rr]_{\eta_{C}} &  &  C \ar[rr]_{\pi_{C}} \ar[rru]^{p} &  & \overline{C} \ar@{-->}[u]_{\exists! p} }
   \end{equation}

  In order to see that $\overline{P}$ is a coalgebra morphism we must prove the commutation of
  
  
  
    \begin{equation}
    \xymatrix{
      \overline{C} \ar[rr]^{\overline{P}} \ar[d]_{\Delta_{\overline{C}}} & &   \overline{SV}\ar[d]^{\Delta_{ \overline{SV}}}  \\
      \overline{C}\otimes \overline{C}\ar[rr]_{\overline{P}\otimes\overline{P}}   &  & \overline{SV}\otimes \overline{SV} \\}
   \end{equation}  
   
   Changing some parentheses in the firsts equalities and using the definition of $\overline{P}$ and the equation  \eqref{copSV} in the fourth equality, we have
   
   \begin{align*}
   & \Delta_{\overline{SV}}\circ \overline{P}=\Delta_{\overline{SV}}\circ(\sum_{n\in \mathbb{N}}\frac{1}{n!}(p^{\odot n}\circ \Delta_{\overline{C}}^{(n)}))=\sum_{n\in \mathbb{N}}\frac{1}{n!}\Delta_{\overline{SV}}\circ(p^{\odot n}\circ \Delta_{\overline{C}}^{(n)})\\
   & =\sum_{n\in \mathbb{N}}\frac{1}{n!} (\Delta_{\overline{SV}}\circ p^{\odot n})\circ \Delta_{\overline{C}}^{(n)} = \sum_{n\in \mathbb{N}}\frac{1}{n!}\sum_{I,J\neq \emptyset  /I\sqcup J=\{1\cdots n\}}(p^{\odot |I|}\otimes p^{\odot |J|})\circ \Delta_{\overline{C}}^{(n)}.\\
   \end{align*}

   \par By rewriting $\Delta^{(n)}_{\overline{C}}$ 
    and using that in a monoidal category the morphisms satisfy $(f\otimes g)\circ (h\otimes r)=(f\circ h)\otimes (g\circ r)$, we see that the latter sum coincides with
   
   \begin{align*}
    & \sum_{n\in \mathbb{N}}\frac{1}{n!}\sum_{I,J\neq \emptyset  /I\sqcup J=\{1\cdots n\}}(p^{\odot |I|}\otimes p^{\odot |J|})\circ (\Delta_{\overline{C}}\otimes Id)^{|I|-1}\circ(Id\otimes \Delta_{\overline{C}})^{|J|-1}\circ \Delta_{\overline{C}}\\
    & =\sum_{n\in \mathbb{N}}\frac{1}{n!}\sum_{I,J\neq \emptyset  /I\sqcup J=\{1\cdots n\}}(p^{\odot |I|}\circ \Delta_{\overline{C}}^{(|I|)} \otimes p^{\odot |J|})\circ(Id\otimes \Delta_{\overline{C}})^{|J|-1}\circ \Delta_{\overline{C}}\\
    & =\sum_{n\in \mathbb{N}}\frac{1}{n!}\sum_{I,J\neq \emptyset  /I\sqcup J=\{1\cdots n\}}(p^{\odot |I|}\circ \Delta_{\overline{C}}^{(|I|)} \otimes p^{\odot |J|}\circ \Delta_{\overline{C}}^{(|J|)})\circ \Delta_{\overline{C}}\\
    & =\sum_{n\in \mathbb{N}}\frac{1}{n!}\sum_{I\subseteq \{1\cdots n\} / 1<|I|<n }(p^{\odot |I|}\circ \Delta_{\overline{C}}^{(|I|)} \otimes p^{\odot n-|I|}\circ \Delta_{\overline{C}}^{(n-|I|)})\circ \Delta_{\overline{C}}.
   \end{align*}
   
   Since $\odot$ is a symmetric product, any term who has $p^{\odot t}\circ \Delta_{\overline{C}}^{(t)}\otimes p^{\odot n-t}\circ \Delta_{\overline{C}}^{(n-t)}$ is equal to any other who has $p^{\odot r}\circ\Delta_{\overline{C}}^{(r)} \otimes p^{\odot n-r}\circ \Delta_{\overline{C}}^{(n-r)}$ if and only if $t=r$. Therefore for a fixed $r\in \{1,...,n-1\}$ there are $\binom{n}{r}$ terms equal to $p^{\odot r}\circ\Delta_{\overline{C}}^{(r)} \otimes p^{\odot n-r}\circ \Delta_{\overline{C}}^{(n-r)}$. Hence,
   
    \begin{align*}
    &  =\sum_{n\in \mathbb{N}}\frac{1}{n!}\sum_{1<|I|<n}\binom{n}{|I|}(p^{\odot |I|}\circ \Delta_{\overline{C}}^{(|I|)} \otimes p^{\odot n-|I|}\circ \Delta_{\overline{C}}^{(n-|I|)})\circ \Delta_{\overline{C}}\\
    &  =\sum_{n\in \mathbb{N}}\sum_{1<|I|<n}\frac{1}{|I|!}\frac{1}{(n-|I|)!} (p^{\odot |I|}\circ \Delta_{\overline{C}}^{(|I|)} \otimes p^{\odot n-|I|}\circ \Delta_{\overline{C}}^{(n-|I|)})\circ \Delta_{\overline{C}}\\
     & =[\sum_{n\in \mathbb{N}}\sum_{1<k<n}\frac{1}{k!}\frac{1}{(n-k)!} (p^{\odot k}\circ \Delta_{\overline{C}}^{(k)} \otimes p^{\odot n-k}\circ \Delta_{\overline{C}}^{(n-k)})]\circ \Delta_{\overline{C}}\\
     & =(\sum_{r\in \mathbb{N}}\frac{1}{r!}(p^{\odot r}\circ \Delta_{\overline{C}}^{(r)})\otimes \sum_{s\in \mathbb{N}}\frac{1}{s!}(p^{\odot s}\circ \Delta_{\overline{C}}^{(s)}))\circ \Delta_{\overline{C}}\\
      & =(\overline{P}\otimes \overline{P})\circ \Delta_{\overline{C}},
   \end{align*}
   
   which is exactly the commutation we want.  
   \end{proof}

    \begin{rem}\label{correspondencia1}The previous proposition implies that there is a bijection 
    \begin{equation}
    Hom_{_{_{c}Coalg_{\mathcal{C}}^{c}}}(C,SV) \leftrightsquigarrow \{f\in Hom_{_{n}Coalg_{\mathcal{C}}}(C,V) /f\circ \eta_{C}=0\} 
    \end{equation}
    for all cocommutative coaugmented coalgebra $(C,\Delta_{C},\varepsilon_{C},\eta_{C})\in {_{c}Coalg_{\mathcal{C}}^{c}}$  in any symmetric monoidal category $\mathcal{C}$.
    \end{rem}
    
    Since $\overline{P}$ is the restriction of $P$, we shall omit the bar in order to simplify the notation.
  
   \par  Given 
   $V\in Obj(\mathcal{C})$, we specialize Proposition \ref{prop3} in the case $C=SV$. 
   To every $P\in Hom_{_{_{c}Coalg_{\mathcal{C}}^{c}}}(SV,SV)$ it corresponds a unique $p=\pi\circ P|_{\overline{SV}}\in  Hom_{\mathcal{C}}(\overline{SV},V)$ that will be represented by a denumerable family of maps $\pi\circ P|_{S^{n}V}$, which we denote $\{p_{n}\}_{n\in \mathbb{N}}$ and call it the \emph{sequential representation} of $P$. If we apply the formula $P|_{\overline{SV}}=\sum_{n\in \mathbb{N}}\frac{1}{n!}(p^{\odot n}\circ \Delta_{SV}^{(n)})$ in this case, we get  
      \begin{equation} \label{3.33}
       P|_{S^{n}V}(v_{1}\cdots v_{n}) =\sum_{m=1}^{n}\sum_{I_{1}\cdots I_{m}\neq \emptyset; I_{1}\sqcup \cdots \sqcup I_{m}=\{1\cdots n\}}\frac{1}{m!}p_{|I_{1}|}(v_{1})\odot ... \odot p_{|I_{m}|}(v_{I_{m}})
      \end{equation}
    where $|X|$ is the cardinality of the set $X$. Note that the sequential representation of $\pi\circ Id_{SV}$ is $\{Id_{V},0,0,\cdots \}$.\\ 
   
     \begin{prop}\label{prop4} Let $\mathcal{C}$ be a symmetric monoidal category with arbitrary direct sums such that $\otimes_{\mathcal{C}}$ commutes with $\oplus$, the morphisms
      $P\in Hom_{_{c}Coalg_{\mathcal{C}}^{c}}(SV, SV)$ are represented by maps $ p\in Hom_{\mathcal{C}}(\overline{SV},V)$. Moreover $P$ is an isomorphism if and only if $p_{1}$ is an isomorphism, where we use the notation of sequential representation introduced previously.
      \end{prop}
     \begin{proof} As before, the reader may assume that  $\mathcal{C}$ is a good category if he or she is willing to, but the proof applies analogously to the general case.
    It is sufficient to prove the second statement, for the first is a consequence of Remark \ref{correspondencia1}. We will prove that $P:SV\rightarrow SV$ is an isomorphism if and only if  $\pi \circ P|_{S^{1}V}=p_{1}$ is an isomorphism.
    Consider $P, Q\in Hom_{_{c}Coalg_{\mathcal{C}}^{c}}(SV,SV)$. By formula \eqref{3.33}, we obtain that $\pi \circ P\circ Q|_{\overline{SV}} = \sum_{n\in \mathbb{N}}h_{n}$ where 
      \begin{equation}\label{composicion}
       h_{n}(v_{1}\cdot v_{2}\cdots v_{n})=\sum_{m=1}^{n}\frac{1}{m!}\sum_{I_{1}\cdots I_{m}\neq \emptyset; I_{1}\sqcup \cdots \sqcup I_{m}=\{1\cdots n\} }p_{m}(q_{|I_{1}|}(v_{I_{1}})\cdots q_{|I_{m}|}(v_{I_{m}})). 
      \end{equation} 
      
   Suppose that $P$ is bijective and take  $Q=P^{-1}$. Then \eqref{composicion} for $n=1$ implies that $q_{1}$ is the inverse of $p_{1}$. 
 \par On the other hand assume that $p_{1}$ is bijective. We will prove that  $P$ is bijective. In order to do so, we will exhibit a left and right inverses of $P$. 
 \par  Define recursively the family $\{q_{n}\}_{n\in \mathbb{N}}$ where $q_{n}\in Hom_{\mathcal{C}}(S^{n}V,V)$ is given as follows. Set $q_{1}=p_{1}^{-1}$ and if we have defined all the  $q_{m}$'s with $m < n$ then
      
    \begin{equation}\label{3.35}
     q_{n}(w_{1}\cdot w_{2}\cdots w_{n})=-\sum_{m=1}^{n-1}\sum_{I_{1}\cdots I_{m}\neq \emptyset ; I_{1}\sqcup \cdots \sqcup I_{m}=\{1\cdots n\} }q_{m}(p_{|I_{1}|}(\tilde{w}_{I_{1}})\cdots p_{|I_{m}|}(\tilde{w}_{I_{m}})) 
    \end{equation}
 for $n\geqslant 2$,where $\tilde{w}_{I_{k}}=q_{1}^{\otimes |I_{k}|}(w_{I_{k}})$. Let $Q:SV\to SV$ be the morphism of coalgebras whose sequential representation is  $\{q_{n}\}_{n\in \mathbb{N}}$. Then by means of \eqref{composicion} one can see that the sequential representation of $\pi\circ (Q\circ P)$ is $\{ Id_{V},0, \cdots  \} $ and thus by uniqueness of the correspondence in Proposition \ref{prop3} we get that $Q\circ P=Id_{\overline{S^{c}V}}$ and consequently $P$ is injective. 
        
        
        
        
        
    \par Let us now prove the surjectivity of $P$. Define a family of morphisms  $\{\tilde{q}_{n}\}_{n\in \mathbb{N}}$, where $\tilde{q}_{n}\in Hom_{\mathcal{C}}(S^{n}V,V) $ is defined as follows. Set $\tilde{q}_{1}=p_{1}^{-1}$ and if $\tilde{q}_{1},\cdots \tilde{q}_{n-1}$ for $n\in \mathbb{N}_{\geqslant 2}$ are defined, we fix
   
     \begin{equation}\label{inv}
      \tilde{q}_{n}(w_{1}\cdot w_{2}\cdots w_{n})=-\sum_{m=2}^{n}\frac{1}{m!}\sum_{I_{1}\cdots I_{m}\neq \emptyset ; I_{1}\sqcup \cdots \sqcup I_{m}=\{1\cdots n\} }\tilde{q}_{1}(p_{m}(\tilde{q}_{|I_{1}|}(w_{I_{1}})\cdots \tilde{q}_{|I_{m}|}(w_{I_{m}})) )
     \end{equation}
     for $n\geqslant 2$.     
        
        
        
        
  \par  This sequential representation defines a unique $\tilde{Q}\in Hom_{_{c}Coalg_{\mathcal{C}}^{c}}(SV,SV)$ for which $P\circ \tilde{Q}=Id_{SV}$. Then, $P$ is surjective. As a consequence, $P$ is bijective if and only if  $p_{1}$ is bijective. 
   \end{proof}
   
  \begin{prop}\label{prop5}Let $\mathcal{C}$ be a monoidal abelian category, $(D,\Delta_{D},\varepsilon_{D})$ a coalgebra in $\mathcal{C}$ and  $M$ a $D$-comodule in $\mathcal{C}$. If $W\in Obj(\mathcal{C})$, then there is a one to one correspondence 
   \begin{equation}\label{ddd}
   \fbox{
   \xymatrix{ Hom_{Com_{D}}(M,W\otimes D) \ar[r] &  Hom_{\mathcal{C}}(M,W) \ar[l] }}
   \end{equation}
   given by $f\mapsto (Id_{W}\otimes \varepsilon_{D})\circ f$, where $W\otimes D$ is the right $D$ comodule described in Proposition \ref{prop 1}. The inverse is given by sending $g$ to $\tilde{g}=(g\otimes Id_{D})\circ \rho_{M}$, where $\rho_{M}$ is the coaction of $M$.
   \end{prop}
   
   \begin{proof} 
   
\par Let us show that if $h\in Hom_{\mathcal{C}}(M,W)$ then the morphism $\tilde{h}:M\to W\otimes D$ is an isomorphism of $D$-comodules, i.e. satisfies that $(\tilde{h}\otimes Id_{D})\circ \rho_{M}=\rho_{W\otimes D}\circ \tilde{h}$, where $\rho_{W\otimes D}$ is the coaction of $W\otimes D$. Indeed, we have the equations

   $$ (\tilde{h}\otimes Id_{D})\circ \rho_{M}=(((h\otimes Id_{D})\circ \rho_{M})\otimes Id_{D})\circ \rho_{M}$$
   \begin{equation}\label{ygyg} =((h\otimes Id_{D})\otimes Id_{D})\circ ( \rho_{M}\otimes Id_{D} ) \circ \rho_{M}\end{equation}
   $$ =(h\otimes (Id_{D}\otimes Id_{D}))\circ ( \rho_{M}\otimes Id_{D} ) \circ \rho_{M},$$
where we have used the definition of $\tilde{h}$ in this first equality. Moreover, since $Id_{D}\otimes Id_{D}=Id_{D\otimes D}$, the latter member of \eqref{ygyg} coincides with 
   \begin{align*}
   &(h\otimes Id_{D\otimes D})\circ ( Id_{M}\otimes \Delta_{D} ) \circ \rho_{M}=(h\circ Id_{M}\otimes Id_{D\otimes D}\circ \Delta_{D}) \circ \rho_{M}\\
   & =(h\otimes \Delta_{D})\circ \rho_{M}=(Id_{W}\otimes \Delta_{D})\circ (h\otimes Id_{D}) \circ \rho_{M}\\
   & =\rho_{W\otimes D}\circ(h \otimes Id_{D})\circ \rho_{M}=\rho_{W\otimes D}\circ \tilde{h},
   \end{align*} where we have used equation \eqref{224} in the first equality. This proves that the map $g\mapsto \tilde{g}$ is well defined.
 
 \par To see that the correspondence is biunivocal take $f\in Hom_{Com_{D}}(M,W\otimes D)$. If we apply \eqref{ddd} to $f$ we get $(Id_{W}\otimes \varepsilon_{D})\circ f$, and the allowed inverse gives us
 \begin{align*}
 & (((Id_{W}\otimes \varepsilon_{D})\circ f)\otimes Id_{D})\circ \rho_{M}=((Id_{W}\otimes \varepsilon_{D})\otimes Id_{D})\circ (f\otimes Id_{D})\circ \rho_{M}\\
 & =((Id_{W}\otimes \varepsilon_{D})\otimes Id_{D})\circ \rho_{W\otimes D}\circ f=(Id_{W}\otimes( \varepsilon_{D}\otimes Id_{D}))\circ (Id_{W}\otimes \Delta_{D}) \circ f\\
 & =(Id_{W}\circ Id_{W}\otimes (\varepsilon_{D}\otimes Id_{D})\circ \Delta_{D})\circ f= (Id_{W}\otimes Id_{D})\circ f = Id_{W\otimes D}\circ f=f,
 \end{align*}
 where we have used the property  $(f\otimes g)\circ(f'\otimes g')=(f\circ f')\otimes (g\circ g')$ in the first step, the fact that $f$ commutes with the coactions in the second equality, the definition of the coaction of $W\otimes D$ and the strictness of the category in the third equality and lastly the definition of coalgebra by mean of equation $(\varepsilon_{D}\otimes Id_{D})\circ \Delta_{D}=Id_{D}$ and the property  $(f\otimes g)\circ(f'\otimes g')=(f\circ f')\otimes (g\circ g')$ again. 
\end{proof}
   
\section{The category of sheaves and the topology of its sections}

\par For the basic definitions and  general references about sheaf theory we refer the reader to \cite{te} chapters 1 to 4.
\par Given a topological space $\mathcal{M}$ and a ringed space $(\mathcal{M},\mathcal{O})$, we shall denote by $_{\mathcal{O}}Mod$ the category of sheaves of $\mathcal{O}$-modules. If $F,G\in {_{\mathcal{O}}Mod}$, $F\otimes_{\mathcal{O}} G $ stands for the tensor product in the category $_{\mathcal{O}}Mod$. If $F\in {_{\mathcal{O}}Mod}$, we denote by $\tau_{\mathcal{M}}$ the topology of $\mathcal{M}$ and  by $r_{V\subseteq U}:F(U)\to F(V)$ the corresponding restriction map. $\Gamma(F)$ will denote the sections of the sheaf and, $\Gamma_{c}(F)$ will be sections of compact support. 

\par Given two sheaves of $\mathcal{O}$-modules $F$ and $G$ and a morphism $f$ from $F$ to $G$, we shall denote by  
 $f_{U}:F(U)\to G(U)$ the family of maps indexed by the open sets $U$ in $\mathcal{M}$ that form the morphism $f$.

\par If $\mathcal{M}$ is a smooth manifold, and $\tau_{\mathcal{M}}$ its topology, then the assignment $\mathcal{C}^{\mathbb{R}}:\tau_{\mathcal{M}}\rightarrow {_{\mathbb{R}}Alg} $ such that $\mathcal{C}^{\mathbb{R}}(U)=\mathcal{C}^{\infty}(U)=\mathcal{C}^{\infty}(U,\mathbb{R})$ is an example of sheaf of $\mathbb{R}$-algebras, where the restrictions maps are the usual restrictions of functions. Moreover the pair $(\mathcal{M}, \mathcal{C}^{\mathbb{R}})$ satisfy the definition of ringed space of $\mathbb{R}$-algebras. 
\par We refer the reader to \cite{hir} to general references about vector bundles, adapted coordinated system and related concepts. 
\par  Let $p:E\to \mathcal{M}$ be a complex finite dimensional vector bundle over a smooth manifold $\mathcal{M}$. Let us denote by $\Phi$ the assignment sending each $U\subseteq \mathcal{M}$ open to $\Phi(U)=\Gamma(U,E)=\{ \sigma\in \mathcal{C}^{\infty}(U,E) / p\circ \sigma=Id_{U} \}$. It is clear that $\Phi(U)$ is a module over $\mathcal{C}^{\infty}(U)$. Let be  $U,V\in \tau_{\mathcal{M}}$ such that $V\subseteq U$. We define restriction maps $\Phi(U)\to \Phi(V)$ by $r_{V\subseteq U}(\sigma)=\sigma|_{V}$ where $\sigma\in \Phi(U)$. It is clear that they satisfies the axiom of presheaf.

Moreover the restriction map $r_{V\subseteq 
U}:\Phi(U)\rightarrow \Phi(V)$ is a $\mathcal{C}^{\infty}(U)$-linear morphism, because it is linear and if $f\in \mathcal{C}^{\infty}(U)$ and $\sigma \in \Phi(U)$, then
$r_{V\subseteq U}(f\sigma)=(f\sigma)|_{V}=f|_{V}\sigma|_{V}=r_{V\subseteq U}(f)r_{V\subseteq U}(\sigma)$, where $r_{V\subseteq U}(f)$ are the restriction maps of the aforementioned sheaf $\mathcal{C}^{\mathbb{R}}$. This proves that $\Phi$ is a presheaf of modules over the ringed space $(\mathcal{M},\mathcal{C}^{\mathbb{R}})$. 

 \par Furthermore, the following is also true, whose proof is immediate.
 
\begin{lemma}\label{lemshe} Let $\mathcal{M}$ be a smooth manifold and $p:E\to \mathcal{M}$ a finite dimensional complex vector bundle over $\mathcal{M}$. Then the above mentioned assignment\\ $\Phi:\tau_{\mathcal{M}}\to {_{\mathbb{R}}Vect}$ is a sheaf of $\mathcal{C}^{\mathbb{R}}$-modules.
\end{lemma}

 \subsection{Sheaf of jets and sheaf of densities}\label{jets}
 
 \begin{definicion}
 Consider a finite dimensional complex vector bundle $(E,p,\mathcal{M})$ and let $q\in \mathcal{M}$. Two local sections $\phi, \psi\in \Gamma_{q}(p)$ around $q$ are said to be \emph{k-equivalent at $q$} if  $\phi(q)=\psi(q)$ and if in some adapted  coordinate system $(x^{i},u^{\alpha})$ around $\phi(p)$, we have that
 \begin{equation}
 \frac{\partial^{|I|} \phi^{\alpha}}{\partial x^{I}}|_{q}=\frac{\partial^{|I|} \psi^{\alpha}}{\partial x^{I}}|_{q} 
 \end{equation}
 for each $1\leqslant |I|\leqslant k$ and $1\leqslant \alpha \leqslant Rank(E)$. This is an equivalence relation on the local sections around $q$, and the equivalent class containing $\phi$ is called the \emph{k-jet of $\phi$ at $q$} and is denoted by $j^{k}_{q}\phi$
 \end{definicion} 
 
  Given $q\in \mathcal{M}$, define the set of \emph{k-th jets of $p$} at $q$ $$J^{k}p=\{j^{k}_{q}\phi \hspace{.2cm}  | \hspace{.2cm}  q\in \mathcal{M}\text{ and } \phi \text{ is a local section of }p\text{ arround } q \}$$
The \emph{k-th jet manifold of $p$} is the disjoint union $\sqcup_{q\in \mathcal{M}}J_{q}^{k}p$. Lots of properties of this set can be found in \cite{sa}. The most important to us are the following two.

\begin{lemma}
[See \cite{sa}, Proposition 4.1.7] $J^{k}p$ is a smooth finite dimensional manifold for any $k\in \mathbb{N}_{0}$. 
\end{lemma}

\par For the manifold $J^{k}p$ we define the projection $\pi:J^{k}p\to \mathcal{M}$ such that $\pi(j^{k}_{q}\phi)=q$.

\begin{lemma}
[See \cite{sa}, Proposition 6.2.13] \label{lema5}  Let $(E,p,\mathcal{M})$ be a finite dimensional vector bundle. Then $(J^{k}p,\pi_{k},\mathcal{M})$ is a finite dimensional vector bundle for $k\in \mathbb{N}_{0}$. 
\end{lemma}

\par For a fixed $k\in \mathbb{N}$, the last lemma enables us to consider sections of the finite dimensional vector bundle $(J^{k}p,\pi_{k},\mathcal{M})$ of jets of order $k$ which by Lemma \ref{lemshe} is a sheaf of modules over the ringed space $(\mathcal{M}, \mathcal{C}^{\mathbb{R}})$. We shall denote this sheaf by $J^{k}\Phi$. 
Analogously one define the jets of infinite order over a fibre bundle $p:E\to \mathcal{M}$ (see \cite{sa}, Chapter 7) which we denote by $J^{\infty}\Phi$ and also conform a vector bundle over $\mathcal{M}$, whose projection is called $\pi$, the only difference is that the fibres are infinite dimensional. 
\par All along this thesis when we write $J\Phi$ we will refer to $J^{k}$ for $k\in \mathbb{N}\cup \{ \infty \}$. 

\par The category of sheaves over a sheaf of commutative algebras is a symmetric monoidal category which satisfies  $Ab3^{\ast}$ (even $Ab4^{\ast}$) axioms and $\otimes_{\mathcal{O}}$ commutes with $\oplus$ (see \cite{te}, Chapter 4). Therefore, given any sheaf $F\in {_{\mathcal{O}}Mod}$, we define $SF$ the symmetric sheaf of $F$ as 
in Remark \ref{symme}.


\par We shall also use the notion of \emph{sheaf of densities} of a smooth manifold.   
The reader is refereed to \cite{bo}, Chapter 7, or \cite{ma}, Chapter 1, for the definition and basic properties. Given a smooth manifold $\mathcal{M}$, we shall denote by $\omega$ the sheaf of $1$-densities over $\mathcal{M}$. We recall it is the sheaf of sections of a complex vector bundle over $\mathcal{M}$, so it has structure of $\mathcal{C}^{\mathbb{R}}$-module in the natural way. 
\ 

\par Given a vector bundle $p:E\to \mathcal{M}$ and we have just defined the sheaf $J\Phi$ of jets and its symmetric sheaf $SJ\Phi$, also we have the sheaf of densities $\omega$ and both have   
structure of $\mathcal{C}^{\mathbb{R}}$-modules. Then we consider its tensor product $\omega\otimes_{\mathcal{C}^{\mathbb{R}}}SJ\Phi$ in the category of  $\mathcal{C}^{\mathbb{R}}$-modules and denote it only by $\omega SJ\Phi$. 

\subsection{Topology on the spaces of sections}

\par We refer the reader to \cite{ru} or \cite{he} for general references about locally convex spaces (LCS).
\par  We will now describe the topology of several spaces of sections of sheaves. We begin with  $\Gamma_{c} \omega SJ\Phi$, following the steps of \cite{di} Chapter  17, Section 2.
\par 
Let $\Omega$ be an open set of $\mathbb{R}^{n}$ and consider a fundamental sequence of compact sets  $\{K_{m}\}_{m\in \mathbb{N}}$, i.e. $K_{m}\subseteq K_{m+1}^{\circ}$ for all $m\in \mathbb{N}$ and $\bigcup_{m\in\mathbb{N}}K_{m}= \Omega$. Denote by $\mathcal{C}_{m}^{\infty}(\Omega)$ the subspace of $\mathcal{C}^{\infty}(\Omega)$ formed by functions whose support is contained in $K_{m}$. One can define a family of \emph{seminorms} on $\mathcal{C}^{\infty}(\Omega)$ by 
\begin{equation*}
p_{m}(f)=\max \{ |D^{\alpha}f(x)|/x\in K_{m}, |\alpha|\leqslant m \}
\end{equation*}
 that turns $\mathcal{C}^{\infty}(\Omega)$ into a \emph{Fréchet} space (see \cite{ru}, 1.46). It is easy to prove that $\mathcal{C}_{c}^{\infty}(\Omega)$ is a closed subspace of $\mathcal{C}^{\infty}(\Omega)$. 

\par We regard $\mathcal{C}_{m}^{\infty}(\Omega)$ as a topological subspace of $\mathcal{C}^{\infty}(\Omega)$ with its Fréchet topology. Finally we endow $\mathcal{C}^{\infty}_{c}(\Omega)$ with the final LCS topology from the family of inclusions  $\{i_{m}:\mathcal{C}_{m}^{\infty}(\Omega) \rightarrow \mathcal{C}^{\infty}_{c}(\Omega)\}$. One can see that with this topology on $\mathcal{C}^{\infty}_{c}(\Omega)$ is a complete LCS.
\par The space $\mathcal{C}_{c}^{\infty}(\Omega)$ is an example of \emph{LF-space}, and we shall topologize the space of compact supported sections in a very similar way as a LF-space.

 
 \par We will now recall the topology on $\Gamma_{c} \omega SJ\Phi$. We will use the following diagram
 
\begin{equation}\label{diag}
\xymatrix{
  \omega S^{l}J\Phi \ar_{\tau}@/_4mm/[d] & \tau^{-1}(U_{\alpha}) \ar_{inc}[l] \ar^{\omega_{\alpha}^{l}}[r] \ar_{\tau |}@/_4mm/[d]  & x_{\alpha}(U_{\alpha})\times S^{l}E \ar^{\pi_{1}}[ddl] \ar^{\pi_{2}}[d] \\
 \mathcal{M} \ar_{u}@/_4mm/[u] & U_{\alpha} \ar^{x_{\alpha}}[d] \ar_{u|}@/_4mm/[u] \ar^{inc}[l] & S^{l}E  \ar_{i_{l}}[r] & SE \\
  \mathbb{R}^{n} & x_\alpha(U_{\alpha}) \ar^{inc}[l]
 }
\end{equation} 
where $\tau$ is the projection of the vector bundle $\omega S^{l}J\Phi$  and $\omega_{\alpha}^{l}$ is a local trivialization of such bundle over the domain $U_{\alpha}$  for $l\in \mathbb{N}$. By reducing the domains if necessary one can always think that the domains of the charts are trivializants for the vector bundle.


 
 \subsubsection{Topology for the fibres of the vector bundle $J\Phi$}
 
 \par We began for define a topology on the fibres of the vector bundle $J\Phi$. This fibres are typically the space $\prod_{k\in \mathbb{N}_{0}}L^{k}_{sym}(\mathbb{C}^{n},\mathbb{C}^{m})$ where $L_{sym}^{k}$ denotes the space of symmetric $k$-multilinear maps from $\mathbb{C}^{n}$ to $\mathbb{C}^{m}$. Let us denote this fibre by $E$, as we did in the diagram  (\ref{diag}). In \cite{le}, Lewis introduced a family of seminorms $\{\lambda_{r}\}_{r\in \mathbb{N}}$ on the space $\prod_{k\in   \mathbb{N}_{0}}L^{k}_{sym}(\mathbb{C}^{n},\mathbb{C}^{m})$ which turns it into a Frechét space. More precisely if $A\in \prod_{k\in   \mathbb{N}_{0}}L^{k}_{sym}(\mathbb{C}^{n},\mathbb{C}^{m})$ and we write $A=\prod_{k\in \mathbb{N}_{0}}A_{k}$ then $\lambda_{r}$ is given by $$\lambda_{r}(A)=\max \{ ||A_{0}||_{0}',||A_{1}||_{1}',\cdots, ||A_{r}||_{r}' \},$$ where $||-||_{j}'$ are norms on $L^{j}_{sym}(\mathbb{C}^{n},\mathbb{C}^{m})$ and we recall that, if $M\in L^{j}_{sym}(\mathbb{C}^{n},\mathbb{C}^{m})$, then $$||M||_{j}'=\sup\{ ||M(v,\cdots,v)|| / v\in \mathbb{C}^{n}, ||v||=1 \}.$$

 
 
 \subsubsection{Possibles topologies for the tensor product of LCS}\label{1.2.2}
 
 \par  We shall use the seminorms $\lambda_{r}$ to define  the \emph{projective topology} over  $T^{2}E$ following which was done at \cite{he} Chapter 3, Section 6. Given two LCS $F$ and $E$, the projective topology for $F\otimes E$ is the finest LCS topology such that the universal mapping $\chi: F\times E \to F \otimes E$ is continuous. 
 We recall that in a LCS there is a correspondence between the Minkowski functionals and the zero open neighbourhoods. Indeed, given $U$ an open neighbourhood, set $p_{U}(x)=\inf \{t\in \mathbb{R}_{>0}/ x\in tU \}$, called the \emph{associated Minkowski functional}. Conversely, given a Minkowski functional $p$, then  $p^{-1}([0,1))$ is an open neighbourhood of zero. Given two LCS $E$ and $F$, the projective topology of $E\otimes F$ can be explicitly defined by its Minkowski functionals. If $\chi:E\times F\rightarrow E\otimes F$ is the universal mapping, then a base of zero neighbourhoods for  $E\otimes F$  is the balanced  convex hull of   
 \begin{equation*} 
  \{ \chi(U\times V) / U \text{ is a open neighbourhood of zero in } E \text{ and V in }F \}
\end{equation*}
which is denoted by  $CH(U\otimes V)$. Moreover if $p_{U}$ is the Minkowski functional associated to $U\subseteq E$ and  $p_{V}$ of $V\subseteq F$, then 
 \begin{equation}\label{semi}
 p_{CH(U\otimes V)}(u)=\inf \{ \sum_{i}p_{U}(x_{i})p_{V}(y_{i}) / u=\sum_{i} x_{i}\otimes y_{i}\}
 \end{equation} 
 is a seminorms on $E\otimes F$ which coincide with the  Minkowski functional associated to $CH(U\otimes V)$  (see result 6.3. from \cite{he}, chapter 3, section 6.). We denote it by $p_{U}\otimes p_{V}$.
\par We denote by $E \otimes_{\pi} F$ the algebraic tensor product with the projective topology just described and call it the projective tensor product between $E$ and $F$. If $E$ and $F$ are LCS then also is  $E \otimes_{\pi} F$. But with this topology if $E$ and $F$ are Fréchet spaces then  $E \otimes_{\pi} F$ is not necessarily a Fréchet space and the same holds for $LF$ (inductive strict limit of Fréchet, whose principal example are the compactly supported sections over a manifold).   
\par Denoting by $E \hat{\otimes}_{\pi} F$ the completion of  $E \otimes_{\pi} F$, we obtain a categorical definition of tensor product in the category of Fréchet spaces or LF.
 
 
 \par In both cases, if $E$ is a LCS  by recursive applications of the LCS structure on the tensor products give us a family of countable seminorms on $T^{l}E$ for $l\in \mathbb{N}$ (that may be is a Fréchet topology, depending if $E$ is Fréchet and if we take $\otimes_{\pi}$ or $\hat{\otimes}_{\pi}$).  We note them by $||-||_{T^{l}E,\beta}$ with $\beta \in \mathbb{N}_{0}$.
 

 
  \subsubsection{Topologies for the symmetric product $S^{l}E$} \label{jhjhjh}
  
 \par Let us consider the linear inclusion of $S^{l}E$ inside $T^{l}E$ given by the symmetrization map  $x_{1}\cdots x_{l}\mapsto \sum_{\sigma\in \mathbb{S}_{l}} x_{\sigma(1)}\otimes\cdots\otimes x_{\sigma(l)}$. It is easy to show that it is a closed subspace of $T^{l}E$. Then, we may consider the induced LCS topology on $S^{l}E$ by means of this inclusion. We shall denote it by $(T^{l}E)^{\mathbb{S}_{l}}$. If $E$ is a Fréchet space, note that $T^{l}E$ is also a Fréchet space and the same holds for the closed subspace $(T^{l}E)^{\mathbb{S}_{l}}$, provided we are using the completion of the projective topology for the algebraic tensor product.
 
 
 \par On the other hand, since $S^{l}E$ is a quotient of $T^{l}E$, the LCS topology on the latter induces a LCS topology on the former (see \cite{top}, section 3). If $E$ is a Fréchet space, then $S^{l}E$ is also.
 
 \par Note that the composition  $(T^{l}E)^{\mathbb{S}_{l}}\hookrightarrow T^{l}E \twoheadrightarrow S^{l}E$ is clearly linear and continuous, because the inclusion is linear and continuous by definition of subspace topology at $(T^{l}E)^{\mathbb{S}_{l}}$ and the quotient projection is clearly linear and continuous. 
  Also, one can see that the composition is bijective. If $E$ is Fréchet the open map theorem for Fréchet spaces (see \cite{ru} chapter 1), implies that the inverse map of this composition is also continuous and as a conclusion $S^{l}E$ and $(T^{l}E)^{\mathbb{S}}$ are homeomorphic. Note that the open map theorem can be used only if the LCS involved are Fréchet spaces and it is true only if we use the tensor product $\hat{\otimes}_{\pi}$.
  
  \subsubsection{Adding the fibre of the density bundle}
  
 \par Continuing the construction of the topology for the section space given the vector bundle  $\tau: \omega S^{l}J\Phi \rightarrow \mathcal{M} $ we denote by  $\omega_{\alpha}^{l}:\tau^{-1}(U_{\alpha})\rightarrow x_{\alpha}(U_{\alpha}) \times S^{l}E$ as its local trivialization. Notice that this vector bundle has the same local fibre as $S^{l}J\Phi$, because the tensor product on the fibres are over $\mathbb{R}$ and $\omega$ is a line bundle, the typical fibre of $S^{l}J\Phi$ is $S^{l}E$ and for  $\omega S^{l}J\Phi$ is $\mathbb{R} \otimes_{\mathbb{R}}S^{l}E$. 
 
 \subsubsection{Topology on $\Gamma \omega SJ\Phi$ and $\Gamma_{c} \omega SJ\Phi$}
 \par For each chart $(U_{\alpha},x_{\alpha})$ of the aforementioned atlas of $\mathcal{M}$ we consider a fundamental sequence of compacts  $ \{K_{m}^{\alpha}\}_{m\in \mathbb{N}}$ 
 and denote by $\Gamma_{m} \omega S^{l}J\Phi$ the sections with compact support contained in $K_{m}$.



\par Given $n\in \mathbb{N}$, define the following family of seminorms on $\Gamma_{m}\omega S^{l}J\Phi$ (see \cite{di}, chapter XV, section 2) defined as

 \begin{align*}
  & p_{s,\beta,\alpha,m}(u)=\sup _{x\in K_{m}\subseteq U_{\alpha}}\sum_{|j|\leqslant s,j\in \mathbb{N}_{o}^{n}}||\partial^{j}(\pi_{2}\circ \omega_{\alpha}^{l}\circ u \circ x_{\alpha}^{-1})(x)||_{S^{l}E,\beta},
 \end{align*}
where  $s$ indicates the maximum order of derivation, $\alpha$ is the index of the chart and $\beta$ the index of the seminorm of $S^{l}E$. 
  

Note that the (partial) derivatives appearing in the previous definition are the classical ones (of any function of the form $f:W\subseteq \mathbb{R}^{n}\to L$, where $L$ is a Fréchet space). 
  
\par We finally consider the final LCS topology on $\Gamma_{c} \omega S^{l}J\Phi$  given by the family of inclusions  $i_{m,l}:\Gamma_{m} \omega S^{l}J\Phi \hookrightarrow \Gamma_{c} \omega S^{l}J\Phi $, and on $\Gamma_{c} \omega SJ\Phi$ with the LCS topology of the direct sum, i.e. the final LCS topology given by the inclusions  $i_{l}:\Gamma_{c} \omega S^{l}J\Phi\rightarrow \Gamma_{c} \omega SJ\Phi$.  With this final topology  $\Gamma_{c} \omega S^{l}J\Phi$ is a LF-space.
  


\section{Distributions on manifolds}

\subsection{H-distributions \& D-distributions }

\par We refer the reader to \cite{hor}, \cite{di}, \cite{ma} and \cite{mich} for general references about the theory of distributions on manifolds. 
During this section $X$ and $Y$ will be open sets of $\mathbb{R}^{n}$, and  $\mathcal{M}$ will be a manifold of dimension $n$. We denote by $\mathcal{D}'(X)$ the space of distributions on $X\subseteq \mathbb{R}^{n}$ (see \cite{hor} for a definition). Finally we denote by $H_{d}(\mathcal{M})$ to the distributions with  density character $0$ on $\mathcal{M}$ and by $D_{d}(\mathcal{M})$ to the distributions with \emph{density character} $1$ over $\mathcal{M}$ (see \cite{mich} for a reference). We will call these two types of distributions H-distributions and D-distributions respectively and shall present a formal definition later in this subsection. 

\par We recall that a subset $V$ of $\mathbb{R}^{n}$ is said  to be \emph{conic} if for each $\xi\in V$ and $t>0$, $t\xi$ belongs to $V$. If $\mathcal{E}'(\mathbb{R}^{n})$ are the distributions of  compact support then for any $u\in\mathcal{E}'(\mathbb{R}^{n})$ we define the set $\Sigma(u)\subseteq \mathbb{R}^{n}\setminus \{ \vec{0} \}$ as the points having no conic neighbourhood $V\subseteq \mathbb{R}^{n}\setminus \{ \vec{0} \}$ and positive constants $C_{N}$ such that \begin{align*}
 & |\hat{u}(\xi)|\leqslant C_{N}(1+|\xi|)^{-N} 
 \end{align*} holds for all $\xi \in V$ and all  $N\in \mathbb{N}$, where $\hat{u}$ is the Fourier transform of $u$. With this definition $\Sigma(u)$ is clearly a closed cone in $\mathbb{R}^{n}\setminus \{ \vec{0} \}$. 
 \par For any  $u\in \mathcal{D}'(X)$ and $\psi \in \mathcal{C}_{c}(X)$, the distribution $\psi u$ belongs to $\mathcal{E}'(X)$. Given $u\in \mathcal{D}'(X)$, for each $x\in X$ we set \begin{align*}
 \Sigma_{x}(u)=\bigcap \{ \Sigma(\phi u) / \phi\in \mathcal{C}_{0}^{\infty}(X), \phi(x)\neq 0 \}
 \end{align*} and call it the \emph{cone of $u$ at $x$}. For a nice exposition of these concepts, and together with a definition of the wavefront  set of a distribution on $\mathbb{R}^{n}$ and also the definition of wavefront set for an H-distribution on a manifold, see \cite{ma}, Chapter 2.
 
\par The objective of this section is to generalise the product of  H-distributions defined by L.H\"{o}rmander in \cite{hor} (see \cite{ma}, Chapter 2.3) to an action of the space of   H-distributions over the space of D-distributions. 
It will be very useful to describe $H_{d}(\mathcal{M})$ and $D_{d}(\mathcal{M})$ as families of distributions in open sets of $\mathbb{R}^{n}$ satisfying some sort of covariance, which we briefly recall.  
\par If  $f:X\rightarrow Y$ is a diffeomorphism, the pullback of  $f$,  noted by   $f^{\ast}:\mathcal{D}'(Y)\rightarrow \mathcal{D}'(X)$, is defined as
\begin{equation}\label{227}
f^{\ast}(w)(\phi)=w(|\det(f^{-1})'|\phi\circ f^{-1}), 
\end{equation}
where $w\in \mathcal{D}'(Y)$ and $\phi$ is a test function in $X$, i.e. $\phi\in \mathcal{C}^{\infty}_{c}(X)$.
\par Take a differentiable structure on $\mathcal{M}$ given by charts $(U_{\alpha}, \varphi_{\alpha})_{\alpha\in \mathbb{N}}$. Then $\{\varphi_{\alpha}(U_{\alpha})\}_{\alpha\in \mathbb{N}}$ is a family of open sets  of $\mathbb{R}^{n}$, so it makes sense to consider $\mathcal{D}'(\varphi_{\alpha}(U_{\alpha}))$.
\begin{definicion}\label{def12}
The space $H_{d}(\mathcal{M})$ is the set whose elements are families of the form $u=\{u_{\alpha}\}_{\alpha\in \mathbb{N}}$, where $u_{\alpha}\in \mathcal{D}'(\varphi_{\alpha}(U_{\alpha}))$, obeying the equality
\begin{equation}\label{228}
u_{\beta}=(\varphi_{\alpha} \circ \varphi_{\beta}^{-1})^{\ast}u_{\alpha}
\end{equation} 
of distributions on $\varphi_{\beta}(U_{\alpha}\cap U_{\beta})$.

\end{definicion}

\par Analogously, if $f$ is a diffeomorphism, the \emph{semi-pullback} is given by  
\begin{equation}\label{229}
f^{\bullet}(w)(\phi)=w(\phi\circ f^{-1}),
\end{equation}
where  $w\in \mathcal{D}'(Y)$ and $\phi$ is a test function in $X$, i.e. $\phi\in \mathcal{C}^{\infty}_{c}(X)$.
Notice that
\begin{equation}\label{2210}
f^{\bullet}(w)=|\det(f')|f^{\ast}(w),
\end{equation}
for all  $w\in \mathcal{D}'(Y)$

\begin{definicion}[see \cite{di}]\label{def13}
The space $D_{d}(\mathcal{M})$ is the set whose elements are  the families $u=\{u_{\alpha}\}_{\alpha\in \mathbb{N}}$, where $u_{\alpha}\in \mathcal{D}'(\varphi(U_{\alpha}))$, satisfying that
\begin{equation}\label{2211}
u_{\beta}=(\varphi_{\alpha} \circ \varphi_{\beta}^{-1})^{\bullet}u_{\alpha}
\end{equation} 
as distributions on  $\varphi_{\beta}(U_{\alpha}\cap U_{\beta})$.
\end{definicion}

\par Note that \eqref{2210} tells us that equation  \eqref{2211} can be rewritten as 
\begin{equation}\label{212}
u_{\beta}=|\det(\varphi_{\alpha} \circ \varphi_{\beta}^{-1})'|(\varphi_{\alpha} \circ \varphi_{\beta}^{-1})^{\ast}u_{\alpha}.
\end{equation}

Definitions \ref{def12} and \ref{def13} are particular examples of what is called \emph{distribution with density character  $q$} (see \cite{mich} Chapter 3, Section 1, where they are also described as continuous dual spaces of sections of certain density bundles).\\


\par We will extend the definition of wavefront set of H-distributions to D-distributions. We remark that the definition of wavefront set of H-distribution was introduced by L. H\"{o}rmander in \cite{hor} (see \cite{ma} for a short exposition).

\begin{lemma}\label{wed}
Let $f\in \mathcal{C}^{\infty}(W)$ positive, $W\subseteq R^{n}$ open and $\mu\in \mathcal{D}'(W)$. Hence
$$ \Sigma_{x}(\mu)=\Sigma_{x}(f\mu) \hspace{2cm}$$ 
for all $x\in W$.
\end{lemma}
\begin{proof} We take $\psi \in \mathcal{C}^{\infty}_{c}(W)$ such that  $\psi$ takes the value $1$ in a closed neighbourhood  of $x$ and  $0$ in the complement of an open set containing that closed set.  We also consider $\varphi\in \mathcal{C}^{\infty}_{c}(W)$ such that  $\varphi(x)\neq 0$. Then, $\varphi \mu \in \mathcal{D}'(W)$ has compact support. If we apply \cite{ma}, Lemma 2.5, to the compact support distribution  $\varphi \mu$ and the smooth compact support function $f\psi$, we see that  $\Sigma((f\psi) (\varphi \mu))\subseteq \Sigma(\varphi \mu)$. Since  $\Sigma((f\psi)(\varphi \mu))= \Sigma((\varphi\psi) (f\mu))$, intersecting the above contention for all $\varphi\in \mathcal{C}^{\infty}_{c}(W)$ such that  $\varphi(x)\neq 0$, we get that   $\bigcap_{\varphi}\Sigma((\varphi\psi) (f \mu))\subseteq \Sigma_{x}(\mu)$. On the other hand, as $\varphi=\varphi\psi$ on a neighbourhood of $x$ and $\Sigma_{x}$ is a local property, we also conclude that $\Sigma_{x}(f\mu)\subseteq \Sigma_{x}(\mu)$.
 \par We apply the argument of the previous paragraph to the positive function  $\phi=\frac{1}{f}$ and the distribution $u=f\mu$, to deduce that  
 \begin{align*}
& \Sigma_{x}(\phi u)\subseteq \Sigma_{x}(u),\\
& \Sigma_{x}(\frac{1}{f} f\mu )\subseteq \Sigma_{x}(f\mu),\\
& \Sigma_{x}(\mu)\subseteq \Sigma_{x}(f\mu).
 \end{align*}
 \par As a consequence $\Sigma_{x}(\mu)=\Sigma_{x}(f\mu)$, which is what we want. 
 \end{proof}

\par Equation (2.133) of \cite{ma}, tells us that 
\begin{equation}\label{qqq}
\Sigma_{x(p)}(u_{x})= \Sigma_{x(p)}(|\det(y\circ x^{-1})'|(y\circ x^{-1})^{\ast}(u_{y}))= \Sigma_{x(p)}((y\circ x^{-1})^{\ast}(u_{y})),
\end{equation} where $(x,U_{x})$ is an atlas for $\mathcal{M}$, $p\in \mathcal{M}$ and  we have used equation \eqref{212} and Lemma \ref{wed}.

\par As we can see in \cite{ma}, equation (2.111), each diffeomorphism $f:X\to Y$ satisfies \begin{equation}\label{kkk}
\Sigma_{x}(f^{\ast}u)=f'(x)^{T}\Sigma_{f(x)}(u)
\end{equation} for any $x\in X$, where $f'(x)$ is the differential of $f$. Now if $u=\{u_{z}\}_{(z,U_{z})}$ is a D-distribution on $\mathcal{M}$ and $p\in \mathcal{M}$ is such that the charts $(x,U_{x})$ and $(y,U_{y})$ satisfy that $p\in U_{x}$  and $p\in U_{y}$, then using the Einstein's summation convention, we get that 
\begin{align*}
& \{ \xi_{k}dx^{k}_{p}/ \xi\in \Sigma_{x(p)}(u_{x}) \}= \{ \xi_{k}dx^{k}_{p}/ \xi\in \Sigma_{x(p)}((y\circ x^{-1})^{\ast}u_{y}) \}\\
& = \{ ((y\circ x^{-1})'(x(p))^{T}\eta)_{k}dx^{k}_{p} / \eta\in \Sigma_{y(p)}(u_{y}) \}= \{ ((\frac{\partial y^{j}}{\partial x^{k}})_{p}\eta_{j}dx^{k}_{p} / \eta\in \Sigma_{y(p)}(u_{y})\} \\
& =\{ \eta_{j}dy^{j}_{p}/ \eta\in \Sigma_{y(p)}(u_{y}) \},
\end{align*}
where we have used equations \eqref{qqq} and \eqref{kkk} in the first equality and in the second one respectively. This allow us to introduce the following notion.
  

\begin{definicion}\label{wfd}
Let $u\in D_{d}(\mathcal{M})$. The wave front set of $u$ is the set $$ WF(u)= \{ \{p\}\times \Sigma_{p}(u) / p\in \mathcal{M} \}\subseteq T^{\ast }\mathcal{M},$$
where $ \Sigma_{p}(u)=\{ \xi_{k}dx^{k}_{p}/ \xi\in \Sigma_{x(p)}(u_{x}) \}$ for any chart $(x,U_{x})$ such that $p\in U_{x}$, and $T^{\ast}\mathcal{M}$ is the cotangent bundle of the manifold.
\end{definicion}



\par The pullback of distributions defined on open sets of $\mathbb{R}^{n}$ is also defined for maps which are not diffeomorphisms.
\begin{thm}[\cite{ma}, Theorem 2.61]\label{hhhh}
Let $X$ and $Y$ be open subsets of $\mathbb{R}^{m}$ and $\mathbb{R}^{n}$ respectively and let $f:X\to Y$ be an smooth map. Define
$$ N_{f}=\{ (f(x),\nu)\in Y\times \mathbb{R}^{n} / x\in X, f'(x)^{T}\nu =\vec{0} \}. $$
Then there is a unique way of defining the pullback
$$ f^{\ast}:\{ u\in \mathcal{D}'(Y) / N_{f}\cap WF(u)=\emptyset \} \to \mathcal{D}'(X),$$
such that $f^{\ast}u=u\circ f$ for all $u\in \mathcal{C}^{\infty}(Y)$.
\end{thm}

\par We recall that given $u\in \mathcal{D}'(X)$ and $v\in \mathcal{D}'(Y)$, the exterior tensor product $u\otimes v$ is a distribution over $X\times Y$ (see for instance \cite{ma}, Definition 1.48).  
The hypotheses of Theorem \ref{hhhh} are satisfied for $\Delta:\mathbb{R}^{n}\rightarrow \mathbb{R}^{n} \times \mathbb{R}^{n}$  and a distribution of the form $\mu\otimes \nu$ if there is no $(x,\xi)\in WF(\mu)$ such that  $(x,-\xi)\in WF(\nu)$ (as explained in \cite{ma}, Theorem 2.167) where $\mu $ and $\nu$ are distributions over $\mathbb{R}^{n}$.

\par The family  $\{\Delta_{\alpha}^{\ast}(u_{\alpha}\otimes v_{\alpha})\}_{\alpha\in \mathbb{N}}$ will be of interest to us: these are distributions $ \Delta_{\alpha}^{\ast}(u_{\alpha}\otimes v_{\alpha})\subseteq \mathcal{D}'(\varphi_{\alpha}(U_{\alpha}))$ where $\Delta_{\alpha}:\varphi_{\alpha}(U_{\alpha})\rightarrow \varphi_{\alpha}(U_{\alpha})\times \varphi_{\alpha}(U_{\alpha})\subseteq \mathbb{R}^{2n}$ are the diagonal maps. 
\par The following will be useful.
\begin{lemma}\label{lema6}
 Given $f:Z\rightarrow X$ and $g:W\rightarrow Y$ diffeomorphisms between open sets of  $\mathbb{R}^{n}$, let $f\times g: Z\times W\rightarrow X\times Y$ be defined by $(f\times g)(z,w)=(f(z),g(w))$. Then 
\begin{equation}
(f\times g)^{\ast}(u\otimes v)(\varphi \otimes \psi )= (f^{\ast}(u)\otimes g^{\ast}(v))(\varphi \otimes \psi),
\end{equation}
 for all $u\in \mathcal{D}'(X)$, $v\in \mathcal{D}'(Y)$, $\varphi\in \mathcal{C}^{\infty}_{c}(Z)$ and $\psi\in \mathcal{C}^{\infty}_{c}(W)$. Since  $\mathcal{C}^{\infty}_{c}(Z) \otimes  \mathcal{C}^{\infty}_{c}(W)\subseteq  \mathcal{C}^{\infty}_{c}(Z\times W)$ is dense and the pullback is sequentially continuous (see \cite{ma}), we can conclude that
 \begin{equation}
 (f\times g)^{\ast}(u\otimes v)= f^{\ast}(u)\otimes g^{\ast}(v)
 \end{equation}
 as elements of $\mathcal{D}'(Z\times W)$.
 \end{lemma}
 \begin{proof}
 First of all we take sequences  $\{u_{j}\}_{j\in \mathbb{N}}$ and $\{v_{m}\}_{m\in \mathbb{N}}$ of continuous functions over $X$ and $Y$, respectively, such that $u_{n}\to u$ and  $v_{n}\to v$. We recall that $\mathcal{D}'(X)$ has the weak$^{\star}$-topology induced by $\mathcal{C}^{\infty}_{c}(X)$, and the same holds for $Y$. Then,
 $$ ((f\times g)^{\ast}(u\otimes v))(\varphi\otimes\psi)=\lim_{l\to \infty}((f\times g)^{\ast}(u_{l}\otimes v_{l}))(\varphi\otimes\psi), $$
 because the pullback is sequentially continuous. Hence,
 $$ \lim_{l\to \infty}((f\times g)^{\ast}(u_{l}\otimes v_{l}))(\varphi\otimes\psi)=\lim_{l\to \infty}\int_{Z\times W}((u_{l}\otimes v_{l})\circ (f\times g))(z,w)\cdot (\varphi\otimes\psi)(z,w) dz dw, $$
 for the pullback map is an extension of the precompositions from functions to distributions. Moreover, by definition we see that
 $$\lim_{l\to \infty}\int_{Z\times W}(u_{l}\circ f(z)\cdot \varphi(z))\cdot((v_{l}\circ g)(w)\cdot\psi(w)) dz dw, $$
 which by Fubinni's Theorem gives us
 $$\lim_{l\to \infty}[\int_{Z}(u_{l}\circ f(z)\cdot \varphi(z))dz\cdot\int_{W}((v_{l}\circ g)(w)\cdot\psi(w))dw]$$
$$ = \lim_{l\to \infty}\int_{Z}(u_{l}\circ f(z)\cdot \varphi(z))dz\cdot \lim_{l\to \infty}\int_{W}((v_{l}\circ g)(w)\cdot\psi(w))dw$$
 $$= \lim_{l\to \infty} f^{\ast}u_{l}(\varphi) \cdot \lim_{l\to \infty} g^{\ast}v_{l}(\psi)=f^{\ast}u(\varphi)\cdot g^{\ast}v(\psi)=(f^{\ast}u \otimes g^{\ast}v)(\varphi\otimes \psi).$$
 The lemma is thus proved.
 \end{proof}

\par Consider now the pullback of $\Delta_{\alpha}^{\ast}(u_{\alpha}\otimes v_{\alpha})\in \mathcal{D}'(\varphi_{\alpha}(U_{\alpha}\cap U_{\beta}))$ by  $\varphi_{\alpha} \circ \varphi_{\beta}^{-1}$. Note that the following diagram 
$$
\xymatrix{
\varphi_{\beta}(U_{\alpha}\cap U_{\beta}) \ar[rrr]^{\Delta_{\beta}} \ar[dd]_{\varphi_{\alpha} \circ \varphi_{\beta}^{-1}} & & & \varphi_{\beta}(U_{\alpha}\cap U_{\beta})\times \varphi_{\beta}(U_{\alpha}\cap U_{\beta}) \ar[dd]^{(\varphi_{\alpha} \circ \varphi_{\beta}^{-1})\times (\varphi_{\alpha} \circ \varphi_{\beta}^{-1})} \\
& & & \\
\varphi_{\alpha}(U_{\alpha}\cap U_{\beta}) \ar[rrr]^{\Delta_{\alpha}} & & & \varphi_{\alpha}(U_{\alpha}\cap U_{\beta})\times \varphi_{\alpha}(U_{\alpha}\cap U_{\beta}) 
}
$$
is commutative. As a consequence,
$$(\varphi_{\alpha} \circ \varphi_{\beta}^{-1})^{\ast}(\Delta^{\ast}_{\alpha}(u_{\alpha} \otimes v_{\alpha}))=
(\varphi_{\alpha} \circ \varphi_{\beta}^{-1})^{\ast}\circ\Delta^{\ast}_{\alpha}(u_{\alpha} \otimes v_{\alpha})=
(\Delta_{\alpha}\circ\varphi_{\alpha} \circ \varphi_{\beta}^{-1})^{\ast}(u_{\alpha} \otimes v_{\alpha})$$
\begin{equation}\label{lll}
=[((\varphi_{\alpha} \circ \varphi_{\beta}^{-1})\times (\varphi_{\alpha} \circ \varphi_{\beta}^{-1}))\circ \Delta_{\beta}]^{\ast}(u_{\alpha} \otimes v_{\alpha})
= \Delta_{\beta}^{\ast}\circ ((\varphi_{\alpha} \circ \varphi_{\beta}^{-1})\times (\varphi_{\alpha} \circ \varphi_{\beta}^{-1}))^{\ast}(u_{\alpha} \otimes v_{\alpha}),
\end{equation}
where we have used Theorem  2.126 of \cite{ma} in the second equality, the commutation of the previous diagram and Lemma 2.130 of \cite{ma}. 

\par By Lemma \ref{lema6}, Definitions 14 and 15, and equation  \eqref{212}, the last member of \eqref{lll} coincides with
$$\Delta_{\beta}^{\ast}\circ ((\varphi_{\alpha} \circ \varphi_{\beta}^{-1})^{\ast}(u_{\alpha})\otimes (\varphi_{\alpha} \circ \varphi_{\beta}^{-1})^{\ast}( v_{\alpha}))=\Delta_{\beta}^{\ast}(u_{\beta}\otimes \frac{1}{|\det(\varphi_{\alpha} \circ \varphi_{\beta}^{-1})'|} v_{\beta})$$
$$
= \frac{1}{|\det(\varphi_{\alpha} \circ \varphi_{\beta}^{-1})'|} \Delta_{\beta}^{\ast}(u_{\beta}\otimes v_{\beta}).
$$
Hence , by multiplying this equality by  $|\det(\varphi_{\alpha} \circ \varphi_{\beta}^{-1})'|$ we obtain
 $$\Delta_{\beta}^{\ast}(u_{\beta}\otimes v_{\beta})=|\det(\varphi_{\alpha} \circ \varphi_{\beta}^{-1})'|(\varphi_{\alpha} \circ \varphi_{\beta}^{-1})^{\ast}(\Delta^{\ast}_{\alpha}(u_{\alpha} \otimes v_{\alpha})).$$
 \par By equation \eqref{212} we get that

\begin{equation}
\Delta_{\beta}^{\ast}(u_{\beta}\otimes v_{\beta})=(\varphi_{\alpha} \circ \varphi_{\beta}^{-1})^{\bullet}(\Delta^{\ast}_{\alpha}(u_{\alpha} \otimes v_{\alpha})).
\end{equation}

\par Hence, the family $\{\Delta_{\alpha}^{\ast}(u_{\alpha}\otimes v_{\alpha})\}_{\alpha\in \mathbb{N}}$ defines an element of $D_{d}(\mathcal{M})$ which we will denote  $u\cdot v$. Indeed we have just proved the following proposition.

\begin{prop} There exist an action of $H_{d}(\mathcal{M})$ over $D_{d}(\mathcal{M})$ which we note
\begin{equation}\label{hd}
\xymatrix{H_{d}(\mathcal{M})\times D_{d}(\mathcal{M})\ar[r] & D_{d}(\mathcal{M}) \\
(u,v) \ar[r] & u\cdot v }
\end{equation}
provided that for each $\alpha \in \mathbb{N}$, there are not points $(x,\xi)$ in $WF(u_{\alpha})$ such that  $(x,-\xi)$ is in $WF(v_{\alpha})$. 
\end{prop}

\begin{rem}\label{rem}
\par Lastly we want mention that there is an alternative definition of H-distributions on a manifold which was introduced  by \cite{hor} or \cite{ma},  as elements of the continuous dual of $\Gamma_{c}(\omega)$ (the compact support sections of the density bundle of $\mathcal{M}$ (the manifold)) respect to a topology given by seminorms (this topology is explicitly defined in both text and was described here in the previous section) which makes of $\Gamma_{c}(\omega)$  a Fréchet space. 
\par On the other hand if we follow the text \cite{di} the distributions were introduced as elements of the continuous dual of $\mathcal{C}^{\infty}_{c}(\mathcal{M})$ (in \cite{hor} these were called distribution densities), we will call these D-distributions.
\par The relation between this and the above definitions is clearly established in both texts, but we say that for any family $\{u_{\alpha}\}_{\alpha\in \mathbb{N}}$ which satisfies the Definition \ref{def12} there exists only one element $\tilde{u}$ in the continuous dual of $\Gamma_{c}(\omega)$ such that locally works in the following way,
 \begin{equation*}
 \tilde{u}(\delta)=\int_{\varphi_{\alpha}(U_{\alpha})}u_{\alpha}(\varphi_{\alpha}^{-1}) ^{\ast}\delta
 \end{equation*}
where $\delta$ is an element of  $\Gamma_{c}(\omega)$ whose support is in $U_{\alpha}$. Using the Definition \ref{def12} one sees that $\tilde{u}$ is well defined independently of $U_{\alpha}$. Similar considerations are valid for the D-distributions. We invite the reader to deepen on these topics in the aforementioned books. 
\end{rem}



\chapter{Basic facts on Feynman measures} 

\section{Lagrangian formulation}

\par We begin by expressing in mathematical terminology certain ideas that are common in physics. The manifolds will allways be $T_{2}$ and satisfy the second axiom of countability in this thesis and consequently they are paracompact and a have locally finite partition of the unity.
\par A relation $\mathcal{R}$ in a topological space $X$ is said to be \emph{closed} if the set  $\{ (x,y)\in X^{2} / x\mathcal{R}y \} \subseteq X^{2}$ is closed for the product topology of $X^{2}$. 

\begin{definicion}\label{space}
A \emph{spacetime} will be a smooth finite-dimensional manifold $\mathcal{M}$, together with a  closed, reflexive and transitive relation  $\preceq$. Two points in the spacetime will be called \emph{spacelike separated} if $x\npreceq y$ and $y\npreceq x$.
\end{definicion}

\begin{ej}\label{ej1}
A first example of spacetime is $\mathbb{R}$ with the usual manifold structure and the classical order $\leqslant$ relation, which is obviously reflexive and transitive. The relation $\leqslant$ is also closed because the set $\{ (x,y) \in \mathbb{R}^{2} / x\leqslant y \}$ has as complement the positivity set of the continuous function $f(x,y)=x-y$. In this example there are no spacelike separated points, for $\leqslant$ is a total order. $\Box$
\end{ej}

\begin{ej}\label{ej2}
 We shall now consider the space $\mathbb{R}^{1,3}$, also denoted by $\mathbb{M}^{4}$, and called the \emph{Minkowski spacetime}. Its underlying manifold is $\mathbb{R}^{4}$ and the relation is given as follows. Define the bilinear form $\langle - ; - \rangle:\mathbb{R}^{4}\times\mathbb{R}^{4}\to \mathbb{R}$  
 by $\langle z^{1},z^{2}\rangle=z^{1}_{0}z^{2}_{0}-z^{1}_{1}z^{2}_{1}-z^{1}_{2}z^{2}_{2}-z^{1}_{3}z^{2}_{3}$ where $z^{i}=(z^{i}_{0},z^{i}_{1},z^{i}_{2},z^{i}_{3})$ for $i=1,2$.
 Then the relation is the following, $z_{1}\preceq z_{2}$ iff $\langle z_{2}-z_{1};z_{2}-z_{1}\rangle \geqslant 0$ and $z^{2}_{0}\geqslant z^{1}_{0}$. 
\par The relation $\preceq$ is clearly reflexive. It is transitive, because if $z^{3}=(z^{3}_{0},z^{3}_{1},z^{3}_{2},z^{3}_{3})$ and $z_{1}\preceq z_{2}$ and $z_{2}\preceq z_{3}$, then: 
\begin{align*}
&  (z^{3}-z^{1};z^{3}-z^{1})=(z^{3}_{0}-z^{1}_{0})^{2}-(z^{3}_{1}-z^{1}_{1})^{2}-(z^{3}_{2}-z^{1}_{2})^{2}-(z^{3}_{3}-z^{1}_{3})^{2} \\
&  =  (z^{3}_{0}-z^{2}_{0}+z^{2}_{0}-z^{1}_{0})^{2}-(z^{3}_{1}-z^{2}_{1}+z^{2}_{1}-z^{1}_{1})^{2}\\
& -(z^{3}_{2}-z^{2}_{2}+z^{2}_{2}-z^{1}_{2})^{2}-(z^{3}_{3}-z^{2}_{3}+z^{2}_{3}-z^{1}_{3})^{2} \\
&  =(z^{3}_{0}-z^{2}_{0})^{2}+(z^{2}_{0}-z^{1}_{0})^{2}+2(z^{3}_{0}-z^{2}_{0})(z^{2}_{0}-z^{1}_{0})\\
& -||((z^{3}_{1},z^{3}_{2},z^{3}_{3})-(z^{2}_{1},z^{2}_{2},z^{2}_{3}))+((z^{2}_{1},z^{2}_{2},z^{2}_{3})-(z^{1}_{1},z^{1}_{2},z^{1}_{3}))||^{2}.
\end{align*}

Applying the triangle inequality we have that
\begin{eqnarray*}
&  (z^{3}-z^{1};z^{3}-z^{1})\geqslant (z^{3}_{0}-z^{2}_{0})^{2}+(z^{2}_{0}-z^{1}_{0})^{2}+2(z^{3}_{0}-z^{2}_{0})(z^{2}_{0}-z^{1}_{0})\\
& -||(z^{3}_{1},z^{3}_{2},z^{3}_{3})-(z^{2}_{1},z^{2}_{2},z^{2}_{3})||^{2}-||(z^{2}_{1},z^{2}_{2},z^{2}_{3})-(z^{1}_{1},z^{1}_{2},z^{1}_{3})||^{2} \\
& = [(z^{3}_{0}-z^{2}_{0})^{2}-||(z^{3}_{1},z^{3}_{2},z^{3}_{3})-(z^{2}_{1},z^{2}_{2},z^{2}_{3})||^{2}]\\
& +[(z^{2}_{0}-z^{1}_{0})^{2}-||(z^{2}_{1},z^{2}_{2},z^{2}_{3})-(z^{1}_{1},z^{1}_{2},z^{1}_{3})||^{2}]+2(z^{3}_{0}-z^{2}_{0})(z^{2}_{0}-z^{1}_{0})\\
& =(z^{3}-z^{2};z^{3}-z^{2})+(z^{2}-z^{1};z^{2}-z^{1})+2(z^{3}_{0}-z^{2}_{0})(z^{2}_{0}-z^{1}_{0}).
\end{eqnarray*}
 The three last summands are $\geqslant 0$ because of the hypotheses $z_{1}\preceq z_{2}$ and $z_{2}\preceq z_{3}$. The transitivity of $\preceq$, follows from the fact that $z^{3}_{0}\geqslant z^{1}_{0}$.
 
 \par This relation is also closed, because the complement of
  $\{ (z^{1},z^{2})\in \mathbb{R}^{8}/ z^{1}\preceq z^{2} \}$ is the preimage of the open set $\{(x,y)\in \mathbb{R}^{2}/x>0 \text{ or }y>0 \}$ by the continuous function $f:\mathbb{R}^{8}\to \mathbb{R}^{2}$ defined by $f (z^{1},z^{2})=(||(z^{1}_{1},z^{1}_{2},z^{1}_{3})-(z^{2}_{1},z^{2}_{2},z^{2}_{3})||^{2}-(z^{1}_{0}-z^{2}_{0})^{2},z^{1}_{0}-z^{2}_{0})$.\hspace{10.8cm}$\Box$ 
 \end{ej}


\par As we mentioned in Section 1.2, if $\mathcal{M}$ is a spacetime then it has an structure $(\mathcal{M}, \mathcal{C}^{\mathbb{R}})$ of ringed space over $\mathbb{R}$. Let  $(E,p,\mathcal{M})$ be a  finite dimensional complex vector bundle over $\mathcal{M}$ and $\Phi$ the sheaf of modules over the ringed space $(\mathcal{M}, \mathcal{C}^{\mathbb{R}})$ associated to it (see Lemma \ref{lemshe}). This vector bundle and the corresponding sheaf are fixed from now on. It will be called the \emph{sheaf of classical fields} over  $\mathcal{M}$ and its local sections, i.e. elements of $\Phi(U)=\Gamma(U,E)$ where $U$ is open in $\mathcal{M}$, are called \emph{classical fields}.

\begin{ej}(Classical mechanics)
Consider the spacetime with the underlying manifold $\mathbb{R}^{3}$ and the identity relation, and the vector bundle $\pi:T\mathbb{R}^{3}\to \mathbb{R}^{3}$ given by the tangent bundle. A classical field in classical mechanics is a section of the previous vector bundle. \hspace{7.5cm} $\Box$
\end{ej}

\begin{ej}\label{ejcamppoescalar} (Classical field theory) In a classical field theory, one usually considers the trivial line vector bundle over the Minkowski spacetime $\pi:\mathbb{M}^{4}\times \mathbb{R}\to\mathbb{M}^{4}$. A classical field in this situation is just a section of $\pi$. 
So a classical field $\Phi$ is given by  a function  $\phi:\mathbb{M}^{4}\to \mathbb{R}$.
\par We recall that the d'Alembertian operator $\Box$ on a function $g:\mathbb{M}^{4}\to \mathbb{R}$ is given by 
$\Box g=\frac{\partial^{2}g}{\partial t^{2}}-\frac{\partial^{2}g}{\partial x_{1}^{2}}-\frac{\partial^{2}g}{\partial x_{2}^{2}}-\frac{\partial^{2}g}{\partial x_{3}^{2}}$. A \emph{massive scalar field of mass $m$} is a classical field $\varPhi(z)=(z,\phi(z))$ such that $\phi$ satisfies the \emph{Klein-Gordon equation}  $(-\Box +m^{2})\phi=0$.\hspace{12.1cm} $\Box$ 
\end{ej}

\begin{ej}\label{electro}(Electromagnetism)
 Let $\mathbb{M}^{4}$ be the Minkowski spacetime and $\Omega^{2}\mathbb{M}^{4}$ be the vector bundle of antisymmetric two forms over $\mathbb{M}^{4}$. If $U\subseteq \mathbb{M}^{4}$ is an open set, a local section of $\Omega^{2}\mathbb{M}^{4}$ over $U$ is of the form $$F(t,\vec{x})=\frac{1}{2}F_{\mu \nu}(t,\vec{x})dx^{\mu}\wedge dx^{\nu}$$ and is a classical field (where the Einstein's summation convention and the notation $x^{0}=t$, $x^{1}=x_{1}$,$x^{2}=x_{2}$ and $x^{3}=x_{3}$ were used). 

  
\par We define the electric $E:U\to \mathbb{R}^{3}$ and the magnetic $B:U\to \mathbb{R}^{3}$ fields by  
  
 \begin{eqnarray*}
   E_{x}(t,\vec{x})=F_{0 1}(t,\vec{x}), &  E_{y}(t,\vec{x})=F_{0 2}(t,\vec{x}), &  E_{z}(t,\vec{x})=F_{0 2}(t,\vec{x}),\\
    B_{z}(t,\vec{x})=-F_{1 2}(t,\vec{x}), &  B_{y}(t,\vec{x})=F_{1 3}(t,\vec{x}), & B_{x}(t,\vec{x})=-F_{2 3}(t,\vec{x}).
 \end{eqnarray*} 
 
 
 \par The classical Maxwell's equations in electromagnetism are written in the form
 
 \begin{align*}
 & \partial_{\eta} F_{\mu \nu} +\partial_{\mu} F_{\nu \eta}+\partial_{\nu} F_{\eta \mu}=0 & & \text{ and } & & \partial_{\nu}F^{\mu \nu}=j^{\mu},  
 \end{align*}
where $j^{o}(t,\vec{x})=\rho(t,\vec{x})$ is the charge density, $j^{i}$ the density current in the direction of $x_{i}$ for $1\leqslant i \leqslant 3$, and $F^{\mu \nu}= g^{\mu \alpha}g^{\nu \beta}F_{\alpha \beta}$ with $g$ the metric tensor $g=diag(-1,1,1,1)$. \hspace{11cm}$\Box$
 \end{ej}
 
\begin{definicion}
 Given the vector bundle $(E,p,\mathcal{M})$ over the spacetime $\mathcal{M}$ and $\Phi$ the associated sheaf, the sheaf of jets $J\Phi$ constructed after Lemma \ref{lema5} is called the sheaf of \emph{derivatives of classical fields}.
\end{definicion}
 

\begin{definicion}
The sheaf $S_{\mathcal{C}^{\mathbb{R}}}J\Phi$ given by the symmetric construction of the sheaf of modules $J\Phi$ over the ringed space $(\mathcal{M},\mathcal{C}^{\mathbb{R}})$ is called the \emph{sheaf of Lagrangians} (or \emph{composite fields}).
\end{definicion}	

\par If $\phi$ is a massive  scalar field of mass $m$ as in Example \ref{ejcamppoescalar}, then

\begin{equation*}
l_{1}=\frac{1}{2}\partial_{\mu}\phi\odot\partial^{\mu}\phi-\frac{1}{2}m^{2}\phi^{\odot 2}
\end{equation*}
 is an example of Lagrangian, where we have used Einstein's convention for the sum. 
We will usually omit the symmetric product to lighten the notation and just write $l_{1}=\frac{1}{2}\partial_{\mu}\phi\partial^{\mu}\phi-\frac{1}{2}m^{2}\phi^{2}$.

\par An example of Lagrangian for Example \ref{electro} is 
\begin{equation*}
l_{2}= F_{\mu,\nu} F^{\mu,\nu}.
\end{equation*} 

\par We will give a few more examples. If $\varphi$ and $\psi$ are classical fields over an spacetime $\mathcal{M}$ (see Example \ref{ejcamppoescalar}) and $f,g$ are smooth functions over $\mathcal{M}$; then
\begin{eqnarray*}
& l_{3}=\overline{\psi} \psi, \\
& l_{4}=f \varphi \partial_{1}\psi+g\varphi^{2} \partial_{1}\partial_{3}\psi,\\
& l_{5}=(f\varphi)\psi=\varphi(f\psi)=f(\varphi\psi),
\end{eqnarray*}
 are examples of Lagrangians.
\par Any combination of this Lagrangians with coefficients in the ring of smooth functions over the spacetime is also a Lagrangian. 

\par We recall that we denote by $\omega$ the sheaf of densities of the  spacetime.

\begin{definicion}\label{def5} 
The \emph{sheaf of Lagrangian densities} is $\omega SJ\Phi$. A Lagrangian density is any of its sections.
\end{definicion}

\par For example, if $\mathcal{M}$ is orientable of dimension $n$ and $\varphi$ is a classical field with support contained in an open set $U$ of $\mathcal{M}$, we can regard $d^{n}x$ (where $(U,x)$ is a chart) as a local section of the density bundle $\omega$ which together with any Lagrangian over $U$ can produce a Lagrangian density $\mathcal{L}$ as for example
\begin{equation*}
 \mathcal{L}=(\varphi+m\varphi^{2}\partial\varphi)d^{n}x,
\end{equation*}
where $m$ is a constant. 

\begin{definicion}\label{noloc}
A \emph{non-local action} is an element of the symmetric algebra $S\Gamma \omega SJ\Phi$ of the $\mathbb{R}$-vector space of global sections of the sheaf $\omega SJ\Phi$.
\end{definicion}


 
\begin{rem}
In order to topologize  the symmetric algebra $S\Gamma \omega SJ\Phi$ of the $\mathbb{R}$-vector space of global sections of the sheaf $\omega SJ\Phi$, we must first consider the algebra $T\Gamma \omega SJ\Phi$. And as we saw in Section \ref{1.2.2} there are two ways of topologize $T^{k}\Gamma \omega SJ\Phi$ ( for $k\in \mathbb{N}$) one considering the projective $\otimes_{\pi}$ topology over the tensor product and the second one considering its completation $\hat{\otimes}_{\pi}$.
\par We can choose the topology we want for $T^{k}\Gamma \omega SJ\Phi$, but if we take the topology given by $\hat{\otimes}_{\pi}$ then the two ways of topologize  $S^{k}\Gamma \omega SJ\Phi$ will be homeomorphic as we saw in Section \ref{jhjhjh}. 
\end{rem} 
  
\par The proofs for the next facts can be found in \cite{con}, Theorem 7.5.5 and Corollary 7.5.6.

\begin{thm}\label{df}
Let $E$ and $F$ be two finite dimensional complex vector bundles $\mathcal{M}$. The $\mathcal{C}^{\infty}(\mathcal{M})$-linear map $\alpha:\Gamma(E)\otimes_{\mathcal{C}^{\infty}(\mathcal{M})}\Gamma(F) \to \Gamma(E\otimes F) $, defined by $\alpha(s\otimes t)(x)=s(x)\otimes t(x)\in E_{x}\otimes F_{x}=(E\otimes F)_{x}$, where $s\in \Gamma(E)$, $t\in \Gamma(F)$ and $x\in \mathcal{M}$, is a canonical isomorphism of $\mathcal{C}^{\infty}(\mathcal{M})$-modules.
\end{thm}
As a corollary of this theorem we have
\begin{coro}\label{dm}
Given a finite dimensional complex vector bundle $E$ over $\mathcal{M}$, there exists a canonical isomorphism of $\mathcal{C}^{\infty}(\mathcal{M})$-modules between $\Gamma(S^{k}E)$ and $S^{k}_{\mathcal{C}^{\infty}(\mathcal{M})}\Gamma (E)$
for all $k\in \mathbb{N}_{0}$.
\end{coro}

\par In Section \ref{jets} we told that $J\Phi$ denote the vector bundle $J^{k}\Phi$ for some $k\in \mathbb{N}\cup \{ \infty \}$ and unless $k=\infty$,  $J^{k}\Phi$ is a finite dimensional vector bundle over $\mathcal{M}$ and we can apply the Theorem \ref{df} and Corollary \ref{dm} to it.
 
\par We want to mention a result similar to Theorem \ref{df}, whose proof involves the Serre-Swan Theorem for non compact manifolds, which establishes that the $\mathcal{C}^{\infty}(\mathcal{M})$-module of global sections of a smooth vector bundle over a manifold $\mathcal{M}$ is projective over $\mathcal{C}^{\infty}(\mathcal{M})$ (see \cite{serre}).  

\begin{prop}\label{uhuhuh}
 Given $E$ and $F$ two finite dimensional complex vector bundles over a manifold $\mathcal{M}$, the restriction of the morphism $\alpha$ of Theorem \ref{df} establishes an isomorphism $\alpha| :\Gamma_{c}(E)\otimes_{\mathcal{C}^{\infty}(\mathcal{M})} \Gamma(F)\to \Gamma_{c}(E\otimes F)$.
\end{prop}
\begin{proof}
Given a finite dimensional vector bundle $H$ over a manifold $\mathcal{M}$, we consider the map $\chi:\mathcal{C}_{c}^{\infty}(\mathcal{M})\times \Gamma(H)\to \Gamma_{c}(H)$ defined by $\chi(f,\varphi)=f\varphi$ where $f\in \mathcal{C}_{c}^{\infty}(\mathcal{M})$, $\varphi\in \Gamma(H)$ and if $p\in \mathcal{M}$, $f\varphi(p)=f(p)\varphi(p)$, then $f\varphi$ is a compact support section of $H$. The map $\chi$ is clearly $\mathcal{C}^{\infty}(\mathcal{M})$-balanced, hence it factors throughout the tensor product as $\overline{\chi}:\mathcal{C}_{c}^{\infty}(\mathcal{M})\otimes_{\mathcal{C}^{\infty}(\mathcal{M})} \Gamma(H)\to \Gamma_{c}(H)$. 
\par The map $\overline{\chi}$ is an isomorphism. It is surjective because if $\psi\in \Gamma_{c}(H)$ then suppose $U$ is an open set containing the support of $\psi$ and $V$ is a compact set such that $supp(\psi)\subseteq U \subseteq V$, then if we take $f\in \mathcal{C}_{c}^{\infty}(\mathcal{M})$ such that $f$ takes the value $1$ on $U$ and the value zero outside $V$ and $\overline{\chi}(f\otimes \psi)=f\psi=\psi$.
\par  On the other hand the map $\overline{\chi}$ is injective because if $\overline{\chi}(f\otimes \varphi)=0\in \Gamma_{c}(H)$, then there is no point $p\in \mathcal{M}$ such that both $f(p)\neq 0$ and $\varphi(p)\neq 0$. Let us consider an elemental sequence of compact sets $C_{n}\nearrow supp(\phi)$ and using the Urysohn's lemma a sequence of functions $g_{n}\in \mathcal{C}^{\infty}_{c}(\mathcal{M})$ such that $g\equiv 1$ over $C_{n}$ and $g\equiv 0$ over $supp(f)$. Then we have $f\otimes_{\mathcal{C}^{\infty}(\mathcal{M})} \varphi=\lim_{n\to \infty} f\otimes_{\mathcal{C}^{\infty}(\mathcal{M})}g_{n} \varphi=\lim_{n\to \infty} fg_{n}\otimes_{\mathcal{C}^{\infty}(\mathcal{M})} \varphi=0\in \mathcal{C}_{c}^{\infty}(\mathcal{M})\otimes_{\mathcal{C}^{\infty}(\mathcal{M})} \Gamma(H)$, as we want.

\par As the module $\Gamma(F)$ is $\mathcal{C}^{\infty}_{c}(\mathcal{M})$-projective then we can apply the functor $(-)\otimes_{\mathcal{C}^{\infty}_{c}(\mathcal{M})} \Gamma(F)$ to the precedent isomorphism 
$$\overline{\chi}:\mathcal{C}_{c}^{\infty}(\mathcal{M})\otimes_{\mathcal{C}^{\infty}(\mathcal{M})} \Gamma(E)\to \Gamma_{c}(E)$$
taking $H=E$ and still have an isomorphism of $\mathcal{C}^{\infty}_{c}(\mathcal{M})$-modules, $$\mathcal{C}_{c}^{\infty}(\mathcal{M})\otimes_{\mathcal{C}^{\infty}(\mathcal{M})} \Gamma(E)\otimes_{\mathcal{C}^{\infty}(\mathcal{M})} \Gamma(F) \to \Gamma_{c}(E)\otimes_{\mathcal{C}^{\infty}(\mathcal{M})} \Gamma(F). $$
\par Moreover by applying Theorem \ref{df} we have an isomorphism,
$$\mathcal{C}_{c}^{\infty}(\mathcal{M})\otimes_{\mathcal{C}^{\infty}(\mathcal{M})} \Gamma(E\otimes F) \to \Gamma_{c}(E)\otimes_{\mathcal{C}^{\infty}(\mathcal{M})} \Gamma(F) $$
and by applying the above result again using $H=E\otimes F$ then we have the isomorphism
$$ \Gamma_{c}(E\otimes F) \to \Gamma_{c}(E)\otimes_{\mathcal{C}^{\infty}(\mathcal{M})} \Gamma(F) $$
which is what we established.

\end{proof}

\begin{rem}\label{ghgh}
\par In this remark we will apply the above results to the case of sections of the vector bundle $J\Phi$ (always under the hypothesis of finite order jets).
\par Given $k\in \mathbb{N}$, the sections of the vector bundle $S^{k}J\Phi$, that is $\Gamma S^{k}J\Phi$ are isomorphic by Corollary \ref{dm} to $S^{k}_{\mathcal{C}^{\infty}(\mathcal{M})}\Gamma J\Phi$. Moreover by Theorem \ref{df} together with the recent comment we have an isomorphism between $\Gamma \omega S^{k}J\Phi$ and $\Gamma \omega \otimes_{\mathcal{C}^{\infty}(\mathcal{M})} S^{k}_{\mathcal{C}^{\infty}(\mathcal{M})}\Gamma J\Phi$.
\par Lastly by  Proposition \ref{uhuhuh} we have an isomorphism similar to the last one, but very important if we want integrate a local action over the spacetime, $\Gamma_{c} \omega S^{k}J\Phi\simeq\Gamma_{c} \omega \otimes_{\mathcal{C}^{\infty}(\mathcal{M})} S^{k}_{\mathcal{C}^{\infty}(\mathcal{M})}\Gamma J\Phi$.
\par As a consequence a \emph{monomial element} of the space of non-local actions $A\in S^{n}\Gamma \omega S^{k}J\Phi$ with $n\in \mathbb{N}$ can be written $A=\prod_{i=1}^{n} \alpha_{i}\otimes (\odot_{j=1}^{k} l_{j}^{i})$, $\alpha_{i}\in \Gamma(\omega)$ and $l_{j}^{i}\in \Gamma(J\Phi)$, and $\prod$ denotes the product of $SV$  for $V=\Gamma\omega SJ\Phi$. In view of the above remarks the densities $\alpha_{i}\in \Gamma(\omega)$ can be taken of compact support, i.e. $\alpha_{i}\in \Gamma_{c}(\omega)$, if $A$ has compact support.
\end{rem}
 
 \section{Propagators}

\subsection{Types of propagators}

Given a complex finite dimensional vector bundle $\pi:E\to \mathcal{M}$ over the spacetime and $\varphi \in \Gamma(J\Phi)$ we denote by $\varphi^{\ast}$ its complex conjugate section, i.e. the section such that $\varphi^{\ast}(p)=\overline{\varphi(p)}$ for all $p\in \mathcal{M}$ (the bar indicates the complex conjugate). And if $f=\delta \varphi \in \Gamma_{c}\omega J\Phi$, we denote by $f^{\ast}$ the section $\delta \varphi^{\ast}$ where $\ast$ does not work over the real densities, i.e. $f^{\ast}(p)=\delta(p)\otimes\overline{\varphi(p)}$ if  $\delta$ is a real density an $p\in \mathcal{M}$.

{\definicion\label{pro}{A \emph{propagator} associated to the vector bundle $\pi:E\to \mathcal{M}$ (or its associated sheaf $\Phi$) is a continuous and $\mathbb{R}$-bilinear function $$\Delta: \Gamma_{c}\omega J\Phi \times \Gamma_{c}\omega J\Phi \rightarrow \mathbb{C}.$$
We denote the space of propagators associated to  $\pi:E\to \mathcal{M}$ by $Prop(E)$.
We say that the propagator is
\begin{itemize}
\item  \emph{local} if $\Delta(f,g)=\Delta(g,f)$ for each $f$ and $g$ whose supports are spacelike separated (see Definition \ref{space}) .
\item  \emph{Feynman} if it is symmetric.
\item  \emph{Hermitian} if $\Delta^{\ast}=\Delta$, where  $\Delta^{\ast}(f^{\ast},g^{\ast}):=\overline{\Delta(g,f)}$ 
and $f^{\ast}$ is the section of $\Gamma_{c}\omega J\Phi$ given by a $\ast$-operation on it.
\item \emph{positive} if $\Delta(f^{\ast},f)\geqslant 0$.
\end{itemize} }}

 
Given two finite dimensional vector bundles $\pi_{E}:E\to M$ and $\pi_{F}:F\to N$, we consider the manifold $N\times M$ and pull-back these bundles by the projections $pr_{1}:N\times M \to N$ and   $pr_{2}:N\times M \to M$ to construct the bundles, $ pr_{1}^{\ast}(\pi_{F}):pr_{1}^{\ast}(F) \to N\times M$ and $ pr_{2}^{\ast}(\pi_{E}):pr_{2}^{\ast}(E) \to N\times M$, so we have two bundles over $N\times M$ and we denote by $F\boxtimes E$ the bundle $ pr_{1}^{\ast}(F)\otimes pr_{2}^{\ast}(E) $ over $N\otimes M$.

 \begin{prop}
 Given two finite dimensional $K$ (real or complex) vector bundles $E$ and $F$, we have an isomorphism  $\Gamma(E\boxtimes F)\simeq\Gamma(E)\hat{\otimes}_{K}\Gamma(F)$ where the completion is with respect to the just described Fréchet topologies.  
 \end{prop}
 
 The next proposition give us a link between propagators and distributions over $\mathcal{M}\times \mathcal{M}$.

 \begin{prop}
 There is a linear map $\iota:Prop(E)\to \Gamma(\omega_{\mathcal{M}\times \mathcal{M}} \otimes (J\Phi \boxtimes J\Phi))'$ given by $\iota(\Delta)((\delta_{1}\otimes \delta_{2})\otimes (f\otimes g))=\Delta(\delta_{1}\otimes f,\delta_{2}\otimes g)$, for all $\delta_{1},\delta_{2}\in \Gamma(\omega)$ and $f,g\in \Gamma(J\Phi)$.
 \end{prop} 
 \begin{proof}
 
 Since $\Gamma(\omega_{\mathcal{M}\times \mathcal{M}}\otimes(J\Phi\boxtimes J\Phi))=\Gamma((\omega_{\mathcal{M}}\boxtimes \omega_{\mathcal{M}})\otimes (J\Phi \boxtimes J\Phi))$ for $\omega_{\mathcal{M}\times \mathcal{M}}=\omega_{\mathcal{M}}\boxtimes \omega_{\mathcal{M}}$, the previous conditions determine precisely one unique $(J\Phi\boxtimes J\Phi)^{\ast}$-valued distribution on $\mathcal{M}\times \mathcal{M}$
 
 
 \end{proof} 

\begin{definicion}
The space of bilinear maps with respect to $\mathcal{C}^{\infty}(M)^{\otimes 2}$ associated to a bundle $\pi:E\to \mathcal{M}$
\begin{equation*}
\Gamma_{c}J\Phi\times \Gamma_{c}J\Phi \to \{ \text{H-distibutions of compact support on }\mathcal{M}\times \mathcal{M}  \}=\Gamma(\omega_{\mathcal{M}\times \mathcal{M}})'
\end{equation*}
will be denoted by $Prop'(E)$
\end{definicion}


\begin{rem} \label{khkg}
There is a linear map $\xi: Prop(E)\to Prop'(E)$ satisfying that $\xi(\Delta)(A\otimes B)(f\otimes g)=\Delta(fA,gB)$ for all $A, B\in \Gamma_{c}J\Phi$ and $f,g\in \Gamma(\omega_{\mathcal{M}})$.
This map is clearly injective, because if $\xi(\Delta)=\xi(\tilde{\Delta})$ then by using  the Proposition \ref{uhuhuh} and the definition of $\xi$ we conclude that $\Delta(A,B)=\tilde{\Delta}(A,B)$ for all $A, B \in \Gamma_{c} \omega J\Phi$, i.e. $\Delta$ is equal to $\tilde{\Delta}$. 
\par Reciprocally each map in $Prop'(E)$ induces a morphism of $\mathcal{C}^{\infty}(M)^{\otimes 2}$-modules between $\Gamma_{c}J\Phi\hat{\otimes}_{K}\Gamma_{c}J\Phi$ and $\Gamma(\omega_{\mathcal{M}\times \mathcal{M}})'$, then $\xi$ is bijective.
\end{rem}

 \par There is an important family of propagators that we want to use. 
 
\begin{definicion}
Given a vector space $V$, we say that a set $C\subset V$ is a cone if $\xi \in C$ implies that $t\xi\in C$ for all $t>0$.
\end{definicion} 
 
If we are working in a t.v.s. we can talk about closed cones.
\begin{definicion}
A cone which is properly contained in a semispace is called  proper. 
\end{definicion} 
Note that a semispace is not a proper cone. 
 
\par  Given a closed proper convex cone $C$ of a t.v.s. $Z$, we can define a partial order on $Z$  of the form $x\curlyeqprec y$ if $y-x\in C$. 
 
\begin{definicion}\label{defcut}
A propagator $\Delta:\Gamma_{c} \omega J\Phi \times \Gamma_{c}\omega J\Phi \to \mathbb{C}$  is called \emph{cut}, if for each $\varphi, \psi \in\Gamma_{c}J\Phi$ and each  $z\in\mathcal{M}$ the  H-distribution over $\mathcal{M}\times\mathcal{M}$: $\xi(\Delta)(\varphi,\psi)$ possesses a partial order in the cotangent space defined by a proper closed convex cone $C_{z}$, such that if $(p,q)$ is in the wavefront set of $\xi(\Delta)(\varphi,\psi)$ at some point $(x,y)\in \mathcal{M}\times \mathcal{M}$ then $ p\curlyeqprec 0 $ and $0\curlyeqprec q $. Moreover if $x=y$ then $p+q=0$.
\end{definicion}

\begin{ej}\label{ejadv}
 In Quantum Field Theory one often encounters the so-called advanced propagator $\Delta_{(+)}$ such that its associated element in $Prop'(E)$. If it is evaluated in the scalar field, that is $\xi(\Delta_{(+)})(\phi,\phi)$, (see \cite{foll} chapter 6) in local coordinates is given by the function $$\Delta_{+}(x)=\frac{i}{2(2\pi)^{3}}\int_{\mathbb{R}^{3}}\frac{e^{-ip_{\mu}x^{\mu}}}{2p_{0}}dp_{1}dp_{2}dp_{3},$$ where 
  $p_{0}=\sqrt{|\vec{p}|^{2}+m^{2}}$. This means that the H-distribution $\xi(\Delta_{(+)})(\phi,\phi)$ evaluated in two compact support densities $f(x)d^{4}x$ and $g(y)d^{4}y$ is given by
\begin{equation*}
\xi(\Delta_{(+)})(\phi,\phi)(f(x)d^{4}x,g(y)d^{4}y)=-i\int_{supp(f)\times supp(g)} \Delta_{+}(x-y)f(x)g(y)d^{4}xd^{4}y.
\end{equation*}
\par  The wavefront set of this distributions was calculated in \cite{reed}, Theorem IX.48, and it gives 
\begin{equation}\label{wfffffff}
WF(\Delta_{+})= \{ (0,-|\vec{p}|,\vec{p})/ 0\in \mathbb{R}^{4}, \vec{p}\in \mathbb{R}^{3}\setminus \{0\}  \}\cup \{ (\pm |\vec{x}|,\vec{x},-\lambda |\vec{x}|,\mp\lambda\vec{x})/ \vec{x}\in\mathbb{R}^{3},\lambda >0 \}
\end{equation}
 

By pulling back this distribution with the map $\mathbb{R}^{4}\times \mathbb{R}^{4}\to \mathbb{R}^{4}$, $(x,y)\mapsto x-y$ we obtain the wavefront set of $\xi(\Delta_{(+)})(\varphi,\varphi)$ over $\mathbb{M}^{4}\times \mathbb{M}^{4}$, which is 

\begin{equation*}\label{ggg}
WF(\xi(\Delta_{(+)})(\varphi,\varphi))=\{ (x,p,y,-p)\in \mathbb{R}^{16}/ (x-y,p)\in WF(\Delta_{+})\}.
\end{equation*}


  \end{ej}

 \begin{ej}\label{fey}
 Consider the distribution $\theta_{x^{0}}$ on $\mathbb{M}^{4}$ defined by $\theta_{x^{0}}(f)=\int_{x^{0}\geqslant0} f(x) d^{3}\vec{x}$. In this case 
 \begin{equation}\label{wafff}
 WF(\theta_{x^{0}})= \{ (0,\vec{x},\lambda,\vec{0})/ \vec{x}\in \mathbb{R}^{3}, \lambda \in \mathbb{R}\setminus\{0\}\}
 \end{equation} (see \cite{ma}). Since $WF(\theta_{x^{0}})$ and $WF(\Delta_{+})$ satisfy that there is no point $(x,p)\in WF(\theta_{x^{0}})$ such that $(x,-p)\in WF(\Delta_{+})$, because the only $x$ which is in both singular supports is $0$ but in this case the cones has not intersections (see formulas \eqref{wfffffff} and \eqref{wafff}), we can define their product $\theta_{x^{0}} \Delta_{+}$ (see \cite{reed}). The Feynman propagator is given by  
 \begin{equation*}
 \Delta_{F}(x)=\theta_{x^{0}} \cdot \Delta_{+}(x)+\theta_{-x^{0}} \cdot \Delta_{+}(-x).
 \end{equation*}
 The wavefront set of this distribution was computed in \cite{viet}, Proposition 26. If we pull back the Feynman propagator by the diagonal map we obtain an H-distribution on $\mathbb{M}^{4}\times \mathbb{M}^{4}$.
 \end{ej} 
 
 \begin{definicion}
 A propagator  $\Delta:\Gamma_{c} \omega J\Phi \times \Gamma_{c}\omega J\Phi \to \mathbb{C}$ is called \emph{polynomially smooth} if the expression in local coordinates of $\xi(\Delta)(A,B)$ is a sum of products of smooth functions and powers of logarithms of polynomials, for any $A,B\in \Gamma_{c}J\Phi$.
 \end{definicion}
 
 \begin{ej}
 We recall that the Feynman propagator of the massive scalar field described in Example \ref{fey} is $$\Delta_{F}(t,\vec{x})=i\int\frac{e^{-i\omega_{p}|t|+i\vec{p\cdot \vec{x}}}}{2\omega_{p}}\frac{d^{3}\vec{p}}{(2\pi)^{3}},$$ where $\omega_{p}=\sqrt{|\vec{p}|^{2}+m^{2}}$. The integral can be explicitly computed (see \cite{foll} Chapter 6, Section 5) and it gives
 \begin{equation*}
 \Delta_{F}(t,\vec{x})=\frac{im}{4\pi^{2}\sqrt{|\vec{x}|^{2}-t^{2}}}K_{1}\Big( m\sqrt{|\vec{x}|^{2}-t^{2}} \Big) \hspace{.8cm} on \hspace{.8cm} \mathbb{R}^{4}\setminus \{ (t,\vec{x})/ |t|=|\vec{x}| \}, 
 \end{equation*}
where $K_{1}$ is the modified Bessel function of order $1$. Then $\Delta_{F}$ is polynomially smooth. For more examples of smooth polynomial propagators see the Dirac and the Proca propagators in Chapter 6, Section 5 of \cite{foll}. 
\end{ej}
 
 \subsection{Generalized propagators}
 

 \par We now proceed to extend  the propagator to a bigger domain. 
 \par We recall that a propagator $\Delta$  gives a bilinear map $\xi(\Delta):\Gamma_{c}J\Phi\times \Gamma_{c}J\Phi\to\Gamma(\omega_{\mathcal{M}\times \mathcal{M}})'$ (see Remark \ref{khkg}).
  
\par We first define  $\tilde{\Delta}:\Gamma_{c} SJ\Phi \times \Gamma_{c} SJ\Phi \rightarrow \Gamma(\omega_{\mathcal{M}\times \mathcal{M}})'$ by means of the formula
\begin{equation}\label{223}
\tilde{\Delta}(a_{1}\odot...\odot a_{n},b_{1}\odot... \odot b_{n})=\sum_{\sigma\in S_{n}}\xi(\Delta)(a_{1},b_{\sigma{1}})\cdot \xi(\Delta)(a_{2},b_{\sigma_{2}})\cdot ...\cdot \xi(\Delta)(a_{n},b_{\sigma_{n}}),
\end{equation}
where the $\cdot$ is a product of H-distributions  (recall the product of H-distribution is described in \cite{ma}, Chapter 2) and we define  $\tilde{\Delta}$ to be zero if it is evaluated at different degree elements. 

\par If $V=\Gamma SJ\Phi$, then $SV$ admits two structures of coalgebra. The first one is the cocommutative cofree coaugmented symmetric coalgebra of $V$, with this structure the elements of $V$ are primitives. Let us call $\Delta_{C}$ to the coproduct of this coalgebra.
\par For the other hand if we make the construction described in \ref{symme} to the object $J\Phi$, we obtain the symmetric coalgebra $SJ\Phi$, in which the elements of $J\Phi$ are primitive elements. This structure has a unique extension to the algebra $SV$, such that $SV$ is a Hopf algebra (see \cite{tak}). Let us denote by $\Delta_{T}$ the coproduct of $SV$ with this structure. 
\par If $\varphi, \psi \in J\Phi$, then $\varphi\psi\in S^{2}J\Phi\subseteq S^{1}\Gamma S^{2}J\Phi\subseteq SV$ (we will denote with an underline the elements of $J\Phi$ viewed as elements of $S^{1}J\Phi$) and,
\begin{equation*}
\Delta_{C}(\underline{\varphi\psi})=1\otimes \underline{\varphi\psi}+\underline{\varphi\psi}\otimes 1
\end{equation*}
but 
\begin{align*}
& \Delta_{T}(\underline{\varphi\psi})=\Delta_{T}(\underline{\varphi})\cdot \Delta_{T}(\underline{\psi})\\
& =(1\otimes \underline{\varphi}+ \underline{\varphi}\otimes 1) \cdot (1\otimes \underline{\psi}+\underline{\psi}\otimes 1)\\
& =1\otimes \underline{\varphi\psi}+\underline{\psi}\otimes \underline{\varphi}+\underline{\varphi}\otimes \underline{\psi}+\underline{\varphi\psi}\otimes 1.
\end{align*} 

\par In the following when we refer to the coproduct of an element in $SV$, it will be in the sense of $\Delta_{T}$. 

\par Finally, we extend $\tilde{\Delta}$ to a map $\hat{\Delta}$ from $S^{m}\Gamma_{c} SJ\Phi \times S^{n}\Gamma_{c} SJ\Phi$ to H-distributions of compact support on  $\mathcal{M}^{m}\times \mathcal{M}^{n}$ by the recursive expression
\begin{equation}\label{magenta}
\hat{\Delta}(A B,C)=\sum \hat{\Delta}(A,C')\otimes\hat{\Delta}(B,C''), 
\end{equation}
\begin{equation*}
\hat{\Delta}(A,1)=\varepsilon(A),
\end{equation*}
where $\sum C'\otimes C''$ is the coproduct of $C$ in $S^{n}\Gamma_{c} SJ\Phi $, and  $\hat{\Delta}(A,C')\otimes\hat{\Delta}(B,C'')$; an H-distribution on $\mathcal{M}^{m}\times \mathcal{M}^{n}$. Note that \eqref{magenta} is precisely the definition of a Laplace pairing (see \cite{su}, Definition 10.2). 
\par As an example suppose $A$, $B$ and $C$ are Lagrangian fields of compact support, i.e. $A,B,C\in \Gamma_{c}SJ\Phi$, then $AB\in S^{2}\Gamma_{c}SJ\Phi$ and $C\in  S^{1}\Gamma_{c}SJ\Phi$. Then the distribution $\hat{\Delta}(A B,C)$ must be evaluated in densities on  $\mathcal{M}^{2}\times \mathcal{M}$. Assume $\alpha,\beta,\gamma\in \Gamma \omega$, then $(\alpha\otimes \beta)\otimes \gamma$ is a density over $\mathcal{M}^{2}\times \mathcal{M}$ and the definition of $\hat{\Delta}$ indicates the following, 
\begin{equation}\label{gfgfgf}
\hat{\Delta}(A B,C)((\alpha\otimes \beta)\otimes \gamma)=\sum \tilde{\Delta}(A,C')(\alpha\otimes \gamma)\cdot\tilde{\Delta}(B,C'')(\beta\otimes \gamma) 
\end{equation}
where we put $\tilde{\Delta}$ in the right hand side, instead of $\hat{\Delta}$, because when we restrict to elements of degree zero or one it takes the same values, and $A,B,C',C''\in S^{\leqslant 1}\Gamma_{c}SJ\Phi$. The product $\cdot$ in \eqref{gfgfgf}, is a product of complex numbers.

\par With this successive extensions of a given propagator $\Delta$, we define $\hat{\Delta}$ whose restrictions $\hat{\Delta}|_{ S^{m}\Gamma_{c}S J\Phi \times S^{n}\Gamma_{c} SJ\Phi}$  satisfy the following. 

\begin{definicion}
A \emph{generalized propagator} is a family $\hat{\Delta}=\{\Delta_{m,n}\}_{m,n\in\mathbb{N}}$ of continuous and $\mathbb{R}$-bilinear functions $$\Delta_{m,n}: S^{m}\Gamma_{c}S J\Phi \times S^{n}\Gamma_{c} SJ\Phi \longrightarrow H_{d}(\mathcal{M}^{m}\times \mathcal{M}^{n})$$
where the topology for the codomain is the weak-$\ast$ topology (see \cite{ru}, Chapter 3, Section 11 and Chapter 6 for a general description of the weak-$\ast$ topology and Section 16 for the specific case $\mathcal{M}^{m}\times \mathcal{M}^{n}=\mathbb{R}^{k}$).
\end{definicion}

\begin{definicion}\label{newcut}
A generalized propagator $\Delta=\{\Delta_{m,n}\}_{m,n\in\mathbb{N}}$ is called \emph{cut} if for all $x\in \mathcal{M}$ there is a partial order in the cotangent space of $\mathcal{M}$ at $x$ defined by a proper closed convex cone $C_{x}$, such that for all $(m,n)\in \mathbb{N}^{2}$ and $(p_{1},\cdots,p_{m},q_{1},\cdots,q_{n})$ in the wavefront set of $\Delta_{m,n}$ at the point $(x_{1},\cdots,x_{m},y_{1},\cdots,y_{n})$ all the $p_{i} \curlyeqprec 0$ and all the $0 \curlyeqprec q_{j}$.  And if $(p_{1},\cdots,p_{m},q_{1},\cdots,q_{n})$ is in the wavefront set of $\Delta_{m,n}$ on the diagonal of $\mathcal{M}^{m}\times \mathcal{M}^{n}$ then $p_{1}+\cdots+p_{m}+q_{1}+\cdots+q_{n}=0$.
\end{definicion}

\section{Feynman measure}

\par We are interested in integrating expressions of the form $e^{iL_{F}}\mathcal{L}$, where $\mathcal{L}$ is a lagrangian density.

\par Let us define a collection of maps $\chi_{n}:(S\Gamma_{c}\omega SJ\Phi)'\times S^{n}\Gamma_{c}\omega SJ\Phi\to \mathcal{C}^{\infty}(\mathcal{M}^{n})'$ for $n\in \mathbb{N}$ by $\chi_{n}(\delta,A_{1}\odot \dots \odot A_{n})(f_{1}\otimes \cdots \otimes f_{n})=\delta(f_{1}A_{1}\odot \cdots \odot f_{n}A_{n})$, where $f_{i}\in \mathcal{C}^{\infty}(\mathcal{M})$
and $A_{i}\in \Gamma_{c}\omega SJ\Phi$, $\forall i=1,\cdots, n$. Since the subspace  of $\mathcal{C}^{\infty}(\mathcal{M})$ formed by sums of functions of the form $f_{1}\otimes \cdots \otimes f_{n}$ is dense, the maps $\chi_{n}$ are well-defined.
\par Set  $$ \chi: (S\Gamma_{c}\omega SJ\Phi)'\times S\Gamma_{c}\omega SJ\Phi\to \oplus_{n\in \mathbb{N}}\mathcal{C}^{\infty}(\mathcal{M}^{n})'$$ the direct sum of those maps.
\begin{definicion}\label{suavidad en la diagonal} 
Let $\delta: S\Gamma_{c}\omega S J\Phi\rightarrow \mathbb{C} $ be a continuous  and  linear map. We say that it is \emph{smooth on the diagonal} 
if  $(p_{1},...,p_{n})\in \Sigma_{(q,...,q)}\chi_{n}(\delta,A)$ implies that $p_{1}+\cdots+p_{n}=0$, for all $A\in S^{n}\Gamma_{c}\omega SJ\Phi$.
\end{definicion}

\par We are now ready to introduce the notion of Feynman measure.

\begin{definicion}\label{dia}
A Feynman measure is a continuous linear map from $S\Gamma_{c}\omega S J\Phi$ to $\mathbb{C}$.
Let $\delta:S\Gamma_{c}\omega S J\Phi\rightarrow \mathbb{C}$ be a continuous linear map and let $\Delta:\Gamma_{c}\omega J\Phi\times \Gamma_{c}\omega J\Phi \to \mathbb{C}$ be a propagator. Then, $\delta$ is said to be associated to the propagator $\Delta$ if it satisfies the following conditions:

\begin{enumerate}
\item $\delta$ is smooth in the diagonal;
\item (non-degeneracy) There is a smooth nowhere vanishing function $g$ so that 
 $\delta(v)=\int_{\mathcal{M}}gv$ for all $v\in S^{1}\Gamma_{c}\omega S^{0}J\Phi=\Gamma_{c}\omega$;
\item (Gaussian condition or weak translational invariance) Let $A\in S^{m}\Gamma_{c}\omega S J \Phi$ and $B\in S^{n}\Gamma_{c}\omega S J \Phi$ such that there is no point in $supp(A)$ that is $\leqslant$ to some point in the $supp(B)$. This means that there is not point at $supp(A)$ which is $\leqslant$  to some point at $supp(B)$, we actually should say that there is no point in $supp(A)\subseteq \mathcal{M}^{m}$ such that any of its coordinate be $\leqslant$ to some coordinate of some point in $supp(B)\subseteq \mathcal{M}^{n}$. Then the following equality holds 
\begin{equation}\label{gauss}
 \chi(\delta,A\cdot B)=\sum (\chi(\delta,A')\otimes \chi(\delta,B')) \cdot \hat{\Delta}(A'',B'')
\end{equation}
where $\sum A'\otimes A''\in S\Gamma_{c}\omega SJ\Phi \otimes S\Gamma_{c}SJ\Phi$ is the image of $A$ given by the coaction $S^{m}\Gamma_{c}\omega SJ\Phi \rightarrow S^{m}\Gamma_{c}\omega SJ\Phi \otimes S^{m}\Gamma_{c}SJ\Phi$. 
\end{enumerate}
 
 \end{definicion}
 \begin{rem}
\par The product in \eqref{gauss} will be regarded in the following manner. Given $A''\in S^{m}\Gamma_{c}SJ\Phi$ and $B''\in S^{n}\Gamma_{c}SJ\Phi$, $\Delta(A'',B'')$ defines an element of $H_{d}(\mathcal{M}^{m}\times \mathcal{M}^{n})$ that can be multiplied by $\chi(\delta,A')\otimes \chi(\delta,B')\in D_{d}(\mathcal{M}^{m}\times \mathcal{M}^{n})$ using the product in \ref{hd}.
\end{rem}
\begin{rem}\label{coac}
\par Let us explain the coaction involved in the Gaussian condition in more detail.


\par By using the Theorem \ref{df} and its Corollary \ref{dm}, $\Gamma_{c}\omega SJ\Phi$ can be written in the form $\Gamma_{c}\omega\otimes_{\mathcal{C}^{\infty}(\mathcal{M})} \Gamma SJ\Phi$. Then by using Proposition \ref{prop 1} we conclude that $\Gamma_{c} \omega SJ\Phi$ is a comodule over $\Gamma SJ\Phi$ (call $\rho$ the coaction), and consequently  $S\Gamma_{c} \omega SJ\Phi$ is a comodule over $S\Gamma SJ\Phi$ and then it has sense to talk about a coaction. 
\par In Sweedler notation  if $\mathcal{L}\in \Gamma\omega SJ\Phi$ then $\rho(\mathcal{L})=\sum_{()}\mathcal{L}_{(0)}\otimes \mathcal{L}_{(1)}$ where $\mathcal{L}_{(0)}\in \Gamma\omega SJ\Phi$ and $\mathcal{L}_{(1)}\in \Gamma SJ\Phi$, then we can define a coaction $\rho^{n}:S^{n}\Gamma\omega SJ\Phi\to S^{n}\Gamma \omega SJ\Phi \otimes S^{n}\Gamma SJ\Phi$ that over sections works as $$\rho^{n}(\mathcal{L}_{1}\odot\cdots\odot\mathcal{L}_{n})=\sum_{()} (\mathcal{L}_{1(0)}\odot \cdots \odot \mathcal{L}_{n(0)}) \otimes (\mathcal{L}_{1(1)}\odot \cdots \odot \mathcal{L}_{n(1)}) $$ where $\mathcal{L}_{j(0)}\in \Gamma \omega SJ\Phi $ for all $1\leqslant j \leqslant n$ and   $\mathcal{L}_{k(1)}\in \Gamma SJ\Phi$ for all $1\leqslant k \leqslant n$, note that $\mathcal{L}_{1(1)}\odot \cdots \odot \mathcal{L}_{n(1)}$ belongs to $\Gamma SJ\Phi$. We denote by $\sigma$ the coaction induced by $\oplus_{n\in \mathbb{N}_{0}}\rho^{n}$ which makes of $S\Gamma\omega SJ\Phi$ a $\Gamma SJ\Phi$-comodule.  
\end{rem}

\begin{rem}\label{obs4}
\par Provided our propagator $\Delta$ is cut, the smoothness on the diagonal allows us to compute the product at the right hand of \eqref{gauss}. Indeed, suppose we have $(p_{1},\cdots ,p_{m},q_{1},\cdots q_{n})$ in the wave front set of $\Delta(A'',B'')$ at a point $(x,\cdots, x,y,\cdots,y)$, then $p_{i}\curlyeqprec 0$ and $0\curlyeqprec q_{j}$ for all $1\leqslant i \leqslant m$ and   $1\leqslant j \leqslant n$.
\par  If $(-p_{1},\cdots ,-p_{m},-q_{1},\cdots -q_{n})$ is in the wave front set of $\chi(\delta,A')\otimes \chi(\delta,B')$ Lemma 2.175 in \cite{ma} implies that $q_{j}=0$ for all  $1\leqslant j \leqslant n$ or $(-q_{1},\cdots ,-q_{n})$ is in the wave front set of $\chi(\delta,B')$ at $(x,\cdots,x)$. Consequently $-q_{1}-q_{2}\cdots-q_{n}=0$ by the smoothness on the diagonal condition of the Feynman measure. In other words we get that $q_{1}+\cdots +q_{n}=0$. This condition together with the previous one stating that $0 \curlyeqprec q_{i}$ (i.e. $q_{i}$ is in the proper convex closed cone $C_{x}$) implies that all the $q_{j}$ must be zero. A similar argument can be used to see that all the $p_{i}$ are zero. This is absurd, for $(p_{1},\cdots ,p_{m},q_{1},\cdots q_{n})$ is in the wave front set of $\Delta(A'',B'')$. 
The contradiction came from the assumption that  $(-p_{1},\cdots ,-p_{m},-q_{1},\cdots -q_{n})$ was in the wave front set of $\chi_{m}(\delta,A')\otimes \chi_{n}(\delta,B')$. 
\par As a consequence, we conclude that the product between these distributions is well defined. 
\end{rem}

\par The idea is to think of the Feynman measure in $A$ as the integral of $A$ over all the classical fields, that is $\delta(A)=\int_{\Phi} A(\varphi)\mathcal{D}\varphi$.

\begin{rem}\label{obs5}
 Given $\Delta$ a Feynman measure and  $A\in S^{n}\Gamma_{c}\omega S J\Phi$, then  $\delta(A)$ is a complex number and  $\chi_{n}(\delta,A)$ is an element of $D_{d}(\mathcal{M}^{n})$. The latter  can be thought of a density over $\mathcal{M}^{n}$ in the sense that the densities take functions and assign numbers (by integration over $\mathcal{M}$), i.e. $\chi_{n}(\delta,A)(f)=\delta(fA)$ for $f\in \mathcal{C}^{\infty}(\mathcal{M}^{n})$. By using the notation for the complex number $\delta(A)=\chi_{n}(\delta,A)(1)$, where $1$ is the function identically $1$ over $\mathcal{M}^{n}$, and as a consequence of this, if we want to see that $\delta(A)=\delta(B)$ for two elements $A,B\in S^{n}\Gamma_{c}\omega S J\Phi$ then we need only to examine $\chi_{n}(\delta,A)$ and  $\chi_{n}(\delta,B)$ as elements of $D_{d}(\mathcal{M}^{n})$, and lastly conclude the equality by evaluating in the function $1$,  $\chi_{n}(\delta,A)(1)=\chi_{n}(\delta,B)(1)$ or $\delta(A)=\delta(B)$ as complex numbers.
 \par As a conclusion, if we want to prove that $\delta(A)=\delta(B)$ for two elements $A, B\in S\Gamma_{c}\omega S J\Phi$ then it is sufficient to prove it for $\chi(\delta,A)$ and  $\chi(\delta,B)$.
 
\end{rem}

\begin{rem}
\par A last comment concerning the physics picture of the Feynman measure  is that the smoothness on the diagonal is exactly the conservation of the momentum, because the momentum coordinates are the coordinates of the Fourier transform of a field. 
\end{rem}

\chapter{Renormalization}        
                  
\section{Characterization of the renormalization group}
   
  As we saw in the precedent chapter  the sheaf $S\omega SJ\Phi$ has a structure of comodule over $SJ\Phi$.
     
  \begin{definicion}
  A \emph{renormalization} is an automorphism of the sheaf of coalgebras $S\omega SJ\Phi$ which preserves the coaction of $SJ\Phi$. The set of renormalizations  is a group under composition and is called the renormalization group.
  \end{definicion}
  
  
     
     
  \par The next theorem characterizes the renormalization group in a nice manner. 
 
 \begin{thm}\label{teorema1}The elements of the renormalization group are in correspondence with the elements in  $Hom(S\omega SJ\Phi,\omega)$ that are zero over $S^{0}\omega SJ\Phi$ and are isomorphisms when they are restricted to $S^{1}\omega S^{0}J\Phi=\omega$.
 \end{thm}
 \begin{proof}
  By Proposition \ref{prop4} there is a correspondence between the coalgebra automorphisms $R\in Hom_{Coalg}(S\omega SJ\Phi, S\omega SJ\Phi) $ and the morphisms  $r\in  Hom_{\mathcal{C}^{\mathbb{R}}-Mod}(S\omega SJ\Phi,\omega SJ\Phi)$ whose sequential representation satisfies $r_{0}=0$ and $r_{1}$ is an isomorphism. 
   
 \par  At the beginning of this chapter we mentioned that $S\omega SJ\Phi$ is a comodule over $SJ\Phi$. We use Proposition \ref{prop5} in the monoidal abelian category of  sheafs over the ringed space $(\mathcal{M},\mathcal{C}^{\mathbb{R}})$, taking $D$ as the sheaf $SJ\Phi$, $M$ as the sheaf $S\omega SJ\Phi$ of $SJ\Phi$-comodules and $W$ as the sheaf $\omega$. It gives a correspondence between the comodule morphisms $r\in Hom_{Com_{SJ\Phi}}(S\omega SJ\Phi,\omega SJ\Phi)$ and the morphisms $\eta\in Hom_{\mathcal{C}^{\mathbb{R}}-Mod}(S\omega SJ\Phi,\omega)$. 
 \par If we begin this process with an $R\in Hom_{Coalg}(S\omega SJ\Phi, S\omega SJ\Phi) $ which is a renormalization, then it is a coalgebra morphism and also a comodule morphism so we successively apply both correspondences to obtain the desired total correspondence between the renormalizations and the subset of $\eta\in Hom_{\mathcal{C}^{\mathbb{R}}-Mod}(S\omega SJ\Phi,\omega)$ such that $\eta_{0}=0$, $\eta_{1}= (Id_{\omega}\otimes \varepsilon_{SJ\Phi})\circ r_{1}|_{\omega S^{0}J\Phi}$ is an isomorphism (because $r_{1}$ is).      
 \end{proof}       

\section{The structure of the renormalization group}

\par Notice that the renormalization group (which we will denote by $G$) preserves the increasing  filtration $S^{\leqslant n}\omega SJ\Phi\subseteq S^{\leqslant n+1}\omega SJ\Phi$, because if we apply $P\in G$ to an element  $v\in S^{\leqslant n}\omega SJ\Phi$, from equation \eqref{3.33} is clear that $P(v)\in S^{\leqslant n}\omega SJ\Phi$. 

\par If $\mathcal{C}$ is a monoidal abelian subcategory of $_{A}Mod$ or $\mathcal{C}^{\mathbb{R}}$-Modules, we denote by $G'$ the group of automorphisms of the cocommutative coaugmented coalgebra  $S^{c}V$ (with $V\in Obj(\mathcal{C})$)  and $G'_{> a}=\{ g\in G'/g(\alpha)=\alpha, \forall \alpha \in S^{\leqslant a}V\}$, where $a\in \mathbb{N}_{0}$. Notice that $G'_{>0}$ is exactly $G'$.  

\begin{prop}
Given $a\in \mathbb{N}$ and $P\in G'$. Then $P\in G'_{>a}$ if and only if the sequential representation of $P$ satisfies $p_{1}=Id_{V}$ and $p_{n}=0$, $\forall 2\leqslant n \leqslant a$.
\end{prop}
\begin{proof} The proof follows directly from equation \eqref{3.33}. If we apply the formula for $n=1$ and $P\in G'_{>a}$, then we obtain:
$$v=P(v)=\sum_{m=1}^{1}\sum_{I_{1}\neq \emptyset; I_{1}=\{1\} }\frac{1}{1!}p_{1}(v_{I_{1}})=p_{1}(v),$$
that is $p_{1}$ is the identity of $V$.
\par  Let us do the same for $n=2$ (supposing that $a\geqslant 2$), in this case the equation \eqref{3.33} gives us,
$$v_{1}\cdot v_{2}=\frac{1}{2}p_{1}(v_{1})\cdot p_{1}(v_{2})+\frac{1}{2}p_{1}(v_{2})\cdot p_{1}(v_{1})+p_{2}(v_{1}\cdot v_{2}),$$
where we use the previous result, $p_{1}(v_{1})\cdot p_{1}(v_{2})=v_{1}\cdot v_{2}$, and $P(v_{1}\cdot v_{2})=v_{1}\cdot v_{2}$ and so $p_{2}(v_{1}\cdot v_{2})=0$. 
\par Let us proceed by induction. Given $j\in \mathbb{N}$ such that $2\leqslant j\leqslant a$ and supposing that $p_{k}=0$ for all $k\in \mathbb{N}$ satisfying $2\leqslant k<j$ and $p_{1}=Id_{V}$, we write the equation \eqref{3.33}(taking $n=j$),
\begin{align*}
& v_{1}\cdots v_{j}=P(v_{1}\cdots v_{j})=\sum_{\text{all permutations}} \frac{1}{j!}p_{1}(v_{1})\cdots p_{1}(v_{j}) \\
& +\sum_{\text{terms with $p_{k}$ with $2\leqslant k<j$}}\cdots +p_{j}(v_{1}\cdots v_{j}) \\
& =j!\frac{1}{j!}v_{1}\cdots v_{j}+0+\cdots+0+p_{j}(v_{1}\cdots v_{j}),
\end{align*} 
hence $p_{j}(v_{1}\cdots v_{j})=0$ for all $v_{1}\cdots v_{j}\in S^{j}V$.\\

\par Reciprocally if the sequential representation of $P$ satisfies $p_{1}=Id_{V}$ and $p_{n}=0$ for all $2\leqslant n \leqslant a$ then if we apply $P$ to $v_{1}\cdots v_{n}\in S^{n}V$ by means of equation \eqref{3.33}, only the terms such that $|I_{j}|=1$ survive, that is when $m=n$. Hence,

\begin{align*}
& P(v_{1}\cdot v_{2}\cdots v_{n})=\sum_{m=1}^{n}\sum_{I_{1}\cdots I_{m}\neq \emptyset; I_{1}\sqcup \cdots \sqcup I_{m}=\{1\cdots n\} }\frac{1}{m!}p_{|I_{1}|}(v_{I_{1}})\cdots p_{|I_{m}|}(v_{I_{m}})\\
& =\sum_{I_{1}\cdots I_{n}\neq \emptyset; I_{1}\sqcup \cdots \sqcup I_{n}=\{1\cdots n\} }\frac{1}{n!}p_{1}(v_{I_{1}})\cdots p_{1}(v_{I_{n}})\\
& =\sum_{I_{1}\cdots I_{n}\neq \emptyset; I_{1}\sqcup \cdots \sqcup I_{n}=\{1\cdots n\} }\frac{1}{n!}Id_{V}(v_{I_{1}})\cdots Id_{V}(v_{I_{n}})\\
& =\sum_{I_{1}\cdots I_{n}\neq \emptyset; I_{1}\sqcup \cdots \sqcup I_{n}=\{1\cdots n\} }\frac{1}{n!}v_{1}\cdots v_{n}=\frac{n!}{n!}v_{1}\cdots v_{n}=v_{1}\cdots v_{n},
\end{align*} 
where we use the fact that there are $n!$ forms of take intervals $I_{i}$ such that $|I_{i}|=1$ and $I_{1}\sqcup \cdots \sqcup I_{n}=\{ 1,\cdots, n \}$.  
 \end{proof}

\par If $P\in G'$, we call $p$ its sequential representation and usually say that $p\in G'$.

\begin{prop}\label{bnbn}
\label{sub}$G'_{>a}$ is a subgroup of $G'$ for all $a\in \mathbb{N}$
\end{prop}
\begin{proof}
By composing with $\pi:S^{c}V\rightarrow V$ it is clear that  $1_{G'}$ has the sequential representation  $\{Id_{V},0,0,\cdots \}$, then by the precedent proposition $1_{G'}\in G_{>a}$ for all $a\in \mathbb{N}$.
\par If the sequential representation of  $P$ and $Q\in G'$ are $\{p_{n}\}_{\mathbb{N}}$ and $\{q_{n}\}_{\mathbb{N}}$ respectively, then the product representation of $P\circ Q$ has sequential representation  $\{p_{n}\}_{\mathbb{N}}\cdot \{q_{n}\}_{\mathbb{N}}=\{(p\cdot q)_{n}\}_{\mathbb{N}}$ given by the following formula (see equation \eqref{composicion}),

$$(p\cdot q)_{n}(v_{1}\cdots v_{n})=\sum_{m=1}^{n}\frac{1}{m!}\sum_{I_{1}\cdots I_{m}\neq \emptyset; I_{1}\sqcup \cdots \sqcup I_{m}=\{1\cdots n\} }p_{m}(q_{|I_{1}|}(v_{I_{1}})\cdots q_{|I_{m}|}(v_{I_{m}})).$$

With this formula it is possible to show that $p\cdot q\in G'_{>a}$ provided that $p$ and $q$ are there. For example if $2\leqslant n \leqslant a$, having in account that  $q_{k}=0$ for all $2\leqslant k\leqslant a$, we see that in the sum only the terms with all $q_{1}$ survive, i.e. $m=n$. Also as $q_{1}$ is the identity: $q_{1}(v_{I_{1}})\cdots q_{1}(v_{I_{n}})=v_{I_{1}}\cdots v_{I_{n}}$, but $v_{I_{i}}=v_{j_{i}}$ because $|I_{i}|=1$ when $m=n$, then $q_{1}(v_{I_{1}})\cdots q_{1}(v_{I_{n}})=v_{j_{1}}\cdots v_{j_{n}}$. But as $\cdot$ is a symmetric product then all the terms $v_{j_{1}}\cdots v_{j_{n}}$ are equal, so there are $n!$ summands which are equal and consequently the last equality can be written as
$$(p\cdot q)_{n}(v_{1}\cdots v_{n})=n!\frac{1}{n!}p_{n}(v_{1}\cdots v_{m}).$$
\par This is zero because $p_{n}=0$. From the same equation \eqref{composicion}, but this time taking  $n=1$, we can conclude $(p\cdot q)_{1}=Id_{V}$.
\par Finally if $p\in G'_{>a}$, we take $q=p^{-1}$ in $G'$. Suppose $q\notin G'_{>a}$, and let $n_{0}$ be the first natural $\geqslant 2$ such that $q_{n_{0}}\neq 0$. Equation \eqref{composicion} show us that $q_{1}=p_{1}^{-1}$ (taking $n=1$). But $q\notin G'_{>a}$, then $n_{0}\leqslant a$. Then by using again the equation \eqref{composicion} taking $n=n_{0}$ and using the fact that  $\{(p\cdot q)_{n}\}_{\mathbb{N}}=\{Id_{V},0,0\cdots \}$ we obtain $0=p_{1}(q_{n_{0}}(v_{1}\cdots v_{n_{0}}))$.
But as $p_{1}=Id_{V}$ the equation turns to: $0=q_{n_{0}}(v_{1}\cdots v_{n_{0}})$  i.e. $q_{n_{0}}=0$, which is absurd. Hence $q\in G'_{>a}$. 
\end{proof}

\par Notice that the precedent proposition is trivially valid for the case $a=0$.
\par For $a\geqslant 2$ we define the sets
\begin{align*}
& G'_{a}=\{P\in G'/ p=\{Id_{V},0\cdots,0,p_{a},0\cdots\}\},\\
\text{and}\hspace{1cm}\\
& G'_{1}=\{P\in G'/ p=\{p_{1},0\cdots\}\}
\end{align*}
were we use the just described notation of sequential representation introduced in Chapter 1. An element in $P\in G'_{a}$  will be an automorphism $P\in G'$ such that when we restrict $\pi\circ P$ to $S^{n}V$ it is zero ($\pi:S^{n}V\to V$ is the canonical projection), and less $n=1$ or $n=a$.
\par The next proposition is proved by induction on $a$.

\begin{prop}\label{prop8}Let  $a\in \mathbb{N}$ and $P\in G'$. There are unique  $P_{1}, P_{2},\cdots P_{a}, P'\in G'$ such that  $P_{i}\in G'_{i}$, $P'\in G'_{>a}$ and $P=P_{1}\circ P_{2}\cdots P_{a}\circ P'$
\end{prop}

\par The proof is deduced from the next lemma:

\begin{lemma}\label{lema1}Given $a\in \mathbb{N}$ and  $P\in G'_{>a}$ there exist unique $P_{a+1}$ and $P'\in G'$ such that $P_{a+1}\in G'_{a+1}$, $P'\in G'_{>a+1}$ and $P=P_{a+1}\circ P'$.
\end{lemma}
\begin{proof} We give the proof for the case $a\geqslant 2$, the case $a=1$ is similar. If the representation of  $P$ is given by  $\{Id_{V},0\cdots 0,p_{a+1},p_{a+2}\cdots \}$, we take $P_{a+1}$ as the only one whose sequential representation is $\tilde{p}=\{Id_{V},0\cdots p_{a+1},0 \cdots \}$ and  $P'$ whose representation is $\{(\tilde{p}^{-1}\cdot p)_{n}\}_{n\in \mathbb{N}}$. By using the equation \eqref{composicion} it is not difficult to see that $(\tilde{p}^{-1}\cdot p)_{1}=Id_{V}$, and as $\tilde{p},p\in G'_{>a}$ (which is a group as we saw in Proposition \ref{sub}) then $(\tilde{p}^{-1}\cdot p)_{n}=0$ for all $2\leqslant n \leqslant a$. By means of \eqref{composicion} and computing $\tilde{p}^{-1}$ up to degree $a+1$ with \eqref{3.35}, we can see that  $\tilde{p}^{-1}=\{Id_{V},0\cdots 0,-p_{a+1},\cdots\}$ and $(\tilde{p}^{-1}\cdot p)_{a+1}=0$, so $\tilde{p}^{-1}\cdot p \in G'_{>a+1}$. Let us see in detail that, if $\tilde{p}=\{Id_{V},0\cdots p_{a+1},0 \cdots \}$ then its inverse satisfies $(\tilde{p}^{-1})_{1}=(\tilde{p}_{1})^{-1}=Id_{V}^{-1}=Id_{V}$ and then by using the equation \eqref{inv} for $2\leqslant k\leqslant a$ we have

\begin{align*}
& (\tilde{p}^{-1})_{k}(v_{1}\cdot v_{2}\cdots v_{k})\\ 
& =-\sum_{m=2}^{k}\frac{1}{m!}\sum_{I_{1}\cdots I_{m}\neq \emptyset ; I_{1}\sqcup \cdots \sqcup I_{m}=\{1\cdots k\} }(\tilde{p}^{-1})_{1}(\tilde{p}_{m}((\tilde{p}^{-1})_{|I_{1}|}(v_{I_{1}})\cdots (\tilde{p}^{-1})_{|I_{m}|}(v_{I_{m}}))),\\
\end{align*}
but as we say $(\tilde{p}^{-1})_{1}$ is the identity and then 
\begin{align*}
& (\tilde{p}^{-1})_{k}(v_{1}\cdot v_{2}\cdots v_{k})\\ 
& =-\sum_{m=2}^{k}\frac{1}{m!}\sum_{I_{1}\cdots I_{m}\neq \emptyset ; I_{1}\sqcup \cdots \sqcup I_{m}=\{1\cdots k\} }\tilde{p}_{m}((\tilde{p}^{-1})_{|I_{1}|}(v_{I_{1}})\cdots (\tilde{p}^{-1})_{|I_{m}|}(v_{I_{m}})),\\
\end{align*}
which is zero because all the $\tilde{p}_{m}$ has $m\leqslant k\leqslant a$, and then   $\tilde{p}_{m}=0$.
\par Lastly if we take $n=a+1$ in \eqref{inv}, we obtain: 

\begin{align*}
& (\tilde{p}^{-1})_{a+1}(v_{1}\cdot v_{2}\cdots v_{a+1})\\ 
& =-\sum_{m=2}^{a+1}\frac{1}{m!}\sum_{I_{1}\cdots I_{m}\neq \emptyset ; I_{1}\sqcup \cdots \sqcup I_{m}=\{1,\cdots, a+1\} }\tilde{p}_{m}((\tilde{p}^{-1})_{|I_{1}|}(v_{I_{1}})\cdots (\tilde{p}^{-1})_{|I_{m}|}(v_{I_{m}})).
\end{align*}
By the same aforementioned considerations all the summands are zero except the only ones in which appears $\tilde{p}_{a+1}$, so
\begin{align*}
& (\tilde{p}^{-1})_{a+1}(v_{1}\cdot v_{2}\cdots v_{a+1})\\ 
& =-\frac{1}{(a+1)!}\sum_{I_{1}\cdots I_{a+1}\neq \emptyset ; I_{1}\sqcup \cdots \sqcup I_{a+1}=\{1,\cdots, a+1\} }\tilde{p}_{a+1}((\tilde{p}^{-1})_{|I_{1}|}(v_{I_{1}})\cdots (\tilde{p}^{-1})_{|I_{a+1}|}(v_{I_{a+1}})).
\end{align*}
The only way such that $a+1$ non-empty sets have disjoint union equal to $\{ 1,\cdots, a+1 \}$ is that each one of them has cardinal $1$ and then for all $1\leqslant j\leqslant a+1$ we have $(\tilde{p}^{-1})_{|I_{j}|}(v_{I_{j}})=(\tilde{p}^{-1})_{1}(v_{i_{j}})=I_{V}(v_{i_{j}})=v_{i_{j}}$. Consequently,
\begin{align*}
 & (\tilde{p}^{-1})_{a+1}(v_{1}\cdot v_{2}\cdots v_{a+1})\\
 & =-\frac{1}{(a+1)!}\sum_{I_{1}\cdots I_{a+1}\neq \emptyset ; I_{1}\sqcup \cdots \sqcup I_{a+1}=\{1,\cdots, a+1\} }\tilde{p}_{a+1}(v_{i_{1}}\cdots v_{i_{a+1}}).
 \end{align*}
Since the $\cdot$ is a symmetric product 
we can write,
\begin{eqnarray*}
  (\tilde{p}^{-1})_{a+1}(v_{1}\cdot v_{2}\cdots v_{a+1})=-\tilde{p}_{a+1}(v_{1}\cdots v_{a+1}).
\end{eqnarray*}
 
 That is $ (\tilde{p}^{-1})_{a+1}=-\tilde{p}_{a+1}$. Now we know that  $\tilde{p}^{-1}=\{Id_{V},0\cdots 0,-p_{a+1},\cdots\}$, it is no so difficult to see that  $(\tilde{p}^{-1}\cdot p)_{a+1}=0$ by using the equation \eqref{composicion} for $n=a+1$.

\par The uniqueness is deduced from a similar argument, if $P=P_{a+1}\cdot P' =Q_{a+1}\cdot Q'$ with $P_{a+1}, Q_{a+1}\in G'_{a+1}$ and $P', Q'\in G'_{>a+1}$ then $P_{a+1}^{-1}\cdot Q_{a+1}=P'\cdot Q'^{-1}$.

 If the sequential representation of $P_{a+1}$ is $\{ Id_{V},0\cdots,\tilde{p}_{a+1},0, \cdots \}$ the element $a+1$ of the sequential representation of $P^{-1}_{a+1}$ is $-\tilde{p}_{a+1}$ and if $Q_{a+1}$ has representation $\{ Id_{V},0\cdots,\tilde{q}_{a+1},0, \cdots \}$ so the $a+1$-element of the representation of $P^{-1}_{a+1}\cdot Q_{a+1}$ is $-\tilde{p}_{a+1}+\tilde{q}_{a+1}$, but $P_{a+1}^{-1}\cdot Q_{a+1}\in G'_{>a+1}$ so $-\tilde{p}_{a+1}+\tilde{q}_{a+1}=0$ and $P_{a+1}=Q_{a+1}$, which says us that $P'=Q'$. 
\end{proof}

\par The proof of Proposition \ref{prop8} is deduced by repeated application of the precedent lemma.

\par For Proposition \ref{prop8} and Lemma \ref{lema1} there exist similar statements in which the order of the compositions is the opposite. We rewrite Proposition \ref{prop8} in these terms:

{\prop\label{prop9}{Let  $a\in \mathbb{N}$ and $P\in G'$. There are unique  $P_{1}, P_{2},\cdots P_{a}, P'\in G'$ such that  $P_{i}\in G'_{i}$, $P'\in G'_{>a}$ and $P=P'\circ P_{a}\circ P_{a-1}\cdots P_{1}$
}}

{\prop{$G'_{>a}$ is a normal subgroup of $G'$ for all $a\in \mathbb{N}$.}}
\begin{proof}
 Thanks to Propositions \ref{bnbn} and \ref{prop8} it is sufficient to prove that $P^{-1}\cdot G'_{>a}\cdot P \subseteq G'_{>a}$, for all $P\in G'_{i}$ with $1\leqslant i \leqslant a$.
\par If $i=1$ and $Q\in G'_{>a}$ then abusing the notation $Q=\{ Id_{V},0\cdots,0,q_{a+1},\cdots \}$,  and with \eqref{composicion} is deduced $P^{-1}\cdot Q\cdot P=\{ Id_{V},0\cdots,0, p_{1}\circ q_{a+1}\circ p_{1}^{\odot (a+1)},\cdots \}$ which is in $G'_{>a}$  (where $\odot$ denotes the symmetric product).
\par If $2\leqslant i \leqslant a$; $Q=\{ Id_{V}, 0,\cdots,0,q_{a+1},\cdots \}$ and $P=\{ Id_{V},0,\cdots,0,p_{i},0\cdots \}$ by using \eqref{composicion} obtain $Q\cdot P=\{ Id_{V},0,\cdots,0,p_{i},0, \cdots ,0,q'_{a+1},\cdots \}$.
\par By using \eqref{3.35} $P^{-1}=\{ Id_{V},0\cdots,0,-p_{i},r_{i+1},\cdots, r_{a},\cdots \}$.But when we multiply $P^{-1}\cdot Q\cdot P$ the first $a$ elements are calculated in the same manner which the first $a$ elements of $P^{-1}\cdot P$, because the elements with order bigger than $a$ do not appear. So $(\pi\circ P^{-1}\cdot Q \cdot P)|_{S^{n}V}=(\pi\circ P^{-1}\cdot P)|_{S^{n}V}$ for all $1\leqslant n \leqslant a$, which implicates that $P^{-1}\cdot Q\cdot P\in G'_{>a}$.   
\end{proof}

\par The preceding proposition is valid in the case $a=0$ too but in this case the proof is immediate. The next corollary is a consequence of the above results.

\begin{coro}\label{iso}
For all $a\in \mathbb{N}$ we have  $G'/G'_{>a}\simeq G'_{\leqslant a}$. 
\end{coro} 

 \begin{prop}For all $a,b\in \mathbb{N}$; $[G'_{>a},G'_{>b}]\subseteq G'_{>a+b}$
 \end{prop}
 \begin{proof}
 Without loss of generality we can suppose $a\leqslant b$. Let $P\in G'_{>a}$ and $Q\in G'_{>b}$. Then the sequential representation of $P$ is $\{ p_{n}\}_{n\in \mathbb{N}}=\{Id_{V},0,\cdots,0,p_{a+1},p_{a+2},\cdots\}$ and similarly for $Q$,   $\{ q_{n}\}_{n\in \mathbb{N}}=\{Id_{V},0,\cdots,0,q_{b+1},q_{b+2},\cdots\}$. By the composition formula we have
 
 \begin{equation}\label{ñññ}
 (p\cdot q)_{n}(v_{1}\cdots v_{n})=\sum_{m=1}^{n}\frac{1}{m!}\sum_{I_{i}\neq \emptyset; \sqcup_{i=1}^{m}I_{i}=\{ 1,\cdots,n\}}p_{m}(q_{|I_{1}|}(v_{I_{1}})\cdots q_{|I_{m}|}(v_{I_{m}}) )
 \end{equation}
 
 If $n=1$ in \eqref{ñññ} we have $(p\cdot q)_{1}=Id_{V}$. For $n\in \{ 2,\cdots,a \}$, in the sum over $m$ only the terms with $m=1$ survive, because $p_{j}=0$ for $2\leqslant j\leqslant n \leqslant a$, then $(p\cdot q)_{n}(v_{1}\cdots v_{n})=q_{n}(v_{1}\cdots v_{n})$ which is equal to zero because of $2\leqslant n \leqslant a\leqslant b$.
 
 For $a+1\leqslant n\leqslant b$, in \eqref{ñññ} only the terms with $m=1, a+1,a+2,\cdots n$ survive, so
 \begin{equation}\label{kbkb}
 (p\cdot q)_{n}(v_{1}\cdots v_{n})=q_{n}(v_{1}\cdots v_{n}) +\sum_{m=a+1}^{n}\frac{1}{m!}\sum_{I_{i}\neq \emptyset; \sqcup_{i=1}^{m}I_{i}=\{ 1,\cdots,n\}}p_{m}(q_{|I_{1}|}(v_{I_{1}})\cdots q_{|I_{m}|}(v_{I_{m}}) )
 \end{equation}
 but $q_{n}=0$. Moreover, if $m\in \{ a+1,\cdots, n-1 \}$ then there exists at least one $I_{i_{0}}$ such that $|I_{i_{0}}|\geqslant 2$ (and $|I_{i_{0}}|\leqslant b$) and consequently $q_{|I_{i_{0}}|}=0$ and \eqref{kbkb} becomes 
 \begin{equation*}
 (p\cdot q)_{n}(v_{1}\cdots v_{n})=p_{n}(v_{1}\cdots v_{n}).
 \end{equation*}
 
 Lastly if $b+1\leqslant n \leqslant a+b$, as before in the sum over $m$ only the terms with $m=1,a+1,a+2,\cdots n$ persist. Then with similar arguments as in the previous cases we have,
 \begin{equation*}
 (p\cdot q)_{n}(v_{1}\cdots v_{n})=q_{n}(v_{1}\cdots v_{n})+p_{n}(v_{1}\cdots v_{n}).
 \end{equation*}
 
 So we have 
 \begin{equation}\label{aaa}
 \{ (p\cdot q)\}_{n\in \mathbb{N}}=\{ Id_{v},0,\cdots,0, p_{a+1},\cdots, p_{b},p_{b+1}+q_{b+1},\cdots,p_{b+a}+q_{b+a},\cdots \}.
 \end{equation}
 
 The computation of the sequential representation of $Q\circ P$ is a bit easier and gives us the same result up to the $a+b$ term.
 For this composition we have the formula
 \begin{equation}\label{zz}
 (q\cdot p)_{n}(v_{1}\cdots v_{n})=\sum_{m=1}^{n}\frac{1}{m!}\sum_{I_{i}\neq \emptyset; \sqcup_{i=1}^{m}I_{i}=\{ 1,\cdots,n\}}q_{m}(p_{|I_{1}|}(v_{I_{1}})\cdots p_{|I_{m}|}(v_{I_{m}}) ).
 \end{equation}
 
 It is easy to check the $(q\cdot p)_{1}=Id_{V}$, and that if $n\in \{ 2,\cdots,a \}$ then  $(q\cdot p)_{n}=0$, because \eqref{zz} becomes 
 \begin{equation}
 (q\cdot p)_{n}(v_{1}\cdots v_{n})=q_{1}(p_{|I_{1}|}(v_{I_{1}}))=p_{n}(v_{1}\cdots v_{n})=0,
 \end{equation}for $I_{1}=\{ 1,\cdots ,n \}$.
 
 For  $n\in \{ a+1,\cdots,b \}$ we have 
 \begin{equation}
 (q\cdot p)_{n}(v_{1}\cdots v_{n})=q_{1}(p_{|I_{1}|}(v_{I_{1}}))=p_{n}(v_{1}\cdots v_{n}),
 \end{equation}but this time $p_{n}\neq 0$, because $n\geqslant a+1$.
 
 Finally for $n\in \{ b+1,\cdots , b+a \}$ in \eqref{zz} only the terms with $m=1,b+1,b+2,\cdots,n$ survive so, 
 
 \begin{equation*}
 (q\cdot p)_{n}(v_{1}\cdots v_{n})=q_{1}(p_{n}(v_{1}\cdots v_{n}))+\sum_{m=b+1}^{n}\frac{1}{m!}\sum_{I_{i}\neq \emptyset; \sqcup_{i=1}^{m}I_{i}=\{ 1,\cdots,n\}}q_{m}(p_{|I_{1}|}(v_{I_{1}})\cdots p_{|I_{m}|}(v_{I_{m}}) ),
 \end{equation*} and for $m\in \{ b+1,\cdots,n-1 \}$ there exists at least one $I_{i_{0}}$ with more than one element and less then $a+1$, hence $p_{|I_{i_{0}}|}=0$. That is, $(q\cdot p)_{n}(v_{1}\cdots v_{n})=p_{n}(v_{1}\cdots v_{n})+q_{n}(v_{1}\cdots v_{n})$.
 
 Then the sequential representation of $Q\circ P$ is the same that the right hand side of \eqref{aaa} up to the $a+b$-term.
 Whence $P\circ Q=Q\circ P$ in $G'_{\leqslant a+b}$, which proves what we want.
 \end{proof}

 \par  We have just seen that it has sense to talk about $G'/G'_{>a}$ for all $a\in \mathbb{N}$ it is possible to see that  $G'$ is the inverse limit of  $G'/G'_{>n}$.
  
{\lemma{$G'$ is the direct limit of the projective system $\{G'/G'_{>n}\}_{n\in \mathbb{N}}$ .}}
\begin{proof}
 If $a<b$ are natural numbers then we can define $\pi_{ab}:G'/G'_{>b}\rightarrow G'/G'_{>a}$ as $\pi_{ab}(P_{1}P_{2}\cdots P_{b})=P_{1}P_{2}\cdots P_{a}$ where we have used that each element that belongs to $G'/G'_{>b}$ can be written uniquely as a product by Lemma \ref{iso}. This function clearly satisfies $\pi_{ab}\circ \pi_{bc}=\pi_{ac}$ and so we have a projective system.
\par Moreover the morphisms $\pi_{a}:G'\rightarrow G'/G'_{>a}$ such that $\pi_{a}(P)=[P]_{G'/G'_{>a}}=P_{1}\cdots P_{a}$ commutes with $\pi_{ab}$ ($a<b$) because $\pi_{ab}\circ\pi_{b}(P)=\pi_{ab}(P_{1}\cdots P_{b})=P_{1}\cdots P_{a}=\pi_{a}(P)$.
\par Suppose that $W$ is another group with a family of arrows $f_{a}:W \rightarrow G'/G'_{>a}$ satisfying $\pi_{ab}\circ f_{b}=f_{a}$ too, for $a<b$. We can define a unique morphism $f:W\rightarrow G'$ such that the following diagram 
 \begin{eqnarray*}
\xymatrix{
G' \ar[rr]^{\pi_{b}}& & G'/G'_{>b} \ar[d]^{\pi_{ab}}\\
W \ar@{-->}[u]^{f} \ar[urr]^{f_{b}} \ar[rr]_{f_{a}} & & G'/G'_{>a}
}
  \end{eqnarray*} commutes.
\end{proof}

\par We will define a topology on $G'$ (its projective topology, the finest such that $\pi_{n}$ are all continuous), where a subbase of the open neighbourhoods of zero is given by the $G'_{>n}$. Then in the above diagram we can define as $f(w)=\lim_{n\to \infty}f_{n}(w)\in G'/G'_{>n}=G'_{1}\cdots G'_{n}\subseteq G'$ which exists because $f_{n}$ commutes with $\pi_{ab}$. With this definition we have $\pi_{b}(f(w))=\pi_{b}(\lim_{n\to \infty}f_{n}(w))=\lim_{n\to \infty}\pi_{b}f_{n}(w)$ where used the continuity of $\pi_{b}$ with respect to this topology. Note that for all  $n>b$,  $f_{n}(w)$ has more than $b$ "elements", and so $\pi_{b}(f_{n}(w))=f_{b}(w)$ and consequently $\pi_{b}\circ f=f_{b}$.

\par As a major consequence of the last results we can stablish the next corollary, its proof is immediate.

{\coro{ $G'$ is equal to $G'_{1}G'_{2}G'_{3}\cdots$in the sense that any element $P\in G'$ can be written as an infinite product $P=P_{1}\circ P_{2}\circ P_{3}\cdots$ of elements $P_{i}\in G'_{i}$ and conversely each of this infinite products in such conditions belongs to $G'$.}}
\begin{proof} 
Given $P\in G'$, call $P_{i}$ the morphisms belonging to $G'_{i}$ proportioned by Proposition \ref{prop8}. The limit $P_{1}\circ P_{2}\circ P_{3}\cdots$ (i.e. $\lim_{n\to \infty} P_{1}\cdots P_{n}$ ) where $P_{i}\in G'_{i}$ converges in $G'$ because a subbase of open sets for the projective topology in $G'$ are $G'_{>k}$ for $k\in \mathbb{N}$ and then the difference between $P$ and $P_{1}\cdots P_{n}$ belongs to any $G'_{>a}$, provided $n$ is sufficiently large.  
\end{proof}
\par This corollary has also a version with the products in the opposite order. 

\begin{coro}\label{hhhhhh}
$G'$ is equal to $\cdots G'_{3}G'_{2}G'_{1}$in the sense that any element $P\in G'$ can be written as an infinite product $P=\cdots P_{3}\circ P_{2}\circ P_{1}$ of elements $P_{i}\in G'_{i}$ and conversely each of this infinite products in such conditions belongs to $G'$.
\end{coro} 

\par  It is of central interest for us the case $V=\omega SJ\Phi$. In this case $G'$ is not the renormalization group, because an element belonging to $G'$ may not preserve the coaction of $SJ\Phi$.\\




\par The natural map $S\Gamma \omega SJ\Phi\rightarrow \Gamma S\omega SJ\Phi$ is not an isomorphism because as we mentioned, there is a difference between the symmetric product in the category of sheafs of modules over the ringed space $(\mathcal{M},\mathcal{C}^{\mathbb{R}})$ and the symmetric product in the category of $\mathbb{R}$-vector spaces. For example if we denote by $\odot$ the symmetric product in the category  of sheafs of modules over the ringed space $(\mathcal{M},\mathcal{C}^{\mathbb{R}})$ and by $\cdot$ in the category of   $\mathbb{R}$-vector spaces, the natural map sends $\psi \cdot \eta\rightarrow \psi\odot\eta$ where $\psi$  and  $\eta$ are lagrangian densities.
\par The difference between these product is that if $f$ is a smooth function over $\mathcal{M}$ then $f\psi \cdot \eta\neq \psi\cdot f\eta$ but $f\psi \odot \eta=\psi\odot f\eta$, hence the natural map is not injective.

{\lemma\label{lema7}{The induced action of a renormalization over $\Gamma S\omega SJ\Phi$ can be lifted to an action over $S\Gamma \omega SJ\Phi$ preserving the coproduct, the coaction of $\Gamma SJ\Phi $ and the product of elements with disjoint support.}}
\begin{proof} By Theorem \ref{df} $\Gamma \omega SJ\Phi\simeq \Gamma\omega \otimes_{\mathcal{C}^{\infty}(\mathcal{M})}\Gamma SJ\Phi$, then applying Proposition \ref{prop 1}, $\Gamma\omega \otimes_{\mathcal{C}^{\infty}(\mathcal{M})}\Gamma SJ\Phi$ can be seen as a $\Gamma SJ\Phi$-comodule over $\mathcal{C}^{\infty}(\mathcal{M})$.

\par Call $N$ the natural map  $S\Gamma \omega SJ\Phi\rightarrow \Gamma S\omega SJ\Phi$, and $\varphi$ the isomorphism existent between $\Gamma S\omega SJ\Phi$ and $S_{\mathcal{C^{\infty}(\mathcal{M})}}\Gamma \omega SJ\Phi$ (see Corollary \ref{dm}). In this context, given a renormalization $P$, we still call $P$ the induced map on the sections of $S\omega SJ\Phi$ and then we consider the composition $p:=\pi_{1}\circ \varphi\circ P\circ N :S\Gamma \omega SJ\Phi \rightarrow \Gamma\omega SJ\Phi$.
\par By using Proposition \ref{prop4},  $p:S\Gamma \omega SJ\Phi \rightarrow \Gamma\omega SJ\Phi$  is in correspondence with a morphism which we call $P':S\Gamma \omega SJ\Phi \rightarrow S\Gamma\omega SJ\Phi$ and is written (as in Proposition \ref{prop3} )
\begin{equation}\label{dhdh}
P'(c)=\sum_{n\in\mathbb{N}}\frac{1}{n!}p^{\odot n}\circ \Delta_{\overline{C}}^{(n)}(c),
\end{equation} 
where $C=S\Gamma \omega SJ\Phi$ and $\Delta^{(n)}_{\overline{C}}$ was defined in Section \ref{sec111}.
\par Now we prove that $P'$ preserves the product of elements with disjoint support, so let $\varphi=\prod_{i=1}^{k}f_{i}$ and $\psi=\prod_{l=1}^{r}g_{l}$ be elements of $S\Gamma \omega SJ\Phi$ with disjoint support, $\varphi \cdot \psi \in S\Gamma \omega SJ\Phi$, call $\chi_{i}=f_{i}$ for all $1\leqslant i\leqslant k$ and $\chi_{k+l}=g_{l}$ for all $1\leqslant l\leqslant r$. From the statement over the supports we can suppose that $supp(f_{i})\subseteq U$ and $supp(g_{l})\subseteq V$ for some $U$ and $V$ disjoint open sets in $\mathcal{M}$.
\par Consider $$\Delta_{S\Gamma \omega SJ\Phi}^{(n)}(\varphi\cdot \psi)=\sum_{I_{1},\cdots,I_{n}\neq \emptyset; I_{1} \sqcup \cdots \sqcup I_{n}=\{1,\cdots,k+r \}} \chi_{I_{1}}\otimes \cdots \otimes \chi_{I_{n}},$$ so by applying $p^{\odot n}$ we obtain

\begin{equation}\label{343}
p^{\odot n}\circ\Delta_{S\Gamma \omega SJ\Phi}^{(n)}(\varphi\cdot \psi)=\sum_{I_{1},\cdots,I_{n}\neq \emptyset; I_{1} \sqcup \cdots \sqcup I_{n}=\{1,\cdots,k+r \}} p(\chi_{I_{1}})\odot \cdots \odot p(\chi_{I_{n}}),
\end{equation}

\par Let $I_{k}$ be a subset of $\{1,\cdots,k+r \}$ such that has at least one index $i_{0}\in \{1,\cdots,k\}$ and one index $j_{0}\in \{k+1,\cdots,k+r\}$, and $h\in \mathcal{C}_{0}^{\infty}(\mathcal{M})$ such that $h|_{Sup(\varphi)}\equiv 1$ and $h|_{Sup(\psi)}\equiv 0$ , this function exists because the spacetime is paracompact. Then,
\begin{align*}
& N(\chi_{I_{k}})=N(\chi_{\alpha}\cdots \chi_{i_{0}}\cdots\chi_{j_{0}}\cdots \chi_{|I_{k}|}) \\
& =\chi_{\alpha}\odot\cdots \odot \chi_{i_{0}}\odot \cdots\odot \chi_{j_{0}}\odot\cdots\odot \chi_{|I_{k}|}\\
& =\chi_{\alpha}\odot\cdots \odot h\chi_{i_{0}}\odot \cdots\odot \chi_{j_{0}}\odot\cdots\odot \chi_{|I_{k}|}\\
& =\chi_{\alpha}\odot\cdots \odot \chi_{i_{0}}\odot \cdots\odot h\chi_{j_{0}}\odot\cdots\odot \chi_{|I_{k}|}\\
& =\chi_{\alpha}\odot\cdots \odot \chi_{i_{0}}\odot \cdots\odot 0 \odot\cdots\odot \chi_{|I_{k}|}=0
\end{align*}where $\alpha$ is the the first element of $I_{k}$.

\par Then, in equation \eqref{343} only terms such that all the $I_{j}$ are totally included in $\{1,\cdots,k\}$ or in $\{k+1,\cdots,k+r\}$ survive. We will call $I$ to the first and $J$ to the seconds.
\par In the equation \eqref{343},  we have terms with only  one $I$ and $n-1$ subsets $J$, or terms with two $I$ and $n-2$ $J$'s, three $I$'s, four $I$'s ... $n-1$. In the following sums all the $I$'s are contained in $\{1,\cdots k\}$ and the $J$'s in $\{ k+1,\cdots k+r \}$, all of them are $\neq \emptyset$ and its disjoint union is $\{1,\cdots, k+r \}$.

\begin{align*}
& p^{\odot n}\circ\Delta_{S\Gamma \omega SJ\Phi}^{(n)}(\varphi \cdot \psi) \\
& =\sum_{J_{1}\cdots J_{n-1}} n [p(\chi_{I_{1}})\odot p(\chi_{J_{1}})\odot \cdots \odot p(\chi_{J_{n-1}})] \\
&+\sum_{I_{1}, I_{2}; J_{1}\cdots J_{n-3}=\{ k+1,\cdots,k+r \}} \binom{n}{2} [p(\chi_{I_{1}})\odot p(\chi_{I_{2}})\odot p(\chi_{J_{1}})\odot \cdots \odot p(\chi_{J_{n-2}})] \\
&+\sum_{I_{1}\cdots  I_{3}; J_{1}\cdots J_{n-3}} \binom{n}{3} [p(\chi_{I_{1}})\odot \cdots \odot  p(\chi_{I_{3}})\odot p(\chi_{J_{1}})\odot \cdots \odot p(\chi_{J_{n-3}})]+\cdots ,
\end{align*}
where the binomial coefficients appear because of the possible reorderings of the expression. We will make the computations for $n=1,2,3,4...$ For the general case one must proceed with a combinatorial argument as in the proof of Proposition \ref{prop3}. For the case $\textbf{n=1}$ the result of the above computation is zero because there are no subset $J$.
 \par Case $\textbf{n=2}$,   $p^{2}\circ\Delta_{S\Gamma \omega SJ\Phi}(\varphi\cdot \psi)=2p(\chi_{I_{1}})\odot p(\chi_{J_{1}})=2p(\varphi)\odot p(\psi)$.
 \par Case $\textbf{n=3}$: 
 $$p^{\odot 3}\circ \Delta_{S\Gamma \omega SJ\Phi}^{(3)}(\varphi\cdot\psi)=3p(\varphi)\odot[\sum_{J_{1},J_{2}}p(\chi_{J_{1}})\odot p(\chi_{J_{2}})]+3[\sum_{J_{1},J_{2}}p(\chi_{I_{1}})\odot p(\chi_{I_{2}})]\odot p(\psi)$$
 $$ =3p(\varphi)\odot p^{\odot 2}\circ \Delta_{S\Gamma \omega SJ\Phi}^{(2)}(\psi) +3p^{\odot 2}\circ\Delta_{S\Gamma \omega SJ\Phi}^{(2)}(\varphi)\odot p(\psi)$$

\par For $\textbf{n=4}$,
\begin{align*}
& p^{\odot 4}\circ \Delta_{S\Gamma \omega SJ\Phi}^{(4)}(\varphi\cdot\psi)=4p(\varphi)\odot p^{\odot 3}\circ \Delta_{S\Gamma \omega SJ\Phi}^{(3)}(\psi)\\
& +\binom{4}{2}p^{\odot 2}\circ \Delta_{S\Gamma \omega SJ\Phi}^{(2)}(\varphi) \odot p^{\odot 2}\circ \Delta_{S\Gamma \omega SJ\Phi}^{(2)}(\psi)+4p^{\odot 3}\circ \Delta_{S\Gamma \omega SJ\Phi}^{(3)}(\varphi)\odot p(\psi)
\end{align*} 

\par Multiplying the last equations by $\frac{1}{n!}$ and adding them as in equation \eqref{dhdh}, we can see how are the first terms of $P'(\varphi\cdot \psi)$, using easy facts as $\frac{1}{4!}\binom{4}{2}=\frac{1}{2}\frac{1}{2}$. Then we conclude,

\begin{align*}
& P'(\varphi\cdot \psi)=p(\varphi)\odot p(\psi)+\frac{1}{2}p(\varphi)\odot p^{\odot 2}\circ \Delta_{S\Gamma \omega SJ\Phi}^{(2)}(\psi) +\frac{1}{2}p^{\odot 2}\circ\Delta_{S\Gamma \omega SJ\Phi}^{2}(\varphi)\odot p(\psi)\\
& +\frac{1}{3!}p(\varphi)\odot p^{\odot 3}\circ \Delta_{S\Gamma \omega SJ\Phi}^{(3)}(\psi) 
+\frac{1}{2}\frac{1}{2}p^{\odot 2}\circ \Delta_{S\Gamma \omega SJ\Phi}^{(2)}(\varphi) \odot p^{\odot 2}\circ \Delta_{S\Gamma \omega SJ\Phi}^{(2)}(\psi)
\end{align*}
\begin{equation}\label{producto}
 +\frac{1}{3!}p^{\odot 3}\circ \Delta_{S\Gamma \omega SJ\Phi}^{(3)}(\varphi)\odot p(\psi)+\cdots
\end{equation}

\par On the other hand if we compute $P'(\varphi)\odot P'(\psi)$, we have

\begin{equation*}
\left(\sum_{k\in\mathbb{N}}\frac{1}{k!}p^{\odot k}\circ \Delta_{S\Gamma \omega SJ\Phi}^{(k)}(\varphi)\right)\odot \left(\sum_{k\in\mathbb{N}}\frac{1}{s!}p^{\odot s}\circ \Delta_{S\Gamma \omega SJ\Phi}^{(s)}(\psi)\right)
\end{equation*}
\begin{equation*}
=\left( p\circ Id_{S\Gamma \omega SJ\Phi}(\varphi)+\frac{1}{2}p^{\odot 2}\circ\Delta_{S\Gamma \omega SJ\Phi}^{(2)}(\varphi)+\frac{1}{3!}p^{\odot 3}\circ \Delta_{S\Gamma \omega SJ\Phi}^{(3)}(\varphi)+\cdots \right)
\end{equation*}
\begin{equation*}
\odot \left( p\circ Id_{S\Gamma \omega SJ\Phi}(\psi)+\frac{1}{2}p^{\odot 2}\circ\Delta_{S\Gamma \omega SJ\Phi}^{(2)}(\psi)+\frac{1}{3!}p^{\odot 3}\circ \Delta_{S\Gamma \omega SJ\Phi}^{(3)}(\psi)+\cdots \right)
\end{equation*}
\begin{align*}
&=p(\varphi)\odot p(\psi)+\frac{1}{2}p(\varphi)\odot p^{\odot 2}\circ \Delta_{S\Gamma \omega SJ\Phi}^{(2)}(\psi) +\frac{1}{2}p^{\odot 2}\circ\Delta_{S\Gamma \omega SJ\Phi}^{2}(\varphi)\odot p(\psi)\\
& +\frac{1}{3!}p(\varphi)\odot p^{\odot 3}\circ \Delta_{S\Gamma \omega SJ\Phi}^{(3)}(\psi) 
+\frac{1}{2}p^{\odot 2}\circ \Delta_{S\Gamma \omega SJ\Phi}^{(2)}(\varphi) \odot \frac{1}{2} p^{\odot 2}\circ \Delta_{S\Gamma \omega SJ\Phi}^{(2)}(\psi)\\
& +\frac{1}{3!}p^{\odot 3}\circ \Delta_{S\Gamma \omega SJ\Phi}^{(3)}(\varphi)\odot p(\psi)+\cdots ,
\end{align*}
which is equal to  $P'(\varphi\cdot\psi)$ by  equation \eqref{producto}. Notice that we used the fact that $\Delta_{S\Gamma \omega SJ\Phi}=Id_{S\Gamma SJ\Phi}$ (see the definition of $\Delta^{(n)}$ in Section \ref{sec111}).
\end{proof}
 
 \par As an observation we can say that renormalizations do not necessarily preserve the product of elements of $S\Gamma_{c}\omega SJ\Phi$ if they have not disjoint support.

\section{Renormalizations $\curvearrowright$ Feynman measures}
 
 
 \par Thanks to the Lemma \ref{lema7},  given a renormalization $P$ we can think it as an automorphism of $S\Gamma \omega SJ\Phi$ which preserves the coaction of $\Gamma SJ\Phi$. A Feynman measure (see Definition \ref{dia}) is a linear map from $S\Gamma_{c} \omega SJ\Phi$ to $\mathbb{C}$; then we can define the action from the renormalization group over the Feynman measures as follows. If  $P$ is a renormalization and $\delta$ is a Feynman measure, we can lift $P$ by using Lemma \ref{lema7} that is still called $P$. Then a new Feynman measure $P(\delta)$ is define by $P(\delta)(A):=\delta(P^{-1}(A))$ if $A\in S\Gamma_{c}\omega SJ\Phi$, . 
 \par It will be interesting to prove that $P(\delta)$ is still a Feynman measure associated with the propagator $\Delta$ if $\delta$ is associated with $\Delta$ too. First of all we will prove the smoothness on the diagonal. 
\par Suppose  $\delta$ is a Feynman measure and let $P$ the lift of some renormalization, and take $A\in S^{n}\Gamma_{c}\omega SJ\Phi$. $P^{-1}$ is an element of the renormalization group that preserves the filtration, so $P^{-1}(A)$ is a polynomial in monomials of degree at least $n$ , say $B_{1}+\cdots+B_{n}$. Given $(p_{1},\cdots,p_{k})\in \sum_{(q,\cdots,q)}\chi_{k}(\delta,B_{k})$ one can affirm that $p_{1}+\cdots +p_{k}=0$ because $\delta$ is smooth in the diagonal, as a conclusion \textbf{$P\delta$ is smooth in the diagonal}. \\

\par To see that $P\delta$ is non-degenerate we can use Theorem \ref{teorema1} to establish that if $P$ is a renormalization  $P|_{\Gamma_{c}\omega}$ is an isomorphism. But the isomorphisms of sheafs (induced by vector bundles) are in correspondence whith the isomorphisms of vector bundles supported by the identity over $\mathcal{M}$ . Given two vector bundles, an isomorphism between them must be linear on each fiber, but in our case the bundle has range $1$, because it is the density bundle. Hence the isomorphism  is given (on the fibers) by multiplication by a constant $\lambda\neq0$, so the bundle morphism is given globally by the multiplication by a nowhere-vanishing function $\lambda\in \mathcal{C}^{\infty}(\mathcal{M})$.
\par Applying this to $P^{-1}$, given $v\in \Gamma_{c} \omega$, we obtain,
$$ P\delta (v)=\delta(P^{-1}(v))=\int_{\mathcal{M}}gP^{-1}(v),$$
where $g$ is a nowhere-vanishing function that exists because $\delta$ is non-degenerate. Moreover,

$$ P\delta (v)=\int_{\mathcal{M}}(g\lambda)(v),$$ 
where $\lambda$ is the nowhere-vanishing function described above. Then \textbf{$P\delta$ is non-degenerate}.\\

\par Let $\delta$ be a Feynman measure, $Q$ an element of the renormalization group (whose inverse is $P$) and let $A\in S^{m}\Gamma_{c} \omega SJ\Phi$ and $B\in S^{n}\Gamma_{c} \omega SJ\Phi$  such that there is no coordinate of any  element in $Sup(A)$ which is $\leqslant$ to some coordinate of any element in $Sup(B)$, this particularly implies that $Sup(A)\cap Sup(B)=\emptyset$, and then by Lemma \ref{lema7}, $P$ preserves the product. Then,
\begin{equation*}
Q\delta(A\cdot B)=\delta(P(A\cdot B))=\delta(P(A)\cdot P(B))
\end{equation*} 
 and as $\Delta$ is cut the products involved in the Gaussian condition are well defined (see Observation \ref{obs4}), so
 
 \begin{equation}\label{345}
 Q\delta(A\cdot B)=\sum (\chi(\delta,(P(A))')\otimes \chi(\delta,(P(B))'))\cdot \hat{\Delta}((P(A))'',(P(B))'').
 \end{equation} 
 
 \par We will now compute $(P(A))'$, $(P(A))''$, $(P(B))''$ and $(P(B))'$. 
  Consider
  \begin{equation*}
      \xymatrix{
       S\Gamma \omega SJ\Phi \ar[rr]^{\sigma} \ar[d]_{P} & &    S\Gamma \omega SJ\Phi\otimes_{\mathcal{C}^{\infty}(\mathcal{M})} S\Gamma SJ\Phi \ar[d]^{P\otimes Id} \\
     S\Gamma \omega SJ\Phi \ar[rr]^{\sigma} &  &  S\Gamma \omega SJ\Phi\otimes_{\mathcal{C}^{\infty}(\mathcal{M})} S\Gamma SJ\Phi  \\}
     \end{equation*} 
   which is commutative.
  \par Applying this maps to the element $A$, and remembering that $\sigma(A)=\sum A'\otimes A''$ we have
  $$ A''=(P(A))'' $$ and the same for $B$. Substituting  the results in equation \eqref{345}, we obtain 
   
 \begin{equation*}
   Q\delta(A\cdot B)=\sum (\chi(\delta,(P(A))')\otimes \chi(\delta,(P(B))'))\cdot \hat{\Delta}(A'',B'').
 \end{equation*}

   \par The commutative diagram tells us that $(P(A))'=P(A')$, using this we have
   
\begin{align*}
  & Q\delta(A\cdot B)=\sum (\chi(\delta,P(A'))\otimes \chi(\delta,P(B')))\cdot \hat{\Delta}(A'',B'')\\
  & =\sum (\chi(Q\delta,A')\otimes \chi(Q\delta,B'))\cdot \hat{\Delta}(A'',B''),
\end{align*}
 which is exactly the Gaussian condition for $Q\delta$. As a conclusion we can say that if $P$ is a renormalization  and $\delta$ a Feynman measure then $P\delta$ is a Feynman measure too. So we have an action of the renormalization group over the set of Feynman measures associated to a fixed cut propagator.\\

  \par We will need the following result.
  
  {\lemma\label{lemacociente}{ Given $k\in \mathbb{N}$, if $\lambda:S^{k}\Gamma_{c}\omega SJ\Phi\rightarrow \mathbb{C} $ is a continuous function satisfying $Sup(\lambda)\subseteq \text{Diag}(\mathcal{M}^{k})$ then  $\lambda$ is $\mathcal{C}^{\infty}(\mathcal{M})$-balanced.}}
\begin{proof}
 Let $f$ be a smooth function over $\mathcal{M}$. To prove that\\ $\lambda(A_{1}\cdots fA_{i_{0}}\cdots A_{j_{0}}\cdots A_{k})=\lambda(A_{1}\cdots A_{i_{0}}\cdots fA_{j_{0}}\cdots A_{k})$ ($A_{i}$'s are Lagrangian densities of compact support) as complex numbers,  we use Remark \ref{obs5} and reduce it to prove the equality in $D_{d}(\mathcal{M}^{k})$. 
\par Let $g_{1}\otimes \cdots \otimes g_{k}$ be an element of $\mathcal{C}^{\infty}(\mathcal{M})^{\otimes k}$ ( recall these element are dense in  $\mathcal{C}^{\infty}(\mathcal{M}^{k})$, and $\lambda$ is continuous ) whose support is included in $Diag(\mathcal{M}^{k})$, then we must compare  $\chi_{k}(\lambda,A_{1}\cdots fA_{i_{0}}\cdots A_{j_{0}}\cdots A_{k})(g_{1}\otimes \cdots \otimes g_{k})$ and\\ $\chi_{k}(\lambda,A_{1}\cdots A_{i_{0}}\cdots fA_{j_{0}}\cdots A_{k})(g_{1}\otimes \cdots \otimes g_{k})$, so if they are equal then $\lambda(A_{1}\cdots fA_{i_{0}}\cdots A_{j_{0}}\cdots A_{k})$ and $\lambda(A_{1}\cdots A_{i_{0}}\cdots fA_{j_{0}}\cdots A_{k})$ will also be equal. But, 
$$\chi_{k}(\lambda,A_{1}\cdots fA_{i_{0}}\cdots A_{j_{0}}\cdots A_{k})(g_{1}\otimes \cdots \otimes g_{k})$$
$$=\lambda(g_{1}A_{1}\cdots g_{i_{0}}fA_{i_{0}}\cdots g_{j_{0}}A_{j_{0}}\cdots g_{k}A_{k})$$
\begin{equation}\label{3.46}
=\chi_{k}(\lambda,A_{1}\cdots A_{i_{0}}\cdots A_{j_{0}}\cdots A_{k})(g_{1}\otimes \cdots \otimes  fg_{i_{0}}\otimes\cdots \otimes g_{k})
\end{equation}
using that $g_{1}\otimes \cdots \otimes  fg_{i_{0}}\otimes\cdots \otimes g_{k}=g_{1}\otimes \cdots \otimes  fg_{j_{0}}\otimes \cdots \otimes g_{k}$ (that is because evaluating the left hand side in a point $(x_{1},\cdots,x_{k})\in \mathcal{M}^{k}$ we obtain $g_{1}(x_{1}) \cdots  f(x_{i_{0}})g_{i_{0}}(x_{i_{0}})\cdots g_{k}(x_{k})$, which is zero if there are two points $x_{i}\neq x_{j}$. But in the relevant case $(x_{1},\cdots,x_{k})\in \text{Diag}(\mathcal{M}^{k})$, that is  $(x,\cdots,x)$ we obtain  $g_{1}(x) \cdots   f(x)g_{i_{0}}(x)\cdots  g_{k}(x)=g_{1}(x) \cdots   g_{i_{0}}(x)  \cdots f(x)g_{j_{0}}(x) \cdots g_{k}(x)$ where we can move $f(x)$ because it is a product on $\mathbb{C}$) we conclude  that the last term in Equation \eqref{3.46} is equal to 
$$\chi_{k}(\lambda,A_{1}\cdots A_{i_{0}}\cdots A_{j_{0}}\cdots A_{k})(g_{1}\otimes \cdots \otimes  fg_{j_{0}}\otimes\cdots \otimes g_{k})$$
$$ \chi_{k}(\lambda,A_{1}\cdots fA_{j_{0}}\cdots A_{j_{0}}\cdots A_{k})(g_{1}\otimes \cdots \otimes g_{k}). $$  



\par Then from Remark \ref{obs5} we conclude that $\chi_{k}(\lambda,A_{1}\cdots fA_{i_{0}}\cdots A_{j_{0}}\cdots A_{k})=\chi_{k}(\lambda,A_{1}\cdots A_{i_{0}}\cdots fA_{j_{0}}\cdots A_{k})$ as element of  $D_{d}(\mathcal{M}^{k})$ for all $k\in \mathbb{N}$, which implies that $\lambda(A_{1}\cdots fA_{i_{0}}\cdots A_{j_{0}}\cdots A_{k})=\lambda(A_{1}\cdots A_{i_{0}}\cdots fA_{j_{0}}\cdots A_{k})$ as complex numbers. 
\end{proof}

\par To prove the next theorem we will use a well-known result,
\begin{lemma}[See \cite{ro}]\label{rot}
If $E\rightarrow \mathcal{M}$ is a  finite rank vector bundle over a manifold, then $\Gamma_{c}(E)$ is a projective $\mathcal{C}^{\infty}(\mathcal{M})$-module. 
\end{lemma} 

\par Consider the vector bundle $\omega S^{k}J\Phi$,  which will be of finite rank if and only we take finite order jets $J\Phi$. Then from now on we must understand $J\Phi$ as the vector bundle of finite jets up to some fixed order, i.e. $J\Phi=J^{k}\Phi$ for some $k\in \mathbb{N}$.
\par Another result we will use is,

\begin{lemma}\label{simet}
If $A$ and $B$ are projective modules over a commutative ring $R$, then the symmetric tensor product $(A\otimes_{R} B)/I$ is also a projective $R$-module ($I$ is the ideal of the vectors which give symmetry).
\end{lemma}     
\begin{proof}
Can be found in many books of elementary algebra as \cite{wei}, and it is deduced from the fact that the tensor product commutes with the direct sum. 
\end{proof}

\par With all these results we are ready to establish and prove the following theorem which gives us an important property of this action.

 \begin{thm}The group of renormalizations acts transitively on the set of Feynman measures associated with a given cut local propagator. 
 \end{thm}
\begin{proof}
 We already proved that the action is well defined in the sense that the result is also a Feynman measure associated to the same given cut local propagator. 
\par To finish the proof we must prove the transitivity. Given $\delta$ and $\delta'$ Feynman measures associated with the same cut local propagator,  we want to prove the existence of a $g\in G$ ($g:S\Gamma\omega SJ\Phi\to S\Gamma\omega SJ\Phi$ ) such that $g\delta=\delta'$, expressing $g$ as $\cdots g_{3}\cdot g_{2}\cdot g_{1}$ (see Corollary \ref{hhhhhh}). We proceed by induction, on the degree of the elements of $S\Gamma\omega SJ\Phi$ i.e. $S^{n}\Gamma\omega S^{k}J\Phi$ for $(n,k)\in \mathbb{N}\times \mathbb{N}_{0}$ in the lexicographic order.
\par We begin proving the existence of $g_{1}:S\Gamma \omega S J\Phi\to S\Gamma\omega SJ\Phi$ whose sequential representation is  $\{g_{1}\}_{n\in \mathbb{N}}=\{ f_{1},0,\cdots \}$ where $f_{1}:\Gamma\omega SJ\Phi\to \Gamma\omega SJ\Phi\simeq \Gamma_{c}\omega \otimes S_{\mathcal{C}^{\infty}(\mathcal{M})}\Gamma J\Phi$ (where we use Theorems \ref{df} and \ref{dm}). Hence by Proposition \ref{prop5}, there exits a map $h_{1}:\Gamma\omega SJ\Phi \to \Gamma_{c}\omega$ such that $f_{1}=(h_{1}\otimes Id_{\Gamma SJ\Phi})\circ(Id_{\Gamma_{c}\omega}\otimes \Delta_{\Gamma SJ\Phi})$. 
\par Taking $n=1$ and $k=0$ in $S^{n}\Gamma\omega S^{k}J\Phi$, we want to define $f_{1}$ over the elements $\mathcal{L}=\alpha\otimes l$ where $\alpha\in \Gamma_{c}\omega$ and $l=1 \in S^{0}\Gamma J\Phi$. We do that by means of $h_{1}$, 
\begin{align*}
& g_{1}(\mathcal{L})=f_{1}(\mathcal{L})\\
& =f_{1}(\alpha\otimes l)=(h_{1}\otimes Id_{\Gamma SJ\Phi})\circ(Id_{\Gamma_{c}\omega}\otimes \Delta_{\Gamma SJ\Phi})(\alpha\otimes 1) \\
& =(h_{1}\otimes Id_{\Gamma\omega SJ\Phi})\circ(\alpha \otimes (1\otimes 1)) \\
&  =(h_{1}\otimes Id_{\Gamma\omega SJ\Phi})\circ((\alpha \otimes 1)\otimes 1)  \\
& =h_{1}(\alpha\otimes 1)\otimes Id_{\Gamma SJ\Phi}(1)= h_{1}(\mathcal{L})\otimes 1 =h_{1}(\mathcal{L}).
\end{align*}
Applying $\delta'$ we have
\begin{align*}
& \delta'(g_{1}(\mathcal{L}))=\delta'(f_{1}(\mathcal{L}))=\delta'(h_{1}(\mathcal{L})).
\end{align*}
We want  $h_{1}$ to satisfy $\delta'\circ h_{1}=\delta$ over the elements with degree $n=1$ and $k=0$. In the case where $\Gamma\omega S^{0} J\Phi$ is a free $\mathcal{C}^{\infty}(\mathcal{M})$-module so it is easy to define $h_{1}|_{\Gamma\omega S^{0} J\Phi}$ over a basis such that $h_{1}|_{\Gamma\omega S^{0} J\Phi}$ satisfies $\delta'\circ h_{1}|_{\Gamma\omega S^{0} J\Phi}=\delta$ because $\delta'$ is epimorphism.\\
 In the general case were $\Gamma\omega S^{0} J\Phi$ (is not necessarily a free $\mathcal{C}^{\infty}(\mathcal{M})$-module) is a  projective $\mathcal{C}^{\infty}(\mathcal{M})$-module, then it is a direct summand of a free $\mathcal{C}^{\infty}(\mathcal{M})$-module (call it $F$), so we apply the previous argument to define $\tilde{h}_{1}$ over $F$ and then by restricting $\tilde{h}_{1}$ to the direct summand we are interested to have $h_{1}|_{\Gamma\omega S^{0} J\Phi}$ satisfying  $\delta'\circ h_{1}|_{\Gamma\omega S^{0} J\Phi}=\delta$. Hence the diagram,
\begin{equation}
   \xymatrix{
      & & &    \Gamma_{c}\omega \ar[d]^{\delta'}  \\
    S^{1}_{\mathcal{C}^{\infty}(\mathcal{M})}\Gamma\omega S^{0}J\Phi \ar[rrr]_{\delta} \ar@{-->}[rrru]^{h_{1}|_{ S^{1}_{\mathcal{C}^{\infty}(\mathcal{M})}\Gamma\omega S^{0}J\Phi}} & &  &  \mathbb{C}  \\}
 \end{equation} 
can be completed with a map $h_{1}|_{ S^{1}_{\mathcal{C}^{\infty}(\mathcal{M})}\Gamma\omega S^{0}J\Phi}$ or which is the same $h_{1}|_{\Gamma\omega S^{0} J\Phi}$. 
\par But $g_{1}\delta(\mathcal{L})=\delta'(\mathcal{L})$ if and only if $\delta(\mathcal{L})=\delta'(g_{1}(\mathcal{L}))$, and this is exactly what we have with the aforementioned election of $h_{1}|_{\omega S^{0}J\Phi}$. 
\par Continuing this process by defining $h_{1}$ over elements of $S^{1}_{\mathcal{C}^{\infty}(\mathcal{M})}\Gamma\omega S^{1}J\Phi$ (i.e. $n=1$ and $k=1$), we consider the Lagrangian density $\mathcal{L}=\alpha \otimes l$ where $\alpha \in \Gamma_{c}\omega$ and $l=\varphi \in S^{1}\Gamma J\Phi$ and compute 
\begin{align*}
& g_{1}(\mathcal{L})=f_{1}(\mathcal{L})\\
& =f_{1}(\alpha\otimes l)=(h_{1}\otimes Id_{\Gamma SJ\Phi})\circ(Id_{\Gamma_{c}\omega}\otimes \Delta_{\Gamma SJ\Phi})(\alpha\otimes \varphi) \\
& =(h_{1}\otimes Id_{\Gamma SJ\Phi})\circ(\alpha \otimes [1\otimes \varphi +\varphi \otimes 1]) \\
&  =(h_{1}\otimes Id_{\Gamma SJ\Phi})((\alpha\otimes 1)\otimes \varphi+(\alpha\otimes \varphi)\otimes 1) \\
& =h_{1}(\alpha\otimes 1)\otimes Id_{\Gamma SJ\Phi}(\varphi)+h_{1}(\alpha\otimes \varphi)\otimes Id_{\Gamma\omega SJ\Phi}(1) \\
&= h_{1}(\alpha\otimes 1)\otimes \varphi+h_{1}(\mathcal{L}).
\end{align*}
\par Applying $\delta'$ to both sides of this equation we have, 
\begin{equation*} \delta'(g_{1}(\mathcal{L})=\delta'(h_{1}(\alpha\otimes 1)\otimes \varphi)+\delta'\circ h_{1}(\mathcal{L}).
\end{equation*}
But as we want $\delta'\circ g_{1}=\delta$ then we are looking for a  $h_{1}|_{S^{1}_{\mathcal{C}^{\infty}(\mathcal{M})}\Gamma\omega S^{1}J\Phi}$ such that the equation
\begin{align*} 
& \delta'(h_{1}(\alpha)\otimes \varphi)+\delta'\circ h_{1}(\mathcal{L})=\delta(\mathcal{L}), \\
& \text{or } \delta'\circ h_{1}(\mathcal{L})=(\delta-\delta'\circ (h_{1}|_{S^{1}_{\mathcal{C}^{\infty}(\mathcal{M})}\Gamma\omega S^{0}J\Phi}\otimes Id_{\Gamma SJ\Phi}))(\mathcal{L})
\end{align*}
holds. Take a $h_{1}|_{S^{1}_{\mathcal{C}^{\infty}(\mathcal{M})}\Gamma\omega S^{1}J\Phi}$ such that the next diagram 
\begin{equation}
   \xymatrix{
      & & & & & & &  \Gamma_{c}\omega \ar[d]^{\delta'}  \\
    S^{1}_{\mathcal{C}^{\infty}(\mathcal{M})}\Gamma\omega S^{1}J\Phi \ar[rrrrrrr]_{\delta-\delta'\circ (h_{1}|_{S^{1}_{\mathcal{C}^{\infty}(\mathcal{M})}\Gamma\omega S^{1}J\Phi}\otimes Id_{\Gamma SJ\Phi})} \ar@{-->}[rrrrrrru]^{h_{1}|_{ S^{1}_{\mathcal{C}^{\infty}(\mathcal{M})}\Gamma\omega S^{0}J\Phi}} & & & & & &  &  \mathbb{C}  \\}
 \end{equation}
commutes. The existence of this  $h_{1}|_{S^{1}_{\mathcal{C}^{\infty}(\mathcal{M})}\Gamma\omega S^{1}J\Phi}$ follows from the same arguments given for $h_{1}|_{S^{1}_{\mathcal{C}^{\infty}(\mathcal{M})}\Gamma\omega S^{0}J\Phi}$.
\par Continuing this process we can recursively define $h_{1}|_{S^{1}_{\mathcal{C}^{\infty}(\mathcal{M})}\Gamma\omega S^{k}J\Phi}$ for all $k\in \mathbb{N}_{0}$ and so we have just defined $h_{1}:S^{1}\Gamma\omega SJ \Phi\to \Gamma_{c}\omega$. Proceeding by induction we can define $h_{n}:S^{n}\Gamma\omega SJ \Phi\to \Gamma_{c}\omega$ for each $n\in \mathbb{N}$ (using the Lemma \ref{rot} and Lemma \ref{simet} to affirm that $\Gamma\omega SJ \Phi$ is a projective module and so is $S\Gamma\omega SJ \Phi$. And consequently a is direct addend of a free $\mathcal{C}^{\infty}(\mathcal{M})$-module), each one of is in correspondence with a $g_{n}:S\Gamma\omega SJ \Phi\to S\Gamma\omega SJ \Phi$ whose sequential representation is $\{ 0,\cdots,0,f_{n},0,\cdots\}$. 

The maps $g_{n}$ define by Corollary \ref{hhhhhh} a renormalization $g=...g_{3}g_{2}g_{1}$, such that $\delta(\mathcal{L})=\delta'(g(\mathcal{L}))$ for all $\mathcal{L}\in S_{\mathcal{C}^{\infty}(\mathcal{M})}\Gamma\omega SJ \Phi$. That is $g(\delta)=\delta'$.  
\end{proof}

\end{document}